\documentclass[12pt]{amsart}
\usepackage{amsfonts}
\usepackage{natbib}
\usepackage{eurosym}
\usepackage{caption}
\usepackage{amssymb, amsmath, url,color,  amsthm, amssymb,graphicx,multirow, appendix}
\usepackage[margin=1in]{geometry}
\usepackage{bigints}
\usepackage{lscape,amssymb, amsmath,verbatim,graphicx}
\usepackage{epsfig,subfigure,float,subfigure}
\usepackage{graphicx}
\usepackage{epsfig}
\usepackage{threeparttable}
\usepackage{ragged2e}
\usepackage{epstopdf}
\usepackage{booktabs}
\usepackage{float}
\usepackage{multirow}
\usepackage{adjustbox}
\usepackage{lscape,lipsum}
\usepackage{url}
\usepackage{natbib}
\usepackage[T1]{fontenc}
\usepackage{a4wide, url}
\usepackage[english]{babel}
\usepackage{graphicx}
\usepackage{array}
\usepackage{multirow}
\usepackage{amsthm}
\usepackage{amsmath}
\usepackage[foot]{amsaddr}
\usepackage{amsfonts}
\usepackage{amssymb}
\usepackage{pgf}
\usepackage{paralist}
\usepackage{pdflscape}
\usepackage{lscape,amssymb, amsmath,verbatim,graphicx}
\usepackage{epsfig,subfigure,float,subfigure}
\usepackage{graphicx}
\usepackage{epsfig}
\usepackage{amssymb}
\usepackage{amsmath}
\usepackage{amsfonts}
\usepackage{geometry}
\usepackage{scalefnt}
\usepackage{graphicx}
\usepackage{setspace}
\usepackage{array}
\usepackage{multirow}
\usepackage{amsthm}
\usepackage{amsmath}
\usepackage{amsfonts}
\usepackage{amssymb}
\usepackage{pgf}
\usepackage{paralist}
\usepackage{pdflscape}
\usepackage{rotating}
\usepackage{scalefnt}
\usepackage{natbib}
\usepackage{graphicx}
\usepackage{array}
\usepackage{multirow}
\usepackage{amsthm}
\usepackage{amsmath}
\usepackage{amsfonts}
\usepackage{amssymb}
\usepackage{pgf}
\usepackage{bm}
\usepackage{paralist}
\usepackage{pdflscape}
\usepackage{rotating}
\usepackage{scalefnt}
\usepackage{xr}
\usepackage{verbatim}
\usepackage{changepage}
\usepackage{epsfig,subfigure,float,subfigure}
\usepackage[font={scriptsize}]{caption}

\RequirePackage[OT1]{fontenc}
\RequirePackage{amsthm,amsmath}
\RequirePackage[colorlinks,citecolor=blue,urlcolor=blue]{hyperref}
\numberwithin{equation}{section}
\theoremstyle{plain}
\newtheorem{theorem}{Theorem}
\newtheorem{proposition}{Proposition}
\newtheorem{lemma}{Lemma}
\newtheorem{assumption}{Assumption}
\newtheorem{corollary}{Corollary}
\newcommand{\mbf}{\mathbf}
\newcommand{\beq}{\begin{equation}}
\newcommand{\eeq}{\end{equation}}
\newcommand{\bea}{\begin{eqnarray}}
\newcommand{\eea}{\end{eqnarray}}
\newcommand{\bit}{\begin{itemize}}
\newcommand{\eit}{\end{itemize}}
\newcommand{\ben}{\begin{enumerate}}
\newcommand{\een}{\end{enumerate}}
\newcommand{\bpm}{\begin{pmatrix}}
\newcommand{\epm}{\end{pmatrix}}
\newcommand{\bbm}{\begin{bmatrix}}
\newcommand{\ebm}{\end{bmatrix}}
\newcommand{\eps}{\varepsilon}
\renewcommand{\l}{\left}
\renewcommand{\r}{\right}
\newcommand{\E}[0]{\mathsf{E}}
\newcommand{\Var}[0]{\mathsf{Var}}
\newcommand{\Cov}[0]{\mathsf{Cov}}
\newcommand{\Cor}[0]{\mathsf{Cor}}
\newcommand{\ra}{\rightarrow}
\newcommand{\nn}{\nonumber}
\newcommand{\wh}{\widehat}
\newcommand{\wt}{\widetilde}
\newcommand{\wb}{\overline}
\newcommand\matt[1]{\textcolor{red}{#1}}

\begin{document}

\title{Inference in heavy-tailed non-stationary multivariate time series}
\author{Matteo Barigozzi}
\address{Matteo Barigozzi, University of Bologna, matteo.barigozzi@unibo.it}
\author{Giuseppe Cavaliere}
\address{Giuseppe Cavaliere, University of Bologna, giuseppe.cavaliere@unibo.it}
\author{Lorenzo Trapani}
\address{Lorenzo Trapani, University of Nottingham, lorenzo.trapani@nottingham.ac.uk}

\begin{abstract}

We study inference on the common stochastic trends in a non-stationary, $N$-variate time series $y_{t}$, in the possible presence of heavy tails. We propose a novel methodology which does not require any knowledge or estimation of the tail index, or even knowledge as to whether certain moments (such as the variance) exist or not, and develop an estimator of the number of stochastic trends $m$ based on the eigenvalues of the sample second moment matrix of $y_{t}$. We study the rates of such eigenvalues, showing that the first $m$ ones diverge, as the sample size $T$ passes to infinity, at a rate faster by $O\left(T \right)$ than the remaining $N-m$ ones, irrespective of the tail index. We thus exploit this eigen-gap by constructing, for each eigenvalue, a test statistic which diverges to positive infinity or drifts to zero according to whether the relevant eigenvalue belongs to the set of the first $m$ eigenvalues or not. We then construct a randomised statistic based on this, using it as part of a sequential testing procedure, ensuring consistency of the resulting estimator of $m$. We also discuss an estimator of the common trends based on principal components and show that, up to a an invertible linear transformation, such estimator is consistent in the sense that the estimation error is of smaller order than the trend itself. Finally, we also consider the case in which we relax the standard assumption of \textit{i.i.d.} innovations, by allowing for heterogeneity of a very general form in the scale of the innovations. A Monte Carlo study shows that the proposed estimator for $m$ performs particularly well, even in samples of small size. We complete the paper by presenting four illustrative applications covering commodity prices, interest rates data, long run PPP and cryptocurrency markets.

\end{abstract}

\maketitle

\doublespacing

\section{Introduction}

Since the seminal works by \citet{englegranger}, \citet{SW1988} and %
\citet{johansen1991}, investigating the presence and number $m$ of common
stochastic trends has become an essential step in the analysis of
multivariate time series which are non-stationary over time. Inference on $m$
is of great importance on its own, as it has a \textquotedblleft
natural\textquotedblright\ interpretation in many applications: for example,
it can provide the number of non-stationary factors in Nelson-Siegel type
term structure models; it allows to assess the presence of (long-run)
integration among financial market (\citealp{kasa}), or, in the context of
asset pricing, the number of portfolios that allow hedging of long-run risks
(\citealp{alexander}). In terms of theory, available estimators are based
either on sequential testing (see e.g. \citealp{boswijk2015improved}), or on
information criteria (see \citealp{aznar} and \citealp{qu2007modified}).

With few exceptions, a standard feature of existing techniques is that they
strongly rely on the assumption that some moments of the data (usually at
least the second, or even the fourth) exist. This assumption, however, often
lacks of empirical support. In fact, data exhibiting heavy tails, which do
not have finite second (or even first) moment, are often encountered in many
areas: macroeconomics (see e.g. \citealp{ibragimov}), finance (%
\citealp{rachev2003handbook}), urban studies (\citealp{gabaix}); and also
insurance, actuarial science, telecommunication network traffic and
meteorology (see e.g. the review by \citealp{embrechts}).

Violation of the moment assumptions may result in (possibly severe)
incorrect determination of the number of common trends - see e.g. the
simulations in \citet{caner} and the empirical evidence in \citet{falk}.
Unfortunately, contributions which explicitly deal with inference on common
stochastic trends under infinite variance are rare. \citet{caner} derives
the asymptotic distribution of Johansen's trace test under infinite variance
and shows that it depends on the (unknown) tail index of the data. %
\citet{ling} study the (non-standard) rate of convergence of estimators in
non-stationary VAR models and show that the limiting distributions depend on
the tail index in a non-trivial fashion; similar results are also found in %
\citet{paulauskas} and \citet{fasen}. In the univariate case, %
\citet{jach2004subsampling} and \citet{cavaliere2018unit} show that suitable
(respectively, $M$-out-of-$N$ and wild) bootstrap approaches could be used
to test whether data are driven by a stochastic trend; knowledge of the tail
index is not needed, but extensions of these bootstrap approaches to
multiple time series are not available, and are likely to be extremely
difficult to develop and implement. Distribution-free approaches could also
help overcome this difficulty. \citet{hallin2} (see also \citealp{hallin1})
apply the rank transformation to the VECM residuals, obtaining nuisance-free
statistics, although this approach requires the correct specification of the
VECM.

\medskip

\textit{Contribution}

\medskip

Our key contribution is the estimator of the number $m$ of common trends of
an $N$-variate time series in the possible presence of heavy tails.
Crucially, our procedure does not require any \textit{a priori} knowledge as
to whether the variance is finite or not, or as to how many moments exist,
thus avoiding having to estimate any nuisance parameters or even pre-testing
for (in)finite moments. We also show that, in contrast to most of the
literature on time series with heavy tails, our methodology also applies to
time series with \textit{heterogeneous} innovations. Specifically, we allow
for changes in the scale of the innovations of a very general form, which
covers, e.g., multiple shifts and smooth scale changes. As far as we are
aware, this paper is the first one where heterogeneity in the scale is
allowed under infinite variance.

Details are spelt out in the remainder of the paper; here, we offer a
heuristic preview of how the methodology works. The starting point of our
analysis is a novel result concerning the properties of the sample second
moment matrix of the data in levels. Specifically, we show that the $m $
largest eigenvalues of the matrix diverge to positive infinity, as the
sample size $T $ passes to infinity, faster than the remaining eigenvalues
by a factor (almost) equal to $T$. Importantly, this result always holds,
irrespective of the variance of the innovations being finite or infinite.
Building on this, for each eigenvalue we construct a statistic which
diverges to infinity under the null that the eigenvalue is diverging at a
\textquotedblleft fast\textquotedblright\ rate, and drifts to zero under the
alternative that the relevant eigenvalue diverges at a \textquotedblleft
slow\textquotedblright\ rate. Although the limiting distribution of our
statistic is bound to depend on nuisance parameters such as the tail index,
the relative rate of divergence between the null and the alternative does
not depend on any nuisance parameters. Therefore, in order to construct a
test, we can rely on rates only, and randomise our statistic using a similar
approach to \citet{bandi2014}. Thence, our estimator of $m$ is based on
running the tests sequentially; in this respect, it mimics the well-known
sequential procedure advocated in \citet{johansen1991} for the determination
of the rank of a cointegrated system. Our approach has some commonalities
with the literature on factor models too, where the spectrum of the
covariance matrix of the data is employed to estimate the number of common
factors. \citet{zyr18} propose a methodology to determine the rank of
cointegration in a possibly fractionally integrated system with finite
second moments, based on the eigenvalues of (functions of) the sample
autocovariance matrices (see also \citealp{tu2020error}).

Our methodology has at least four desirable features. Firstly, as mentioned
above, our technique can be applied to data with infinite variance (and even
infinite expectation), with no need to know this \textit{a priori}.
Secondly, our procedure does not require at any stage the estimation of the
tail index of a distribution, which is a notoriously delicate issue.
Thirdly, our procedure does not require the correct specification of the lag
structure of the underlying VECM model, and it is therefore completely
robust to misspecification of the dynamics. Finally, our results are based
only on first order asymptotics, i.e. rates, which makes our procedure
extremely easy to implement in practice.

All these advantages are illustrated via Monte Carlo simulation, where we
show that even in samples of moderate size our method performs extremely
well, and through four empirical applications, which cover commodity prices,
interest rates data, purchasing power parity in the long run, and
cryptocurrency markets.

\medskip

\textit{Structure of the paper and notation}

\medskip

The remainder of the paper is organised as follows. Assumptions and
asymptotics are in Section \ref{theory}. The main results on the number of
common trends and the estimation of the common trends (and associated
loadings) are presented in Section \ref{inftrends}, and extensions are
collected in Section \ref{extensions}. We provide Monte Carlo evidence in
Section \ref{montecarlo}, and we validate our methodology through four real
data applications in Section \ref{empirics}. Section \ref{conclusions}
concludes. Further simulations, and all technical lemmas and proofs, are
relegated to the Supplement.

Throughout the paper, we make use of the following notation. For a given
matrix $A\in 
\mathbb{R}
^{n\times m}$, we denote its element in position $\left( i,j\right) $ as $%
A_{i,j}$; we also let let $\lambda ^{\left( j\right) }\left( A\right) $
denote the $j$-th largest eigenvalue of $A$. We denote with $c_{0},c_{1},...$
positive, finite constants whose value can change from line to line. We
denote the $k$-iterated logarithm of $x$ (truncated at zero) as $\ln _{k}x$
- e.g. $\ln _{2}x=\max \left\{ \ln \ln x,0\right\} $. The backshift operator
is denoted as $L$.

\section{Theory\label{theory}}

Consider an $N$-dimensional vector $y_{t}$ with $MA\left( \infty \right) $
representation%
\begin{equation}
\Delta y_{t}=C\left( L\right) \varepsilon _{t},  \label{ma}
\end{equation}%
where $C\left( L\right) =\sum_{j=0}^{\infty }C_{j}L^{j}$. 
We assume that $y_{0}=0$ and $\varepsilon _{t}=0$ for $t\leq 0$, for
simplicity and with no loss of generality; inspecting the proofs, all the
results derived here can be extended to more general assumptions concerning
e.g. the initial value $y_{0}$. Standard arguments allow to represent (\ref%
{ma}) as%
\begin{equation}
y_{t}=C\sum_{s=1}^{t}\varepsilon _{s}+C^{\ast }\left( L\right) \varepsilon
_{t}  \label{bn}
\end{equation}%
where $C=\sum_{j=0}^{\infty }C_{j}$, $C^{\ast }\left( L\right)
=\sum_{j=0}^{\infty }C_{j}^{\ast }L^{j}$ and $C_{j}^{\ast
}=\sum_{k=j+1}^{\infty }C_{k}$. Equation (\ref{bn}) is derived from the
multivariate Beveridge-Nelson decomposition of the filter $C\left( L\right)$
(see \citealp{watson1994vector}).

We assume that the $N\times N$ matrix $C$ can have reduced rank, say $m$.
This corresponds to assuming that the long-run behavior of the N-dimensional
vector $y_{t}$ is driven by m non-stationary common factors.

\begin{assumption}
\label{as-B}It holds that: (i) $rank\left( C\right) =m$, where $0\leq m\leq
N $; (ii) $\left\Vert C_{j}\right\Vert =O\left( \rho ^{j}\right) $ for some $%
0<\rho <1$.
\end{assumption}

Part \textit{(ii)} of the assumption requires that the $MA$ coefficients $%
C_{j}$ decline geometrically. This assumption is similar to Assumption 1 in %
\citet{caner}, where the $C_{j}$s are assumed to decline at a rate which
increases as the tail index of the innovations $\varepsilon _{t}$ decreases.
The case $m=0$ correspond to $y_{t}$ being (asymptotically) strictly
stationary. Assumption \ref{as-B} is also implied by Assumption 2.1 in %
\citet{ling}, where a finite-order VAR model under the classic $I\left(
1\right) $ conditions stated e.g. in \citet{reisnel} is considered.

\smallskip

On account of the possible rank reduction of $C$, a different formulation of
(\ref{bn}) may be helpful. It is always possible to write $C=PQ$, where $P$
and $Q$ are full rank matrices of dimensions $N\times m$ and $m\times N$
respectively. Defining the $m$-dimensional process $x_{t}=Q\sum_{s=1}^{t}%
\varepsilon _{s}$, and using the short-hand notation $u_{t}=C^{\ast }\left(
L\right) \varepsilon _{t}$, then (\ref{bn}) can be given the following
common trend representation 
\begin{equation}
y_{t}=Px_{t}+u_{t}\text{.}  \label{trends}
\end{equation}

We now make some assumptions on the error term $\varepsilon _{t}$.

\begin{assumption}
\label{as-moments}It holds that: (i) ${\ \varepsilon_{t}, 1\leq t \leq T }$
is an \textit{i.i.d.} sequence, with $\varepsilon_{t}=0$ for all $t<0$ and $%
y_0=0$; (ii) for all nonzero vectors $l\in 
\mathbb{R}
^{N}$, $l^{\prime }\varepsilon _{t}$ has distribution $F_{l\varepsilon }$
with strictly positive density, which is in the domain of attraction of a
strictly stable law $G$ with tail index $0<\eta \leq 2$.
\end{assumption}

As is generally the case in the analysis of time series with possibly
infinite variance, we assume that $\varepsilon _{t}$ is \textit{i.i.d}. Part 
\textit{(ii)} of the assumption implicitly states that the vector $%
\varepsilon _{t}$ has a multivariate distribution which belongs to the
domain of attraction of a strictly stable, multivariate law (see Theorem
2.1.5(a) in \citealp{taqqu}) with common tail index $\eta $, so that linear
combinations of them can be constructed. The assumption implies that, when $%
\eta <2$, $E\left\vert \varepsilon _{i,t}\right\vert ^{p}<\infty $ for all $%
0\leq p<\eta $, whereas $E\left\vert \varepsilon _{i,t}\right\vert ^{\eta
}=\infty $ (see \citealp{petrov1974}). Also, by Property 1.2.15 in %
\citet{taqqu}, it holds that, as $x\rightarrow \infty $%
\begin{equation*}
F_{l\varepsilon }\left( -x\right) =\frac{c_{l,1}+o\left( 1\right) }{x^{\eta }%
}L\left( x\right) \text{, and }1-F_{l\varepsilon }\left( x\right) =\frac{%
c_{l,2}+o\left( 1\right) }{x^{\eta }}L\left( x\right) ,
\end{equation*}%
where $L\left( x\right) $ is a slowly varying function in the sense of
Karamata (see \citealp{seneta} for a review), and $c_{l,1},c_{l,2}\geq 0$, $%
c_{l,1}+c_{l,2}>0$. The condition that $G$ is \textit{strictly }stable
entails that $c_{l,1}=c_{l,2}$ when $\eta =1$, thus ruling out asymmetry in
that case (see Property 1.2.8 in \citealp{taqqu}).

\subsection{Asymptotics\label{asymptotics}}

Define 
\begin{equation}
S_{11}=\sum_{t=1}^{T}y_{t}y_{t}^{\prime }\text{, and }S_{00}=\sum_{t=1}^{T}%
\Delta y_{t}\Delta y_{t}^{\prime }\text{.}  \label{s00}
\end{equation}

We report a set of novel results for the eigenvalues of $S_{11}$ and $S_{00}$%
, which we require for the construction of the test statistics.

\begin{proposition}
\label{s11-lambda}Let Assumptions \ref{as-B}-\ref{as-moments} hold. Then
there exists a random variable $T_{0}$ such that, for all $T\geq T_{0}$%
\begin{equation}
\lambda ^{\left( j\right) }\left( S_{11}\right) \geq c_{0}\frac{T^{1+2/\eta }%
}{\left( \ln \ln T\right) ^{2/\eta }},\text{ for }j\leq m\text{.}
\label{s11-lambda-1}
\end{equation}%
Also, for every $\epsilon >0$, it holds that%
\begin{equation}
\lambda ^{\left( j\right) }\left( S_{11}\right) =o_{a.s.}\left(
T^{2/p}\left( \ln T\right) ^{2\left( 2+\epsilon \right) /p}\right), \text{
for }j>m,  \label{s11-lambda-2}
\end{equation}%
for every $0<p<\eta $ when $\eta \leq 2$ with $E\left\Vert \varepsilon
_{t}\right\Vert ^{\eta }=\infty $, and $p=2$ when $\eta =2$ and $E\left\Vert
\varepsilon _{t}\right\Vert ^{\eta }<\infty $.
\end{proposition}

Proposition \ref{s11-lambda} states that the first $m$\ eigenvalues of $%
S_{11}$ diverge at a faster rate than the other ones (faster by an order of
\textquotedblleft almost\textquotedblright\ $T$), thus entailing that the
spectrum of $S_{11}$ has a \textquotedblleft spiked\textquotedblright\
structure.

\smallskip

In order to study $S_{00}$ and its eigenvalues, we need the following
assumption, which complements Assumption \ref{as-B}\textit{(ii)}.

\begin{assumption}
\label{as-deltayt}$\varepsilon _{t}$ has density $p_{\varepsilon }\left(
u\right) $ satisfying $\int_{-\infty }^{+\infty }\left\vert p_{\varepsilon
}\left( u+y\right) -p_{\varepsilon }\left( u\right) \right\vert du\leq
c_{0}\left\Vert y\right\Vert $.
\end{assumption}

The integral Lipschitz condition in Assumption \ref{as-deltayt} is a
technical requirement needed for $\Delta y_{t}$ to be strong mixing with
geometrically declining mixing numbers, and it is a standard requirement in
this literature (see e.g. \citealp{withers} and \citealp{phamtran}).

\begin{proposition}
\label{s00-lambda}Let Assumptions \ref{as-B}-\ref{as-deltayt} hold. Then 
\begin{equation}
\lambda ^{\left( 1\right) }\left( S_{00}\right) =o_{a.s.}\left( T^{2/\eta
}\left( \prod\limits_{i=1}^{n}\ln _{i}T\right) ^{2/\eta }\left( \ln
_{n+1}T\right) ^{\left( 2+\epsilon \right) /\eta }\right) ,
\label{s00-lambda-max}
\end{equation}%
for every $\epsilon >0$ and every integer $n$. Also, there exists a random
variable $T_{0}$ such that, for all $T\geq T_{0}$ and every $\epsilon >0$.%
\begin{equation}
\lambda ^{\left( N\right) }\left( S_{00}\right) \geq c_{0}\frac{T^{2/\eta }}{%
\left( \ln T\right) ^{\left( 2/\eta -1\right) \left( 2+\epsilon \right) }}.
\label{s00-lambda-min}
\end{equation}
\end{proposition}

Similarly to Proposition \ref{s11-lambda}, Proposition \ref{s00-lambda}
provides bounds for the spectrum of $S_{00}$. Part (\ref{s00-lambda-max})
has been shown in \citet{trapaniLIL}, and it is a Chover-type Law of the
Iterated Logarithm (\citealp{chover}). As shown in \citet{trapaniLIL}, the
bound in (\ref{s00-lambda-max}) is almost sharp. The lower bound implied in (%
\ref{s00-lambda-min}) is also almost sharp. \newline

The spectrum of $S_{11}$ (and, in particular, the different rates of
divergence of its eigenvalues) can -- in principle -- be employed in order
to determine $m$. However, $S_{11}$ has two main problems which make it
unsuitable for direct usage. First, by Proposition \ref{s11-lambda}, the
spectrum of $S_{11}$ depends on the nuisance parameter $\eta $. Also, the
matrix $S_{11}$ and, consequently, its spectrum depend on the unit of
measurement of the data, and thus they are not scale-free. In order to
construct scale-free and nuisance-free statistics, we propose to rescale $%
S_{11}$ by $S_{00}$. Proposition \ref{s00-lambda} ensures that this is
possible: by equation (\ref{s00-lambda-min}), the inverse of $S_{00}$ cannot
diverge too fast, and therefore the spectrum of the matrix $%
S_{00}^{-1}S_{11} $ should still have $m$ eigenvalues that diverge at a
faster rate than the others. This is shown in the next theorem.

\begin{theorem}
\label{s11s00}Let Assumptions \ref{as-B}-\ref{as-deltayt} hold. Then there
exists a random variable $T_{0}$ such that, for all $T\geq T_{0}$,%
\begin{equation}
\lambda ^{\left( j\right) }\left( S_{00}^{-1}S_{11}\right) \geq c_{0}\frac{%
T^{1-\epsilon ^{\prime }}}{\left( \ln \ln T\right) ^{2/\eta }\left( \ln
T\right) ^{2\left( 2+\epsilon \right) /p}},\text{ for }0\leq j\leq m\text{,}
\label{lambdaj-1}
\end{equation}%
for every $\epsilon ,\epsilon ^{\prime }>0$, and every $0<p<\eta $.
Moreover, for every $\epsilon ,\epsilon ^{\prime }>0$,%
\begin{equation}
\lambda ^{\left( j\right) }\left( S_{00}^{-1}S_{11}\right)
=o_{a.s.}(T^{\epsilon ^{\prime }}\left( \ln T\right) ^{2\left( 2+\epsilon
\right) /\eta +\left( 1+\epsilon \right) }),\text{ for }j>m\text{.}
\label{lambdaj-2}
\end{equation}
\end{theorem}

Theorem \ref{s11s00} states that the spectrum of $S_{00}^{-1}S_{11}$ has a
similar structure to the spectrum of $S_{11}$: the first $m$ eigenvalues are
spiked and their rate of divergence is faster than that of the remaining
eigenvalues by a factor of almost $T$. More importantly, by normalising $%
S_{11}$ by $S_{00}$, dependence on $\eta $ is relegated to the
slowly-varying (logarithmic) terms. In essence, apart from the slowly
varying sequences, equations (\ref{lambdaj-1}) and (\ref{lambdaj-2}) imply
that the rates of divergence of the eigenvalues of $S_{00}^{-1}S_{11}$ are
of order (arbitrarily close to) $O\left( T\right) $ for the spiked
eigenvalues, and (arbitrarily close to) $O\left( 1\right) $ for the other
ones. This is the key property of $\lambda ^{\left( j\right) }\left(
S_{00}^{-1}S_{11}\right) $: dividing $S_{11}$ by $S_{00}$\ washes out the
impact of the tail index $\eta $, which essentially does not play any role
in determining the divergence or not of $\lambda ^{\left( j\right) }\left(
S_{00}^{-1}S_{11}\right) $. Although this result pertains to first order
asymptotics (i.e., rates), it is possible to find an analogy between the
result in Theorem \ref{s11s00} and approaches based on eliminating nuisance
parameters using self-normalisation (see e.g. \citealp{shao}).

\section{Inference on the common trends\label{inftrends}}

In this section we collect our main results about estimation and inference
on the common trends in the possible presence of heavy tails. In Section \ref%
{test_stats}, we report a novel one-shot test about the (minimum) number of
common trends. Then, in Section \ref{estimation_R} we introduce a sequential
procedure for the determination of the number of common trends. Estimation
of the common trends and associated factor loading is presented in Section %
\ref{esttrend}.

\subsection{Testing hypotheses on the number of common trends\label%
{test_stats}}

We consider testing hypotheses on the number of common trends $m$. The tests
we propose herein will form the basis of our sequential procedure for the
determination of the number of common trends -- see Section \ref%
{estimation_R}. Specifically, we consider the null hypothesis%
\begin{equation}
\left\{ 
\begin{array}{c}
H_{0}:m\geq j \\ 
H_{A}:m<j%
\end{array}%
\right.  \label{fmwk-j}
\end{equation}%
where $j\in \{1,...,N\}$ is a (user-chosen) lower bound on the number of
common trends. For instance, testing whether there is at least one common
trend would entail setting $j=N-1$; a test of non-stationarity against the
alternative of strict stationarity corresponds to $j=1$.

Based on Theorem \ref{s11s00}, we propose to use%
\begin{equation}
\phi _{T}^{\left( j\right) }=\exp \left\{ T^{-\kappa }\lambda ^{\left(
j\right) }\left( S_{00}^{-1}S_{11}\right) \right\} -1\text{,}
\label{phi-j-t}
\end{equation}%
where $\kappa \in \left( 0,1\right) $; criteria for the choice of $\kappa $
in applications are discussed in Section \ref{montecarlo}. Importantly, the $%
T^{-\kappa }$ term in (\ref{phi-j-t}) is used in order to exploit the
discrepancy in the rates of divergence of the $\lambda ^{\left( j\right)
}\left( S_{00}^{-1}S_{11}\right) $ under $H_{0}$ and under $H_{A}$. In
particular, it ensures that $T^{-\kappa }\lambda ^{\left( j\right) }\left(
S_{00}^{-1}S_{11}\right) $ drifts to zero under $H_{A}$ (i.e., whenever $j>m$%
), whereas it still passes to infinity under $H_{0}$ (i.e., when $j\leq m$).
According to (\ref{lambdaj-2}), this only requires a very small value of $%
\kappa $, which would also allow $\lambda ^{\left( j\right) }\left(
S_{00}^{-1}S_{11}\right) $ to diverge at a rate close to $T$ under $H_{0}$.
On account of Theorem \ref{s11s00}, it holds that $P(\omega
:\lim_{T\rightarrow \infty }\phi _{T}^{\left( j\right) }=\infty )=1$ for $%
0\leq j\leq m$; hence, we can assume that under the null that $m\geq j$, it
holds that $\lim_{T\rightarrow \infty }\phi _{T}^{\left( j\right) }=\infty $%
. Conversely, we have $P(\omega :\lim_{T\rightarrow \infty }\phi
_{T}^{\left( j\right) }=0)=1$ for $j>m$, which entails that under the
alternative that $j>m$, we have $\lim_{T\rightarrow \infty }\phi
_{T}^{\left( j\right) }=0$. In essence, $\phi _{T}^{\left( j\right) }$
diverges to positive infinity, or converges (to zero), according to whether $%
\lambda ^{\left( j\right) }\left( S_{00}^{-1}S_{11}\right) $ is
\textquotedblleft large\textquotedblright\ or \textquotedblleft
small\textquotedblright .

Since the limiting law of $\phi _{T}^{\left( j\right) }$ under the null is
unknown, we propose a randomised version of it. The construction of the test
statistic is based on the following three step algorithm, which requires a
user-chosen weight function $F\left( \cdot \right) $ with support $%
U\subseteq \mathbb{R}$.

\smallskip

\noindent \textbf{Step 1} Generate an artificial sample $\{\xi _{i}^{\left(
j\right) },1\leq i\leq M\}$, with $\xi _{i}^{\left( j\right) }\sim $\textit{%
i.i.d.}$N\left( 0,1\right) $, independent of the original data.\newline
\noindent \textbf{Step 2} For each $u\in U$, define the Bernoulli sequence $%
\zeta _{i}^{\left( j\right) }\left( u\right) =I(\phi _{T}^{\left( j\right)
}\xi _{i}^{\left( j\right) }\leq u)$, and let 
\begin{equation}
\theta _{T,M}^{\left( j\right) }\left( u\right) =\frac{2}{\sqrt{M}}%
\sum_{i=1}^{M}\left( \zeta _{i}^{\left( j\right) }\left( u\right) -\frac{1}{2%
}\right) .  \label{theta-minor}
\end{equation}

\noindent \textbf{Step 3} Compute 
\begin{equation}
\Theta _{T,M}^{\left( j\right) }=\int_{U}\displaystyle\lbrack \theta
_{T,M}^{\left( j\right) }\left( u\right) ]^{2}dF\left( u\right) ,
\label{theta-maior}
\end{equation}%
where $F\left( \cdot \right) $ is the user-chosen weight function.

\smallskip

In Step 2, the binary variable $\zeta _{i}^{\left( j\right) }\left( u\right) 
$ is created for several values of $u \in U$, and in Step 3, the resulting
statistics $\theta _{T,M}^{\left( j\right) }\left( u\right) $ are averaged
across $u$, through the weight function $F\left( \cdot \right) $, thus
attenuating the dependence of the test statistic on an arbitrary value $u$.
The following assumption characterizes $F\left( \cdot \right) $.

\begin{assumption}
\label{weight} It holds that (i) $\int_{U}dF\left( u\right) =1$; (ii) $%
\int_{U}u^{2}dF\left( u\right) <\infty $.
\end{assumption}

A possible choice for $F\left( \cdot \right) $ could be a distribution
function with finite second moment; note that part \textit{(ii)} of the
assumption is trivially satisfied if $U$ is bounded. Possible examples
include a Rademacher distribution, based on choosing $U=\left\{ -c,c\right\} 
$ with $F\left( c\right) =F\left( -c\right) =1/2$, or the standard normal
distribution function.

\smallskip

Let $P^{\ast }$ denote the conditional probability with respect to the
original sample; we use the notation \textquotedblleft $\overset{D^{\ast }}{%
\rightarrow }$\textquotedblright\ and \textquotedblleft $\overset{P^{\ast }}{%
\rightarrow }$\textquotedblright\ to define, respectively, conditional
convergence in distribution and in probability according to $P^{\ast }$.

\begin{theorem}
\label{theta}Let Assumptions \ref{as-B}-\ref{weight} hold. Under $H_{0}$, as 
$\min (T,M)\rightarrow \infty $ with%
\begin{equation}
M^{1/2}\exp \left( -T^{1-\kappa -\epsilon }\right) \rightarrow 0
\label{rate-constraint}
\end{equation}%
for any arbitrarily small $\epsilon >0$, it holds that 
\begin{equation}
\Theta _{T,M}^{\left( j\right) }\overset{D^{\ast }}{\rightarrow }\chi
_{1}^{2}  \label{null-convergence}
\end{equation}%
for almost all realisations of $\left\{ \varepsilon _{t},0<t<\infty \right\} 
$. Under $H_{A}$, as $\min (T,M)\rightarrow \infty $, it holds that 
\begin{equation}
M^{-1}\Theta _{T,M}^{\left( j\right) }\overset{P^{\ast }}{\rightarrow }\frac{%
1}{4}  \label{alt-convergence}
\end{equation}%
for almost all realisations of $\left\{ \varepsilon _{t},0<t<\infty \right\} 
$.
\end{theorem}

Theorem \ref{theta} provides the limiting behaviour of $\Theta
_{T,M}^{\left( j\right) }$, also illustrating the impact of $M$ on the size
and power trade-off. According to (\ref{alt-convergence}), the larger $M$
the higher the power. Conversely, upon inspecting the proof, it emerges that 
$\theta _{T,M}^{\left( j\right) }\left( u\right) $ contains a non-centrality
parameter of order $O\left( M^{1/2}\exp \left( -T^{1-\kappa -\epsilon
}\right) \right) $, whence the upper bound in (\ref{rate-constraint}).

The one-shot test developed in this section has at least three advantages
compared to existing methods. First, our approach can also be implement to
check (asymptotic) strict stationarity. Indeed, running the test for $j=1$
corresponds to the null hypothesis that the data are driven by at least one
common trend; rejection supports the alternative of stationarity. Second,
running the test with $j=N$ corresponds to the null hypothesis that the $N$
variables do not cointegrate, thus offering a test for the null of no
cointegration against the alternative of (at least) one cointegrating
relation. Finally, we point out a further advantage over the well-known
method of \citet{johansen1991}. Johansen's likelihood ratio test allows to
test the null of rank $R$ (i.e., of $m=N-R$ common trends), where $R$ is
user-chosen, versus the alternative of rank greater than $R$ (i.e., less
than $N-R$ common trends). However, while the limiting distribution under
the null is well-known, if the true rank is \textit{lower} than $R$, then
the limiting distribution is different (see \citealp{bernstein2019}). Hence,
Johansen's test should be use only if the practitioner knows that the rank
cannot be lower than $R$. In contrast, since the null hypothesis is
formulated as a minimum bound on the number of common trends, our test does
not have this drawback.

A final remark on the test is in order. Letting $0<\alpha <1$ denote the
nominal level of the test, and defining $c_{\alpha }$ such that $P\left(
\chi _{1}^{2}>c_{\alpha }\right) =\alpha $, an immediate consequence of the
theorem is that under $H_{A}$ it holds that $\lim_{\min \left( T,M\right)
\rightarrow \infty }P^{\ast }(\Theta _{T,M}^{\left( j\right) }>c_{\alpha
})=1 $ for almost all realisations of $\left\{ \varepsilon _{t},0<t<\infty
\right\} $ - i.e. the test is consistent under the alternative. Conversely,
under $H_{0}$ we have, for almost all realisations of $\left\{ \varepsilon
_{t},0<t<\infty \right\} $%
\begin{equation}
\lim_{\min \left( T,M\right) \rightarrow \infty }P^{\ast }(\Theta
_{T,M}^{\left( j\right) }>c_{\alpha })=\alpha .  \label{size}
\end{equation}%
As noted also in \citet{HT16}, our test is constructed using a randomisation
which does not vanish asymptotically, and therefore the asymptotics of $%
\Theta _{T,M}^{\left( j\right) }$ is driven by the added randomness. Thus,
different researchers using the same data will obtain different values of $%
\Theta _{T,M}^{\left( j\right) }$ and, consequently, different $p$-values;
indeed, if an infinite number of researchers were to carry out the test, the 
$p$-values would follow a uniform distribution on $\left[ 0,1\right] $. This
is a well-known feature of randomised tests. In order to ameliorate this, %
\citet{HT16} suggest to compute $\Theta _{T,M}^{\left( j\right) }$ $S$
times, at each time $s$ using an independent sequence $\left\{ \xi
_{i,s}^{\left( j\right) }\right\} $ for $1\leq j\leq M$ and $1\leq s\leq S$,
thence defining%
\begin{equation}
Q_{\alpha ,S}=S^{-1}\sum_{s=1}^{S}I\left[ \Theta _{T,M,s}^{\left( j\right)
}\leq c_{\alpha }\right] .  \label{q}
\end{equation}%
Based on standard arguments, under $H_{0}$ the LIL\ yields%
\begin{equation}
\lim \inf_{S\rightarrow \infty }\lim_{\min \left( T,M\right) \rightarrow
\infty }\sqrt{\frac{S}{2\ln \ln S}}\frac{Q_{\alpha ,S}-\left( 1-\alpha
\right) }{\sqrt{\alpha \left( 1-\alpha \right) }}=-1.  \label{lil}
\end{equation}%
Hence, a \textquotedblleft strong rule\textquotedblright\ to decide in
favour of $H_{0}$ is 
\begin{equation}
Q_{\alpha ,S}\geq \left( 1-\alpha \right) -\sqrt{\alpha \left( 1-\alpha
\right) }\sqrt{\frac{2\ln \ln S}{S}}.  \label{strong}
\end{equation}%
Decisions made on the grounds of (\ref{strong}) have vanishing probabilities
of Type I and Type II errors, and are the same for all researchers: having $%
S\rightarrow \infty $ washes out the added randomness. No restrictions are
required on $S$ as it passes to infinity, so it can be arbitrarily large.

\subsection{Determining $m$\label{estimation_R}}

In order to determine the number of common trends $m$, we propose to cast
the individual one-shot tests discussed above in a sequential procedure,
where different values $j=1,2,...$ for $m$ are tested sequentially (note
that the individual tests must be based on artificial random samples
independent across $j$, see below). \newline
The estimator of $m$ (say, $\widehat{m}$) is the output of the following
algorithm: 

\smallskip

\noindent \textsc{Algorithm 1.}

\noindent \textbf{Step 1} Run the test for $H_{0}:m\geq 1$ based on $\Theta
_{T,M}^{\left( 1\right) }$. If the null is rejected, set $\widehat{m}=0$\
and stop, otherwise go to the next step.

\noindent \textbf{Step 2} Starting from $j=2$, run the test for $H_{0}:m\geq
j$ based on $\Theta _{T,M}^{\left( j\right) }$, constructed using an
artificial sample $\{\xi _{i}^{\left( j\right) }\}_{i=1}^{M}$ generated
independently of $\{\xi _{i}^{\left( 1\right) }\}_{i=1}^{M},$ $...,$ $\{\xi
_{i}^{\left( j-1\right) }\}_{i=1}^{M}$. If the null is rejected, set $%
\widehat{m}=j-1$\ and stop; otherwise, if $j=N$, set $\widehat{m}=N$;
otherwise, increase $j$ and repeat Step 2.

\smallskip

Consistency of the proposed procedure is presented in the next theorem.

\begin{theorem}
\label{family}Let Assumptions \ref{as-B}-\ref{weight} hold and define the
critical value of each individual test as $c_{\alpha }=c_{\alpha }\left(
M\right) $. As $\min \left( T,M\right) \rightarrow \infty $ under (\ref%
{rate-constraint}), if $c_{\alpha }\left( M\right) \rightarrow \infty $ with 
$c_{\alpha }=o\left( M\right) $, then it holds that $P^{\ast }(\widehat{m}%
=m)=1$ for\ almost all realisations of $\left\{ \varepsilon _{t},-\infty
<t<\infty \right\} $.
\end{theorem}

Theorem \ref{family} states that $\widehat{m}$ is consistent, as long as the
nominal level $\alpha $\ of the individual tests is chosen so as to drift to
zero: no specific rates are required. This can be better understood upon
inspecting the proof of the theorem: letting $\alpha $ denote the level of
each \ individual test, in (\ref{procedure-size}), we show that, as $\min
\left( T,M\right) \rightarrow \infty $, $P^{\ast }(\widehat{m}=m)=\left(
1-\alpha \right) ^{N-m}$, whence the requirement $c_{\alpha }\rightarrow
\infty $, which entails$\ \alpha \rightarrow 0$. This can also be read in
conjunction with Johansen's procedure (\citealp{johansen1991}) and its
bootstrap implementations (\citealp{cavaliereECTA}), whose outcome is an
estimate of $m$, say $\widetilde{m}$, such that, asymptotically, $P(%
\widetilde{m}=m)=1-\alpha $ for a given nominal value $\alpha $ for the
individual tests. On account of (\ref{procedure-size}), in our case choosing
a fixed nominal level $\alpha $ would yield, as mentioned above, $P^{\ast }(%
\widehat{m}=m)=\left( 1-\alpha \right) ^{N-m}$, which depends on the unknown 
$m$ and is, for $m>1$, worse than Johansen's procedure. On the other hand,
in our procedure the individual tests are independent (conditional on the
sample). Thus, in order to match the $1-\alpha $ probability of selecting
the true $m$, one can use a Bonferroni correction with $\alpha /N$ as
nominal level for each test, rather than $\alpha $. In this case, the same
calculations as in the proof of Theorem \ref{family} (and Bernoulli's
inequality) yield 
\begin{equation}
P^{\ast }(\widehat{m}=m)=\left( 1-\alpha /N\right) ^{N-m}\geq 1-\frac{N-m}{N}%
\alpha \geq 1-\alpha .  \label{bonferroni}
\end{equation}%
Finally, as an alternative to Bonferroni correction, one could consider the
following top-down algorithm.

\smallskip

\noindent \textsc{Algorithm 2.}

\noindent \textbf{Step 1} Run the test for $H_{0}:m=N$ based on $\Theta
_{T,M}^{\left( N\right) }$. If the null is not rejected, set $\widetilde{m}%
=N $\ and stop, otherwise go to the next step.

\noindent \textbf{Step 2} Starting from $j=N-1$, run the test for $%
H_{0}:m\geq j$ based on $\Theta _{T,M}^{\left( j\right) }$, constructed
using an artificial sample $\{\xi _{i}^{\left( j\right) }\}_{i=1}^{M}$
generated independently of $\{\xi _{i}^{\left( 1\right) }\}_{i=1}^{M},$ $%
..., $ $\{\xi _{i}^{\left( j-1\right) }\}_{i=1}^{M}$. If the null is
rejected, set $\widetilde{m}$\ and stop; otherwise, if $j=1$, set $%
\widetilde{m}$; otherwise, decrease $j$ and repeat Step 2.

\smallskip

The proof of the consistency of $\widetilde{m}$, and some Monte Carlo
evidence (showing that the performance of $\widetilde{m}$ is virtually
indistinguishable from that of $\widehat{m}$), are in the Supplement.

\subsection{Estimation of the common trends\label{esttrend}}

Recall the common trend representation provided in (\ref{trends}),%
\begin{equation*}
y_{t}=Px_{t}+u_{t}\text{.}
\end{equation*}%
After determining $m$, it is possible to estimate the non-stationary common
stochastic trends $x_{t}$ by using Principal Components (PC), in a similar
fashion to \citet{penaponcela} and \citet{zyr18}.

In particular, let $\widehat{\upsilon }_{j}$ denote the eigenvector
corresponding to the $j$-th largest eigenvalue of $S_{11}$ under the
orthonormalisation restrictions $\left\Vert \widehat{\upsilon }%
_{j}\right\Vert =1$ and $\widehat{\upsilon }_{i}^{\prime }\widehat{\upsilon }%
_{j}=0$ for all $i\neq j$. Then, defining $\widehat{P}=(\widehat{\upsilon }%
_{1},...,\widehat{\upsilon }_{m})$, the estimator of the common trends $%
x_{t} $ is%
\begin{equation*}
\widehat{x}_{t}=\widehat{P}^{\prime }y_{t}.
\end{equation*}%
The next theorem provides the consistency (up to a rotation) of the
estimators of $P$ and $x_{t}$. Interestingly, the convergence rate of $%
\widehat{P}$ is not affected by the tail index.

\begin{theorem}
\label{bai04}We assume that Assumptions \ref{as-B}-\ref{function} are
satisfied. Then it holds that%
\begin{equation}
\left\Vert \widehat{P}-PH\right\Vert =O_{P}\left( T^{-1}\right) ,
\label{loadings}
\end{equation}%
\begin{equation}
\left\Vert \widehat{x}_{t}-H^{-1}x_{t}\right\Vert =O_{P}\left( 1\right)
+O_{P}\left( T^{-1+1/\eta }\right) ,  \label{factors}
\end{equation}%
for each $1\leq t\leq T$, where $H$ is an $N\times N$ invertible matrix.
\end{theorem}

Theorem \ref{bai04} states that both $\widehat{P}$ and $\widehat{x}_{t}$ are
consistent estimators of $P$ and $x_{t}$ -- up to an invertible linear
transformation, since it is only possible to provide a consistent estimate
of the eigenspace, as opposed to the individual eigenvectors. Equation (\ref%
{loadings}) states that the estimated loadings\ $\widehat{P}$ provide a
superconsistent estimator of (a linear combination of the columns of) $P$.
This result, which is the same as in the case of finite variance, is a
consequence of the fact that $x_{t}$ is an \textquotedblleft
integrated\textquotedblright\ process, and it is related to the eigen-gap
found in Proposition \ref{s11-lambda}. Equation (\ref{loadings}) could also
be read in conjunction with the literature on large factor models, where --
contrary to our case -- it is required that $N\rightarrow \infty $. In that
context, \citet{bai04} obtains the same result as in (\ref{loadings}) albeit
for the case of finite variance: thus, in the presence of integrated
processes, the PC\ estimator is always superconsistent, irrespective of $N$
passing to infinity or not.

According to (\ref{factors}), $\widehat{x}_{t}$ also is a consistent
estimator of the space spanned by $x_{t}$. The \textquotedblleft
noise\textquotedblright\ component does not drift to zero and, when $\eta <1$%
, it may even diverge; however, the \textquotedblleft
signal\textquotedblright\ $x_{t}$ is of order $O_{P}\left( t^{1/\eta
}\right) $, and therefore it dominates the estimation error (in fact, when $%
\eta <1$, the estimation error is smaller by a factor $T$). This result can
be compared to the estimator proposed by \citet{gg95}, which is studied
under finite second moment and requires a full specification of the VECM.
Furthermore, (\ref{factors}) can also be compared with the findings in the
large factor models literature (see, in particular, Lemma 2 in %
\citealp{bai04}). As far as uniform rates in $t$ are concerned, in the proof
of the theorem we also show the strongest result%
\begin{equation}
\max_{1\leq t\leq T}\left\Vert \widehat{x}_{t}-H^{-1}x_{t}\right\Vert
=O_{P}\left( T^{1/\eta }\right) .  \label{unif-x}
\end{equation}%
This result is unlikely to be sharp, since it is based on the (rather crude)
fact that the maximum of any $T$-dimensional sequence with finite $p$-th
moment is bounded by $O_{P}\left( T^{1/p}\right) $.

\section{Extensions\label{extensions}}

The framework developed in the previous section does not allow for
deterministic terms in the data, and requires $\varepsilon _{t}$ to be
identically distributed. We now discuss possible extensions of our set-up,
to accommodate for heterogeneous innovations and deterministics. The main
result of this section is that our procedure can be implemented even in
these cases, with no modifications required.

\subsection{Heterogeneous innovations}

We consider a novel framework where we allow for innovation heterogeneity of
a very general form. Specifically, we assume that 
\begin{equation}
\varepsilon _{t}=h\left(\frac{t}{T}\right)u_{t}\text{,}  \label{bvh}
\end{equation}%
where $u_{t}$ satisfies Assumption \ref{as-moments} and $h\left( \cdot
\right) $ is a deterministic function. The representation in (\ref{bvh}) has
also been employed in order to consider the presence of heteroskedasticity
in the case of data with finite variance (see e.g. \citealp{cavaliere2009}).

\begin{assumption}
\label{function}$h\left( \cdot \right) $ is nontrivial, nonnegative and of
bounded variation on $\left[ 0,1\right] $.
\end{assumption}

The only requirement on the scale function $h\left( \cdot \right) $ is that
it has bounded variation within the interval $\left[ 0,1\right] $. Functions
of bounded variations and their properties are very well studied, and we
refer to e.g. \citet{hewitt} for details. The design in (\ref{bvh}) includes
several cases which can be of potential interest: $h\left( r\right) $ can be
piecewise linear, i.e. $h\left( r\right) =\sum_{i=1}^{n+1}h_{i}I\left(
c_{i-1}\leq r<c_{i}\right) $, with $c_{0}=0$ and $c_{n+1}=1$, thus
considering the possible presence of jumps/regimes in the heterogeneity of $%
\varepsilon _{t}$; or it could be a polynomial function.

\begin{corollary}
\label{heterosk}Let Assumptions \ref{as-B}-\ref{function} hold, with
Assumption \ref{as-moments} modified to contain only symmetric stable $u_{t}$%
. Then, as $\min \left( T,M\right) \rightarrow \infty $ with (\ref%
{rate-constraint}), it holds that, for all $j$%
\begin{equation}
P^{\ast }(\Theta _{T,M}^{\left( j\right) }>c_{\alpha })\rightarrow \alpha ,
\label{size-het}
\end{equation}%
under $H_{0}$, with probability tending to $1$. Under $H_{A}$, (\ref%
{alt-convergence}) holds for each $j$, for almost all realisations of $%
\left\{ u_{t},0<t<\infty \right\} $.
\end{corollary}

Repeating the proof of Theorem \ref{family}, the results in Corollary \ref%
{heterosk} entail that, using the Algorithm 1 in Section \ref{estimation_R}, 
$P^{\ast }(\widehat{m}=m)\rightarrow 1$ with probability tending to 1: $%
\widehat{m}$ is still a consistent estimator of $m$. Similarly, it can be
shown that the asymptotic properties of our estimator of the common trends
and associated loadings are unaffected by heterogeneity in the scale, as in (%
\ref{bvh}).

\subsection{Deterministics}

In this section we show that our procedure can be also applied to data which
contain a deterministic term. Specifically, we consider the representation%
\begin{equation}
y_{t}=\mu +C\sum_{s=1}^{t}\varepsilon _{s}+C^{\ast }\left( L\right)
\varepsilon _{t}  \label{var-const}
\end{equation}%
where $C$ and $C^{\ast }\left( L\right) $ are defined as before. Equation (%
\ref{bn}) is derived from the multivariate Beveridge-Nelson decomposition of
the filter $C\left( L\right) $. As is known, (\ref{var-const}) can also be
obtained from a VECM representation (see \citealp{ling} and %
\citealp{yapreisnel}) 
\begin{equation}
\Delta y_{t}=\mu +\alpha \beta ^{\prime }y_{t-1}+\sum_{j=1}^{p-1}\Gamma
_{j}\Delta y_{t-j}+\varepsilon _{t},  \label{vecm-const}
\end{equation}%
under the constraint $\mu =\alpha \rho $ with $\rho$ an $m \times 1$ vector.

In this case, our procedure still yields the same results as without the
deterministic term.

\begin{corollary}
\label{constant} We assume that (\ref{vecm-const}) holds. Then, Theorems \ref%
{theta}, \ref{family} and \ref{bai04} hold under the same assumptions.
\end{corollary}

\section{Monte Carlo evidence\label{montecarlo}}

In this section, we discuss the implementation/specification of our
procedure and illustrate its finite sample properties through a small scale
Monte Carlo exercise. To save space, we report only a limited number of
results; further results are in the Supplement.

\subsection{Design}

Following \citet{ling}, we simulate the $N$-variate $VAR\left( 1\right) $
model%
\begin{equation}
y_{t}=Ay_{t-1}+\varepsilon _{t}\text{,}  \label{var-1}
\end{equation}%
initialized at $y_{0}=0$ and where $A=I_{N}-PP^{\prime }$, $P$ being an $%
N\times \left( N-m\right) $ matrix with orthonormal columns (i.e., $%
P^{\prime }P=I_{N-m}$).\footnote{%
Specifically, in order to create $P$, we have used $P=D\left( D^{\prime
}D\right) ^{-1/2}$, where $\left( M\right) ^{-1/2}$ denotes the Choleski
factor of a matrix $M$, and have set $D\sim \mathbf{1}_{N\times
(N-m)}+d_{N\times (N-m)}$, where $\mathbf{1}_{N\times (N-m)}$ is an $N\times
(N-m)$ matrix of ones and $d_{N\times (N-m)}$ is an $N\times (N-m)$ matrix
such that $vec\left( d_{N\times (N-m)}\right) \sim N\left( 0,\mathbf{1}%
_{N(N-m)}\right) $. In particular, $d_{N\times (N-m)}$ is kept fixed across
Monte Carlo iterations.}.The innovations $\varepsilon _{t}$ in (\ref{var-1})
are \textit{i.i.d.} and coordinate-wise independent, from a power law
distribution with tail index $\eta \in \left\{ 0.5,1,1.5,2\right\} $. We
follow the procedure proposed by \citet{clauset} and generate $\varepsilon
_{i,t}$ as%
\begin{equation}
\varepsilon _{i,t}=\left( 1-v_{i,t}\right) ^{-1/\eta }\text{,}
\label{powerlaw}
\end{equation}%
where $v_{i,t}$ is $i.i.d.U\left[ 0,1\right] $; $\varepsilon _{i,t}$ is
subsequently centered.\footnote{%
In unreported experiments, we have also considered the case $\varepsilon
_{t}\sim i.i.d.N\left( 0,I_{N}\right) $, but results are essentially the
same as in the case $\eta=2$.}

First, we note from unreported experiments that our procedure for the
determination of the number of common trends is not particularly sensitive
to the choice of the various specifications. In our experiments, we have
used $M=100$ to speed up the computational time, but we note that results do
not change when setting e.g. $M=T$, $M=T/2$ or $M=T/4$. In (\ref{phi-j-t}),
we have used $\kappa =10^{-4}$. This is a conservative choice, whose
rationale follows from the fact that, in (\ref{phi-j-t}), dividing by $%
T^{\kappa }$ serves the purpose of making the non-spiked eigenvalues drift
to zero. The upper bound provided in (\ref{lambdaj-2}) for such non-spiked
eigenvalues is given by slowly varying functions, which suggests that even a
very small value of $\kappa $ should suffice. Indeed, altering the value of $%
\kappa $ has virtually no consequence. In order to compute the integral in (%
\ref{theta-maior}), we use the Gauss-Hermite quadrature.\footnote{%
In our case, we have used%
\begin{equation}
\Theta _{T,M}^{j}=\frac{1}{\sqrt{\pi }}\sum_{s=1}^{n_{S}}w_{s}\left( \theta
_{T,M}^{\left( j\right) }\left( \sqrt{2}z_{s}\right) \right) ^{2},
\label{upsilon-feasible}
\end{equation}%
where the $z_{s}$s, $1\leq s\leq n_{S}$, are the zeros of the Hermite
polynomial $H_{n_{S}}\left( z\right) $ and the weights $w_{s}$ are defined
as 
\begin{equation}
w_{s}=\frac{\sqrt{\pi }2^{n_{S}-1}\left( n_{S}-1\right) !}{n_{S}\left[
H_{n_{S}-1}\left( z_{s}\right) \right] ^{2}}.  \label{hermite-weights}
\end{equation}%
Thus, when computing $\theta _{T,M}^{\left( j\right) }\left( u\right) $ in
Step 2 of the algorithm, we construct $n_{S}$ of these statistics, each
using $u=\pm \sqrt{2}z_{s}$. The values of the roots $z_{s}$, and of the
corresponding weights $w_{s}$, are tabulated e.g. in \citet{salzer}. In our
case, we have used $n_{S}=2$, which corresponds to $u=\pm 1$ with equal
weight $\frac{1}{2}$; we note that in unreported experiments we tried $%
n_{S}=4$ with the corresponding weights, but there were no changes up to the 
$4$-th decimal in the empirical rejection frequencies.} Finally, as far as
the family-wise detection procedure is concerned, the level of the
individual tests is $\alpha \left( T\right) =0.05/T$, as also used in %
\citet{bt1}; this corresponds to having a critical value $c_{\alpha }$ which
grows logarithmically with $T$. All routines are based on $1,000$ iterations
and are written using GAUSS 21.

\begin{table}[t]
\caption{Estimation frequencies - $N=3$}
\label{tab:TableF1}{\tiny 
\begin{tabular}{ccccccccccccc}
\multicolumn{13}{c}{$N=3$} \\ 
&  &  & \multicolumn{4}{c}{$T=100$} &  & \multicolumn{4}{c}{$T=200$} &  \\ 
\hline
&  &  &  &  &  &  &  &  &  &  &  &  \\ 
&  & $m$ & $3$ & $2$ & $1$ & $0$ &  & $3$ & $2$ & $1$ & $0$ &  \\ 
& $\widehat{m}$ &  &  &  &  &  &  &  &  &  &  &  \\ 
& $3$ &  & $0.963$ & $0.004$ & $0.000$ & $0.000$ &  & $0.986$ & $0.001$ & $%
0.000$ & $0.000$ &  \\ 
$\eta =0.5$ & $2$ &  & $0.037$ & $0.990$ & $0.011$ & $0.000$ &  & $0.013$ & $%
0.994$ & $0.002$ & $0.000$ &  \\ 
& $1$ &  & $0.000$ & $0.006$ & $0.989$ & $0.023$ &  & $0.001$ & $0.003$ & $%
0.998$ & $0.005$ &  \\ 
& $0$ &  & $0.000$ & $0.000$ & $0.000$ & $0.977$ &  & $0.000$ & $0.002$ & $%
0.000$ & $0.995$ &  \\ 
&  &  &  &  &  &  &  &  &  &  &  &  \\ 
& $3$ &  & $0.986$ & $0.001$ & $0.000$ & $0.000$ &  & $0.995$ & $0.000$ & $%
0.000$ & $0.000$ &  \\ 
$\eta =1.0$ & $2$ &  & $0.014$ & $0.995$ & $0.001$ & $0.000$ &  & $0.004$ & $%
0.997$ & $0.000$ & $0.000$ &  \\ 
& $1$ &  & $0.000$ & $0.004$ & $0.999$ & $0.004$ &  & $0.001$ & $0.002$ & $%
1.000$ & $0.003$ &  \\ 
& $0$ &  & $0.000$ & $0.000$ & $0.000$ & $0.996$ &  & $0.000$ & $0.001$ & $%
0.000$ & $0.997$ &  \\ 
&  &  &  &  &  &  &  &  &  &  &  &  \\ 
& $3$ &  & $0.991$ & $0.000$ & $0.000$ & $0.000$ &  & $0.996$ & $0.000$ & $%
0.000$ & $0.000$ &  \\ 
$\eta =1.5$ & $2$ &  & $0.009$ & $1.000$ & $0.001$ & $0.000$ &  & $0.003$ & $%
0.999$ & $0.000$ & $0.000$ &  \\ 
& $1$ &  & $0.000$ & $0.000$ & $0.999$ & $0.001$ &  & $0.001$ & $0.001$ & $%
1.000$ & $0.002$ &  \\ 
& $0$ &  & $0.000$ & $0.000$ & $0.000$ & $0.999$ &  & $0.000$ & $0.000$ & $%
0.000$ & $0.998$ &  \\ 
&  &  &  &  &  &  &  &  &  &  &  &  \\ 
& $3$ &  & $0.994$ & $0.000$ & $0.000$ & $0.000$ &  & $0.998$ & $0.000$ & $%
0.000$ & $0.000$ &  \\ 
$\eta =2$ & $2$ &  & $0.006$ & $1.000$ & $0.001$ & $0.000$ &  & $0.001$ & $%
0.999$ & $0.000$ & $0.000$ &  \\ 
& $1$ &  & $0.000$ & $0.000$ & $0.999$ & $0.000$ &  & $0.001$ & $0.001$ & $%
1.000$ & $0.000$ &  \\ 
& $0$ &  & $0.000$ & $0.000$ & $0.000$ & $1.000$ &  & $0.000$ & $0.000$ & $%
0.000$ & $1.000$ &  \\ 
&  &  &  &  &  &  &  &  &  &  &  &  \\ \hline\hline
\end{tabular}
}
\end{table}

\begin{table*}[h]
\caption{Estimation frequencies - $N=4$}
\label{tab:TableF2}{%
{\tiny 
\begin{tabular}{ccccccccccccccc}
\multicolumn{15}{c}{$N=4$} \\ 
&  &  & \multicolumn{5}{c}{$T=100$} &  & \multicolumn{5}{c}{$T=200$} &  \\ 
\hline
&  &  &  &  &  &  &  &  &  &  &  &  &  &  \\ 
&  & $m$ & $4$ & $3$ & $2$ & $1$ & $0$ &  & $4$ & $3$ & $2$ & $1$ & $0$ & 
\\ 
& $\widehat{m}$ &  &  &  &  &  &  &  &  &  &  &  &  &  \\ 
& $4$ &  & $0.890$ & $0.004$ & $0.000$ & $0.000$ & $0.000$ &  & $0.976$ & $%
0.001$ & $0.000$ & $0.000$ & $0.000$ &  \\ 
& $3$ &  & $0.104$ & $0.964$ & $0.005$ & $0.000$ & $0.000$ &  & $0.023$ & $%
0.988$ & $0.002$ & $0.000$ & $0.000$ &  \\ 
$\eta =0.5$ & $2$ &  & $0.003$ & $0.029$ & $0.989$ & $0.015$ & $0.000$ &  & $%
0.000$ & $0.011$ & $0.996$ & $0.008$ & $0.000$ &  \\ 
& $1$ &  & $0.002$ & $0.000$ & $0.006$ & $0.984$ & $0.022$ &  & $0.000$ & $%
0.000$ & $0.001$ & $0.992$ & $0.012$ &  \\ 
& $0$ &  & $0.001$ & $0.003$ & $0.000$ & $0.001$ & $0.978$ &  & $0.001$ & $%
0.000$ & $0.001$ & $0.000$ & $0.988$ &  \\ 
&  &  &  &  &  &  &  &  &  &  &  &  &  &  \\ 
& $4$ &  & $0.948$ & $0.000$ & $0.000$ & $0.000$ & $0.000$ &  & $0.995$ & $%
0.000$ & $0.000$ & $0.000$ & $0.000$ &  \\ 
& $3$ &  & $0.046$ & $0.990$ & $0.000$ & $0.000$ & $0.000$ &  & $0.004$ & $%
0.997$ & $0.001$ & $0.000$ & $0.000$ &  \\ 
$\eta =1.0$ & $2$ &  & $0.003$ & $0.007$ & $0.995$ & $0.003$ & $0.000$ &  & $%
0.000$ & $0.003$ & $0.998$ & $0.001$ & $0.000$ &  \\ 
& $1$ &  & $0.002$ & $0.001$ & $0.004$ & $0.994$ & $0.009$ &  & $0.000$ & $%
0.000$ & $0.000$ & $0.999$ & $0.000$ &  \\ 
& $0$ &  & $0.001$ & $0.002$ & $0.001$ & $0.003$ & $0.991$ &  & $0.001$ & $%
0.000$ & $0.001$ & $0.000$ & $1.000$ &  \\ 
&  &  &  &  &  &  &  &  &  &  &  &  &  &  \\ 
& $4$ &  & $0.967$ & $0.001$ & $0.000$ & $0.000$ & $0.000$ &  & $0.998$ & $%
0.000$ & $0.000$ & $0.000$ & $0.000$ &  \\ 
& $3$ &  & $0.029$ & $0.990$ & $0.000$ & $0.000$ & $0.000$ &  & $0.001$ & $%
0.999$ & $0.001$ & $0.000$ & $0.000$ &  \\ 
$\eta =1.5$ & $2$ &  & $0.001$ & $0.005$ & $0.997$ & $0.000$ & $0.000$ &  & $%
0.000$ & $0.000$ & $0.998$ & $0.000$ & $0.000$ &  \\ 
& $1$ &  & $0.001$ & $0.002$ & $0.003$ & $0.998$ & $0.001$ &  & $0.000$ & $%
0.000$ & $0.000$ & $1.000$ & $0.001$ &  \\ 
& $0$ &  & $0.002$ & $0.002$ & $0.000$ & $0.002$ & $0.999$ &  & $0.001$ & $%
0.001$ & $0.001$ & $0.000$ & $0.999$ &  \\ 
&  &  &  &  &  &  &  &  &  &  &  &  &  &  \\ 
& $4$ &  & $0.974$ & $0.000$ & $0.000$ & $0.000$ & $0.000$ &  & $0.998$ & $%
0.000$ & $0.000$ & $0.000$ & $0.000$ &  \\ 
& $3$ &  & $0.022$ & $0.995$ & $0.000$ & $0.000$ & $0.000$ &  & $0.001$ & $%
0.999$ & $0.000$ & $0.000$ & $0.000$ &  \\ 
$\eta =2$ & $2$ &  & $0.001$ & $0.002$ & $0.999$ & $0.000$ & $0.000$ &  & $%
0.000$ & $0.000$ & $0.999$ & $0.000$ & $0.000$ &  \\ 
& $1$ &  & $0.001$ & $0.001$ & $0.000$ & $0.998$ & $0.001$ &  & $0.000$ & $%
0.000$ & $0.000$ & $1.000$ & $0.001$ &  \\ 
& $0$ &  & $0.002$ & $0.002$ & $0.001$ & $0.002$ & $0.999$ &  & $0.001$ & $%
0.001$ & $0.001$ & $0.000$ & $0.999$ &  \\ 
&  &  &  &  &  &  &  &  &  &  &  &  &  &  \\ \hline\hline
\end{tabular}
} }
\end{table*}

\begin{table*}[t]
\caption{Estimation frequencies - $N=5$}
\label{tab:TableF3}{%
{\tiny 
\begin{tabular}{ccccccccccccccccc}
\multicolumn{17}{c}{$N=5$} \\ 
&  &  & \multicolumn{6}{c}{$T=100$} &  & \multicolumn{6}{c}{$T=200$} &  \\ 
\hline
&  &  &  &  &  &  &  &  &  &  &  &  &  &  &  &  \\ 
&  & $m$ & $5$ & $4$ & $3$ & $2$ & $1$ & $0$ &  & $5$ & $4$ & $3$ & $2$ & $1$
& $0$ &  \\ 
& $\widehat{m}$ &  &  &  &  &  &  &  &  &  &  &  &  &  &  &  \\ 
& $5$ &  & $0.787$ & $0.003$ & $0.000$ & $0.000$ & $0.000$ & $0.000$ &  & $%
0.944$ & $0.002$ & $0.000$ & $0.000$ & $0.000$ & $0.000$ &  \\ 
& $4$ &  & $0.201$ & $0.919$ & $0.005$ & $0.000$ & $0.000$ & $0.000$ &  & $%
0.055$ & $0.985$ & $0.004$ & $0.000$ & $0.000$ & $0.000$ &  \\ 
$\eta =0.5$ & $3$ &  & $0.009$ & $0.077$ & $0.974$ & $0.007$ & $0.000$ & $%
0.000$ &  & $0.000$ & $0.013$ & $0.992$ & $0.004$ & $0.000$ & $0.000$ &  \\ 
& $2$ &  & $0.000$ & $0.000$ & $0.020$ & $0.993$ & $0.001$ & $0.001$ &  & $%
0.000$ & $0.000$ & $0.003$ & $0.996$ & $0.014$ & $0.000$ &  \\ 
& $1$ &  & $0.003$ & $0.001$ & $0.001$ & $0.000$ & $0.998$ & $0.041$ &  & $%
0.001$ & $0.000$ & $0.001$ & $0.000$ & $0.986$ & $0.018$ &  \\ 
& $0$ &  & $0.000$ & $0.000$ & $0.000$ & $0.000$ & $0.001$ & $0.958$ &  & $%
0.000$ & $0.000$ & $0.000$ & $0.000$ & $0.000$ & $0.982$ &  \\ 
&  &  &  &  &  &  &  &  &  &  &  &  &  &  &  &  \\ 
& $5$ &  & $0.874$ & $0.000$ & $0.000$ & $0.000$ & $0.000$ & $0.000$ &  & $%
0.984$ & $0.000$ & $0.000$ & $0.000$ & $0.000$ & $0.000$ &  \\ 
& $4$ &  & $0.125$ & $0.959$ & $0.002$ & $0.000$ & $0.000$ & $0.000$ &  & $%
0.014$ & $0.997$ & $0.000$ & $0.000$ & $0.000$ & $0.000$ &  \\ 
$\eta =1.0$ & $3$ &  & $0.000$ & $0.039$ & $0.990$ & $0.002$ & $0.000$ & $%
0.000$ &  & $0.001$ & $0.003$ & $0.999$ & $0.000$ & $0.000$ & $0.000$ &  \\ 
& $2$ &  & $0.000$ & $0.001$ & $0.007$ & $0.995$ & $0.004$ & $0.000$ &  & $%
0.001$ & $0.000$ & $0.000$ & $1.000$ & $0.001$ & $0.000$ &  \\ 
& $1$ &  & $0.001$ & $0.001$ & $0.001$ & $0.002$ & $0.995$ & $0.019$ &  & $%
0.000$ & $0.000$ & $0.000$ & $0.000$ & $0.999$ & $0.004$ &  \\ 
& $0$ &  & $0.000$ & $0.000$ & $0.000$ & $0.001$ & $0.001$ & $0.981$ &  & $%
0.000$ & $0.000$ & $0.001$ & $0.000$ & $0.000$ & $0.996$ &  \\ 
&  &  &  &  &  &  &  &  &  &  &  &  &  &  &  &  \\ 
& $5$ &  & $0.909$ & $0.000$ & $0.000$ & $0.000$ & $0.000$ & $0.000$ &  & $%
0.996$ & $0.000$ & $0.000$ & $0.000$ & $0.000$ & $0.000$ &  \\ 
& $4$ &  & $0.090$ & $0.974$ & $0.000$ & $0.000$ & $0.000$ & $0.000$ &  & $%
0.002$ & $0.998$ & $0.000$ & $0.000$ & $0.000$ & $0.000$ &  \\ 
$\eta =1.5$ & $3$ &  & $0.000$ & $0.024$ & $0.995$ & $0.002$ & $0.000$ & $%
0.000$ &  & $0.001$ & $0.001$ & $0.999$ & $0.000$ & $0.000$ & $0.000$ &  \\ 
& $2$ &  & $0.000$ & $0.000$ & $0.005$ & $0.995$ & $0.001$ & $0.000$ &  & $%
0.001$ & $0.001$ & $0.000$ & $1.000$ & $0.000$ & $0.000$ &  \\ 
& $1$ &  & $0.001$ & $0.002$ & $0.000$ & $0.002$ & $0.999$ & $0.002$ &  & $%
0.000$ & $0.000$ & $0.000$ & $0.000$ & $1.000$ & $0.002$ &  \\ 
& $0$ &  & $0.000$ & $0.000$ & $0.000$ & $0.001$ & $0.000$ & $0.998$ &  & $%
0.000$ & $0.000$ & $0.001$ & $0.000$ & $0.000$ & $0.998$ &  \\ 
&  &  &  &  &  &  &  &  &  &  &  &  &  &  &  &  \\ 
& $5$ &  & $0.914$ & $0.000$ & $0.000$ & $0.000$ & $0.000$ & $0.000$ &  & $%
0.997$ & $0.000$ & $0.000$ & $0.000$ & $0.000$ & $0.000$ &  \\ 
& $4$ &  & $0.085$ & $0.981$ & $0.000$ & $0.000$ & $0.000$ & $0.000$ &  & $%
0.002$ & $0.998$ & $0.000$ & $0.000$ & $0.000$ & $0.000$ &  \\ 
$\eta =2$ & $3$ &  & $0.000$ & $0.017$ & $0.993$ & $0.000$ & $0.000$ & $%
0.000 $ &  & $0.000$ & $0.001$ & $0.999$ & $0.000$ & $0.000$ & $0.000$ &  \\ 
& $2$ &  & $0.000$ & $0.001$ & $0.005$ & $0.999$ & $0.001$ & $0.000$ &  & $%
0.000$ & $0.001$ & $0.000$ & $1.000$ & $0.000$ & $0.000$ &  \\ 
& $1$ &  & $0.001$ & $0.001$ & $0.001$ & $0.000$ & $0.999$ & $0.002$ &  & $%
0.001$ & $0.000$ & $0.000$ & $0.000$ & $1.000$ & $0.002$ &  \\ 
& $0$ &  & $0.000$ & $0.000$ & $0.001$ & $0.001$ & $0.000$ & $0.998$ &  & $%
0.000$ & $0.000$ & $0.001$ & $0.000$ & $0.000$ & $0.998$ &  \\ 
&  &  &  &  &  &  &  &  &  &  &  &  &  &  &  &  \\ \hline\hline
\end{tabular}
} }
\end{table*}

\subsection{Results}

Results are reported in Tables \ref{tab:TableF1}-\ref{tab:TableF3}, where we
analyse the properties of our estimator of $m$ with $N\in \left\{
3,4,5\right\} $. The reported frequencies of the estimates of $m$ show that
the finite sample properties are largely satisfactory. Our procedure seems
to be scarcely affected by the value of $m$, although, especially for the
smaller sample sizes, it appears to be marginally better when $m=0$ as
opposed to the case $m=N$. This difference, however, vanishes as $T$
increases. The impact of $N$ is also very clear: as the $VAR$ dimension
increases, the performance of $\widehat{m}$ tends to deteriorate, as
expected.\ Inference improves for larger values of $T$. Indeed, whilst
results for $N=3$ are good even when $\eta =0.5$ and $T=100$, when $N=5$ the
estimator $\widehat{m}$ requires at least $T=200$ in order to have a
frequency of correctly picking the true value of $m$ higher than $90\%$.
This is, as noted above, more pronounced when $m=N$, and less so when $m=0$.
As it can also be expected, our procedure improves as $\eta $ increases;
results are anyway very good even in the (very extreme) case $\eta =0.5$,
and the impact of $\eta $ is less and less important as $T$ increases.
Finally, although Tables \ref{tab:TableF1}-\ref{tab:TableF3} focus only on
the \textit{i.i.d.} case, unreported experiments showed that results are
essentially the same when allowing for serial dependence.

In the Supplement, we report a broader set of results which, in addition to
serial dependence in the errors $\varepsilon _{i,t}$, also compare the
proposed method with classic information criteria. Results are in Tables \ref%
{tab:Table1a}-\ref{tab:Table3b}. Broadly speaking, our procedure is very
good on average at estimating $m$ - and better than the best performing
information criterion, BIC -- for all values of $N$ and $T$ (and $\eta $).
This is true across all values of $m$, including the stationary case ($m=0$)
and the no-cointegration case ($m=N$). Information criteria seem to perform
marginally better when $m=0$, but this is more than offset when considering
that they tend to overestimate $m$ in general, especially so when $m=N$ and $%
m=N-1$. When errors are serially correlated (see Tables \ref{tab:Table2a}-%
\ref{tab:Table3b}), results are affected, albeit marginally, but the
relative performance of the various methods remains as described above.
Finally, in order to evaluate our procedure for medium-large systems, in
Tables \ref{tab:Table4}-\ref{tab:Table6} in the Supplement we report results
from a further experiment where we set $N\in \left\{ 10,15,20\right\} $. We
find that our estimator delivers a superior performance with respect to
information criteria for $m=N $ and $m=N-1$, especially for small values of $%
\eta $ (when BIC\ grossly understates $m$, even for large values of $T$);
BIC is (albeit marginally) superior when $m=0$.

\section{Real data examples\label{empirics}}

We illustrate our methodology through four empirical applications to:
commodity prices (Section \ref{comm}),\ U.S. interest rate data (Section \ref%
{rate}), long run purchasing power parity (Section \ref{ppp}) and
cryptocurrency markets (Section \ref{crypto}). In these applications, $N$
ranges from $3$ to $7$.

\subsection{Comovements among commodity prices\label{comm}}

We consider a set of $N=7$ commodity prices: three oil prices (WTI, Brent
crude, and Dubai crude) and the prices of four metals (copper, gold, nickel,
and cobalt). The presence of common trends can be anticipated due to global
demand factors (e.g. growth in emerging Asian countries and especially in
China; or changes of preferences towards greener energy sources, which
increase demand for copper and decrease demand for oil), and also due to
global supply factors (e.g. related to the effect that oil prices have on
transportation costs of other commodities; or driven by technological
innovations which often require the use of cobalt -- see e.g. %
\citealp{alquist}). Moreover, the three oil prices should exhibit strong
comovements, and similarly should the prices of metals, which are often used
in combination in industry (e.g. copper and nickel). In order to study the
presence of such common trends, we use a dataset consisting of monthly data
from January 1990 to March 2021, corresponding to a sample of $T=373$
monthly observations.\footnote{%
Data have been downloaded from
https://www.imf.org/en/Research/commodity-prices} We use the logs of the
data, which are subsequently demeaned and detrended.

We have applied our methodology using the same specifications as described
in Section \ref{montecarlo}, i.e. $\kappa =10^{-4}$, $M=100$ and $n_{S}=2$
in (\ref{upsilon-feasible}). In order to assess robustness to these
specifications, we have also considered other values of $M$ (including $M=T$%
) and $n_{S}=4$.

\begin{table}[h!]
\centering
\begin{threeparttable}
{\tiny
\caption{Estimated number of common trends; whole dataset} 

\label{tab:TableComm1}
\par

\begin{tabular}{llllllllllll}
\hline\hline
&  &  &  &  &  &  &  &  &  &  &  \\ 
\multicolumn{12}{c}{\textbf{Results and sensitivity analysis}} \\ 
&  &  &  &  &  &  &  &  &  &  &  \\ 
\textit{Commodity} &  & \multicolumn{2}{l}{\textit{Tail index}} &  & \textit{%
Test }$H_{0}:E\left\vert X\right\vert ^{2}=\infty $ & \multicolumn{1}{|l}{}
& \textit{nominal level} &  & \multicolumn{1}{c}{$1\%$} & \multicolumn{1}{c}{%
$5\%$} & \multicolumn{1}{c}{$10\%$} \\ 
&  & \multicolumn{2}{l}{} &  &  & \multicolumn{1}{|l}{} &  &  &  &  &  \\ 
Copper &  & \multicolumn{2}{c}{$\underset{\left( 1.144,2.171\right) }{1.658}$%
} & \multicolumn{1}{c}{} & \multicolumn{1}{c}{$\underset{\left( \text{do not
reject }H_{0}\right) }{0.9527}$} & \multicolumn{1}{|l}{} & Johansen's trace
test &  & \multicolumn{1}{c}{$1$} & \multicolumn{1}{c}{$2$} & 
\multicolumn{1}{c}{$2$} \\ 
Gold &  & \multicolumn{2}{c}{$\underset{\left( 1.090,2.069\right) }{1.580}$}
& \multicolumn{1}{c}{} & \multicolumn{1}{c}{$\underset{\left( \text{do not
reject }H_{0}\right) }{0.9525}$} & \multicolumn{1}{|l}{} &  &  &  &  & 
\multicolumn{1}{c}{} \\ 
Brent crude &  & \multicolumn{2}{c}{$\underset{\left( 2.050,3.893\right) }{%
2.972}$} & \multicolumn{1}{c}{} & \multicolumn{1}{c}{$\underset{\left( \text{%
do not reject }H_{0}\right) }{0.9504}$} & \multicolumn{1}{|l}{} & Johansen's 
$\lambda _{\max }$ test &  & \multicolumn{1}{c}{$0$} & \multicolumn{1}{c}{$0$%
} & \multicolumn{1}{c}{$1$} \\ 
Dubai crude &  & \multicolumn{2}{c}{$\underset{\left( 1.171,3.252\right) }{%
2.483}$} & \multicolumn{1}{c}{} & \multicolumn{1}{c}{$\underset{\left( \text{%
do not reject }H_{0}\right) }{0.9499}$} & \multicolumn{1}{|l}{} &  &  & 
\multicolumn{1}{c}{} & \multicolumn{1}{c}{} &  \\ 
Nickel &  & \multicolumn{2}{c}{$\underset{\left( 1.423,2.702\right) }{2.063}$%
} & \multicolumn{1}{c}{} & \multicolumn{1}{c}{$\underset{\left( \text{do not
reject }H_{0}\right) }{0.9548}$} & \multicolumn{1}{|l}{} &  &  &  &  &  \\ 
WTI\ crude &  & \multicolumn{2}{c}{$\underset{\left( 1.747,3.316\right) }{%
2.532}$} & \multicolumn{1}{c}{} & \multicolumn{1}{c}{$\underset{\left( \text{%
do not reject }H_{0}\right) }{0.9502}$} & \multicolumn{1}{|l}{} &  &  &  & 
&  \\ 
Cobalt &  & \multicolumn{2}{c}{$\underset{\left( 1.166,2.215\right) }{1.691}$%
} & \multicolumn{1}{c}{} & \multicolumn{1}{c}{$\underset{\left( \text{do not
reject }H_{0}\right) }{0.9504}$} & \multicolumn{1}{|l}{} &  &  &  &  &  \\ 
&  &  &  &  &  & \multicolumn{1}{|l}{} &  &  &  &  &  \\ \hline
&  &  &  &  &  &  &  &  &  &  &  \\ 
$\left( \kappa =10^{-4},n_{S}=2\right) $ &  & \multicolumn{1}{c}{$\frac{0.05%
}{T}$} & \multicolumn{1}{c}{$\frac{0.05}{\ln T}$} & \multicolumn{1}{c}{$%
\frac{0.05}{N}$} &  & \multicolumn{1}{|l}{} & $\left( \kappa
=10^{-2},n_{S}=2\right) $ &  & \multicolumn{1}{c}{$\frac{0.05}{T}$} & 
\multicolumn{1}{c}{$\frac{0.05}{\ln T}$} & \multicolumn{1}{c}{$\frac{0.05}{N}
$} \\ 
&  &  &  &  &  & \multicolumn{1}{|l}{} &  &  &  &  &  \\ 
$T/2$ &  & \multicolumn{1}{c}{$5$} & \multicolumn{1}{c}{$4$} & 
\multicolumn{1}{c}{$4$} &  & \multicolumn{1}{|l}{} & $T/2$ &  & 
\multicolumn{1}{c}{$5$} & \multicolumn{1}{c}{$4$} & \multicolumn{1}{c}{$4$}
\\ 
$T$ &  & \multicolumn{1}{c}{$4$} & \multicolumn{1}{c}{$4$} & 
\multicolumn{1}{c}{$4$} &  & \multicolumn{1}{|l}{} & $T$ &  & 
\multicolumn{1}{c}{$4$} & \multicolumn{1}{c}{$4$} & \multicolumn{1}{c}{$4$}
\\ 
$2T$ &  & \multicolumn{1}{c}{$4$} & \multicolumn{1}{c}{$4$} & 
\multicolumn{1}{c}{$4$} &  & \multicolumn{1}{|l}{} & $2T$ &  & 
\multicolumn{1}{c}{$4$} & \multicolumn{1}{c}{$4$} & \multicolumn{1}{c}{$4$}
\\ 
&  &  &  &  &  & \multicolumn{1}{|l}{} &  &  &  &  &  \\ \hline
&  &  &  &  &  &  &  &  &  &  &  \\ 
$\left( \kappa =10^{-4},n_{S}=4\right) $ &  & \multicolumn{1}{c}{$\frac{0.05%
}{T}$} & \multicolumn{1}{c}{$\frac{0.05}{\ln T}$} & \multicolumn{1}{c}{$%
\frac{0.05}{N}$} &  & \multicolumn{1}{|l}{} & $\left( \kappa
=10^{-2},n_{S}=4\right) $ &  & \multicolumn{1}{c}{$\frac{0.05}{T}$} & 
\multicolumn{1}{c}{$\frac{0.05}{\ln T}$} & \multicolumn{1}{c}{$\frac{0.05}{N}
$} \\ 
&  &  &  &  &  & \multicolumn{1}{|l}{} &  &  &  &  &  \\ 
$T/2$ &  & \multicolumn{1}{c}{$5$} & \multicolumn{1}{c}{$4$} & 
\multicolumn{1}{c}{$4$} &  & \multicolumn{1}{|l}{} & $T/2$ &  & 
\multicolumn{1}{c}{$5$} & \multicolumn{1}{c}{$4$} & \multicolumn{1}{c}{$4$}
\\ 
$T$ &  & \multicolumn{1}{c}{$5$} & \multicolumn{1}{c}{$4$} & 
\multicolumn{1}{c}{$4$} &  & \multicolumn{1}{|l}{} & $T$ &  & 
\multicolumn{1}{c}{$5$} & \multicolumn{1}{c}{$4$} & \multicolumn{1}{c}{$4$}
\\ 
$2T$ &  & \multicolumn{1}{c}{$4$} & \multicolumn{1}{c}{$4$} & 
\multicolumn{1}{c}{$4$} &  & \multicolumn{1}{|l}{} & $2T$ &  & 
\multicolumn{1}{c}{$4$} & \multicolumn{1}{c}{$4$} & \multicolumn{1}{c}{$4$}
\\ 
&  &  &  &  &  & \multicolumn{1}{|l}{} &  &  &  &  &  \\ \hline\hline
\end{tabular}

\smallskip
\begin{tablenotes}

     \item In the top part of the table we report the estimated values of the tail index using the Hill's estimator - the package 'ptsuite' in R has been employed, using a number of order statistics equal to $k_{T}=40$. We also report, in light of Hill's estimator being inconclusive, the outcome of the test by \citet{trapani16}. As far as this is concerned, we have used the modification suggested in Section 3.2 in \citet{HT16}; whilst we refer to that paper for details, essentially this consists in running the test for $S$ times (we have used $S=10,000$), and compute the average number of non-rejections. This is then compared against the threshold $1-\alpha-\sqrt{\left( \alpha\left( 1-\alpha\right)\right)\frac{2\ln \ln S}{S}}$; in our case, such threshold is $0.9454$, with $H_{0}$ being rejected if the average number of non rejections falls below the threshold, and not rejected otherwise. 
     \item In the table (top, right part), we also report the number of cointegration relationships found by Johansen's procedure; this has been implemented using $p=2$ lags in the $VAR$ specification, as suggested using BIC, and constant and trend when implementing the test. 
     \item In the bottom half of the table, we report results on $\widehat{m}$ obtained using different specifications, as written in the table. In particular, in each sub-panel, the columns contain different values of the nominal level of the family-wise procedure, set equal to $\frac{0.05}{T}$, $\frac{0.05}{\ln(T)}$ and $\frac{0.05}{N}$.

\end{tablenotes}
}
\end{threeparttable}
\end{table}

\smallskip

We report the results for the $7$-dimensional series in Table \ref%
{tab:TableComm1}. Initially, we report the (Hill's) estimates of the tail
indices for the seven series; the associated confidence sets are quite
large, but the test by \citet{trapani16} supports the hypothesis that all
series have infinite variance.

Estimation of $m$ based on Johansen's sequential procedure for the
determination of the cointegration rank (using either the trace tests or the
maximum eigenvalue tests) provides ambiguous results, with the estimate of $%
m $ ranging between $5$ and 7 (which corresponds to no cointegration). These
results, in the presence of heavy tails, are not reliable: the simulations
in \citet{caner} show that Johansen's procedure does not control size under
infinite variance, thus making the estimate of $m$ inaccurate (in
particular, it tends to be over-sized, thus leading to overestimation of $m$%
). In contrast, through our test we find strong evidence of $m=4$ common
stochastic trends. As shown in the table, our results are broadly robust to
different values of $M$, $\kappa $. In (much) fewer cases, we find $m=5$,
which might suggest the presence of a slowly mean reverting component in the
data.

In order to shed more light on these findings, we split the series into two
sub-groups: one of dimension $N=3$ (comprising the three crude prices --
Brent, Dubai and WTI crude), and one of dimension $N=4$ (containing the four
metal prices). Results for the $3$-dimensional series of crude prices are in
Table \ref{tab:TableComm2}. On the one hand, Johansen's tests in this case
identifies (at $5\%$ level, and only using the trace test) two common trends
($m=2$). On the other hand, our methodology provides evidence of a single ($%
m=1$) common stochastic trend (and, in some, more rare, cases, of $m=2$).
Results concerning the $N=4$ metals are in Table \ref{tab:TableComm3}; in
this case, evidence of $m=3$ common trends emerges from all the procedures
considered.

\begin{table}[h!]
\centering
\begin{threeparttable}
{\tiny
\caption{Estimated number of common trends; oil prices} 

\label{tab:TableComm2}
\par

\begin{tabular}{llllllllllll}
\hline\hline
&  &  &  &  &  &  &  &  &  &  &  \\ 
\multicolumn{12}{c}{\textbf{Results and sensitivity analysis}} \\ 
&  &  &  &  &  & \multicolumn{1}{|l}{} &  &  &  &  &  \\ 
\textit{nominal level} &  & \multicolumn{1}{c}{$1\%$} & \multicolumn{1}{c}{$%
5\%$} & \multicolumn{1}{c}{$10\%$} &  & \multicolumn{1}{|l}{} &  & 
\multicolumn{4}{l}{Cointegration vectors} \\ 
&  &  &  &  &  & \multicolumn{1}{|l}{} &  &  & \multicolumn{1}{c}{} & 
\multicolumn{1}{c}{} &  \\ 
Johansen's trace test &  & \multicolumn{1}{c}{$2$} & \multicolumn{1}{c}{$2$}
& \multicolumn{1}{c}{$2$} &  & \multicolumn{1}{|l}{} &  & $\widehat{\beta }%
_{1,1}$ & \multicolumn{1}{c}{$1.00$} & $\widehat{\beta }_{2,1}$ & 
\multicolumn{1}{c}{$0.00$} \\ 
&  &  &  & \multicolumn{1}{c}{} &  & \multicolumn{1}{|l}{} &  & $\widehat{%
\beta }_{1,2}$ & \multicolumn{1}{c}{$0.00$} & \multicolumn{1}{c}{$\widehat{%
\beta }_{2,2}$} & \multicolumn{1}{c}{$1.00$} \\ 
Johansen's $\lambda _{\max }$ test &  & \multicolumn{1}{c}{$2$} & 
\multicolumn{1}{c}{$2$} & \multicolumn{1}{c}{$2$} &  & \multicolumn{1}{|l}{}
&  & $\widehat{\beta }_{1,3}$ & \multicolumn{1}{c}{$-1.08$} & $\widehat{%
\beta }_{2,3}$ & \multicolumn{1}{c}{$-1.05$} \\ 
&  &  &  &  &  & \multicolumn{1}{|l}{} &  &  &  &  &  \\ \hline
&  &  &  &  &  & \multicolumn{1}{|l}{} &  &  &  &  &  \\ 
$\left( \kappa =10^{-4},n_{S}=2\right) $ &  & \multicolumn{1}{c}{$\frac{0.05%
}{T}$} & \multicolumn{1}{c}{$\frac{0.05}{\ln T}$} & \multicolumn{1}{c}{$%
\frac{0.05}{N}$} &  & \multicolumn{1}{|l}{} & $\left( \kappa
=10^{-2},n_{S}=2\right) $ &  & \multicolumn{1}{c}{$\frac{0.05}{T}$} & 
\multicolumn{1}{c}{$\frac{0.05}{\ln T}$} & \multicolumn{1}{c}{$\frac{0.05}{N}
$} \\ 
&  &  &  &  &  & \multicolumn{1}{|l}{} &  &  &  &  &  \\ 
$T/2$ &  & \multicolumn{1}{c}{$2$} & \multicolumn{1}{c}{$2$} & 
\multicolumn{1}{c}{$1$} &  & \multicolumn{1}{|l}{} & $T/2$ &  & 
\multicolumn{1}{c}{$2$} & \multicolumn{1}{c}{$1$} & \multicolumn{1}{c}{$1$}
\\ 
$T$ &  & \multicolumn{1}{c}{$2$} & \multicolumn{1}{c}{$1$} & 
\multicolumn{1}{c}{$1$} &  & \multicolumn{1}{|l}{} & $T$ &  & 
\multicolumn{1}{c}{$1$} & \multicolumn{1}{c}{$1$} & \multicolumn{1}{c}{$1$}
\\ 
$2T$ &  & \multicolumn{1}{c}{$1$} & \multicolumn{1}{c}{$1$} & 
\multicolumn{1}{c}{$1$} &  & \multicolumn{1}{|l}{} & $2T$ &  & 
\multicolumn{1}{c}{$1$} & \multicolumn{1}{c}{$1$} & \multicolumn{1}{c}{$1$}
\\ 
&  &  &  &  &  & \multicolumn{1}{|l}{} &  &  &  &  &  \\ \hline
&  &  &  &  &  &  &  &  &  &  &  \\ 
$\left( \kappa =10^{-4},n_{S}=4\right) $ &  & \multicolumn{1}{c}{$\frac{0.05%
}{T}$} & \multicolumn{1}{c}{$\frac{0.05}{\ln T}$} & \multicolumn{1}{c}{$%
\frac{0.05}{N}$} &  & \multicolumn{1}{|l}{} & $\left( \kappa
=10^{-2},n_{S}=4\right) $ &  & \multicolumn{1}{c}{$\frac{0.05}{T}$} & 
\multicolumn{1}{c}{$\frac{0.05}{\ln T}$} & \multicolumn{1}{c}{$\frac{0.05}{N}
$} \\ 
&  &  &  &  &  & \multicolumn{1}{|l}{} &  &  &  &  &  \\ 
$T/2$ &  & \multicolumn{1}{c}{$2$} & \multicolumn{1}{c}{$2$} & 
\multicolumn{1}{c}{$1$} &  & \multicolumn{1}{|l}{} & $T/2$ &  & 
\multicolumn{1}{c}{$2$} & \multicolumn{1}{c}{$2$} & \multicolumn{1}{c}{$1$}
\\ 
$T$ &  & \multicolumn{1}{c}{$2$} & \multicolumn{1}{c}{$1$} & 
\multicolumn{1}{c}{$1$} &  & \multicolumn{1}{|l}{} & $T$ &  & 
\multicolumn{1}{c}{$2$} & \multicolumn{1}{c}{$1$} & \multicolumn{1}{c}{$1$}
\\ 
$2T$ &  & \multicolumn{1}{c}{$1$} & \multicolumn{1}{c}{$1$} & 
\multicolumn{1}{c}{$1$} &  & \multicolumn{1}{|l}{} & $2T$ &  & 
\multicolumn{1}{c}{$1$} & \multicolumn{1}{c}{$1$} & \multicolumn{1}{c}{$1$}
\\ 
&  &  &  &  &  & \multicolumn{1}{|l}{} &  &  &  &  &  \\ \hline\hline
\end{tabular}

\smallskip
\begin{tablenotes}

     \item See Table \ref{tab:TableComm1} for details; the three series considered are the three crude prices: WTI, Brent, and Dubai.

\end{tablenotes}
}
\end{threeparttable}
\end{table}

\smallskip

\begin{table}[h!]
\centering
\begin{threeparttable}
{\tiny
\caption{Estimated number of common trends; metal prices} 

\label{tab:TableComm3}
\par

\begin{tabular}{llllllllllll}
\hline\hline
&  &  &  &  &  &  &  &  &  &  &  \\ 
\multicolumn{12}{c}{\textbf{Results and sensitivity analysis}} \\ 
&  &  &  &  &  & \multicolumn{1}{|l}{} &  & \multicolumn{4}{l}{Cointegration
vectors} \\ 
\textit{nominal level} &  & \multicolumn{1}{c}{$1\%$} & \multicolumn{1}{c}{$%
5\%$} & \multicolumn{1}{c}{$10\%$} &  & \multicolumn{1}{|l}{} &  &  & 
\multicolumn{1}{c}{} & \multicolumn{1}{c}{} & \multicolumn{1}{c}{} \\ 
&  &  &  &  &  & \multicolumn{1}{|l}{} &  & $\widehat{\beta }_{1,1}$ & 
\multicolumn{1}{c}{$1.00$} & \multicolumn{1}{c}{} & \multicolumn{1}{c}{} \\ 
Johansen's trace test &  & \multicolumn{1}{c}{$0$} & \multicolumn{1}{c}{$1$}
& \multicolumn{1}{c}{$1$} &  & \multicolumn{1}{|l}{} &  & $\widehat{\beta }%
_{1,2}$ & \multicolumn{1}{c}{$-0.69$} & \multicolumn{1}{c}{} & 
\multicolumn{1}{c}{} \\ 
&  &  &  & \multicolumn{1}{c}{} &  & \multicolumn{1}{|l}{} &  & $\widehat{%
\beta }_{1,3}$ & \multicolumn{1}{c}{$-0.50$} & \multicolumn{1}{c}{} & 
\multicolumn{1}{c}{} \\ 
Johansen's $\lambda _{\max }$ test &  & \multicolumn{1}{c}{$0$} & 
\multicolumn{1}{c}{$1$} & \multicolumn{1}{c}{$1$} &  & \multicolumn{1}{|l}{}
&  & $\widehat{\beta }_{1,4}$ & \multicolumn{1}{c}{$-0.13$} & 
\multicolumn{1}{c}{} & \multicolumn{1}{c}{} \\ 
&  &  &  &  &  & \multicolumn{1}{|l}{} &  &  & \multicolumn{1}{c}{} & 
\multicolumn{1}{c}{} & \multicolumn{1}{c}{} \\ \hline
&  &  &  &  &  & \multicolumn{1}{|l}{} &  &  &  &  &  \\ 
$\left( \kappa =10^{-4},n_{S}=2\right) $ &  & \multicolumn{1}{c}{$\frac{0.05%
}{T}$} & \multicolumn{1}{c}{$\frac{0.05}{\ln T}$} & \multicolumn{1}{c}{$%
\frac{0.05}{N}$} &  & \multicolumn{1}{|l}{} & $\left( \kappa
=10^{-2},n_{S}=2\right) $ &  & \multicolumn{1}{c}{$\frac{0.05}{T}$} & 
\multicolumn{1}{c}{$\frac{0.05}{\ln T}$} & \multicolumn{1}{c}{$\frac{0.05}{N}
$} \\ 
&  &  &  &  &  & \multicolumn{1}{|l}{} &  &  &  &  &  \\ 
$T/2$ &  & \multicolumn{1}{c}{$1$} & \multicolumn{1}{c}{$3$} & 
\multicolumn{1}{c}{$3$} &  & \multicolumn{1}{|l}{} & $T/2$ &  & 
\multicolumn{1}{c}{$3$} & \multicolumn{1}{c}{$3$} & \multicolumn{1}{c}{$3$}
\\ 
$T$ &  & \multicolumn{1}{c}{$3$} & \multicolumn{1}{c}{$3$} & 
\multicolumn{1}{c}{$3$} &  & \multicolumn{1}{|l}{} & $T$ &  & 
\multicolumn{1}{c}{$3$} & \multicolumn{1}{c}{$3$} & \multicolumn{1}{c}{$3$}
\\ 
$2T$ &  & \multicolumn{1}{c}{$3$} & \multicolumn{1}{c}{$3$} & 
\multicolumn{1}{c}{$3$} &  & \multicolumn{1}{|l}{} & $2T$ &  & 
\multicolumn{1}{c}{$3$} & \multicolumn{1}{c}{$3$} & \multicolumn{1}{c}{$3$}
\\ 
&  &  &  &  &  & \multicolumn{1}{|l}{} &  &  &  &  &  \\ \hline
&  &  &  &  &  &  &  &  &  &  &  \\ 
$\left( \kappa =10^{-4},n_{S}=4\right) $ &  & \multicolumn{1}{c}{$\frac{0.05%
}{T}$} & \multicolumn{1}{c}{$\frac{0.05}{\ln T}$} & \multicolumn{1}{c}{$%
\frac{0.05}{N}$} &  & \multicolumn{1}{|l}{} & $\left( \kappa
=10^{-2},n_{S}=4\right) $ &  & \multicolumn{1}{c}{$\frac{0.05}{T}$} & 
\multicolumn{1}{c}{$\frac{0.05}{\ln T}$} & \multicolumn{1}{c}{$\frac{0.05}{N}
$} \\ 
&  &  &  &  &  & \multicolumn{1}{|l}{} &  &  &  &  &  \\ 
$T/2$ &  & \multicolumn{1}{c}{$3$} & \multicolumn{1}{c}{$3$} & 
\multicolumn{1}{c}{$3$} &  & \multicolumn{1}{|l}{} & $T/2$ &  & 
\multicolumn{1}{c}{$3$} & \multicolumn{1}{c}{$3$} & \multicolumn{1}{c}{$3$}
\\ 
$T$ &  & \multicolumn{1}{c}{$3$} & \multicolumn{1}{c}{$3$} & 
\multicolumn{1}{c}{$3$} &  & \multicolumn{1}{|l}{} & $T$ &  & 
\multicolumn{1}{c}{$3$} & \multicolumn{1}{c}{$3$} & \multicolumn{1}{c}{$3$}
\\ 
$2T$ &  & \multicolumn{1}{c}{$3$} & \multicolumn{1}{c}{$3$} & 
\multicolumn{1}{c}{$3$} &  & \multicolumn{1}{|l}{} & $2T$ &  & 
\multicolumn{1}{c}{$3$} & \multicolumn{1}{c}{$3$} & \multicolumn{1}{c}{$3$}
\\ 
&  &  &  &  &  & \multicolumn{1}{|l}{} &  &  &  &  &  \\ \hline\hline
\end{tabular}

\smallskip
\begin{tablenotes}

     \item See Table \ref{tab:TableComm1} for details; the four series considered are the four metals: copper, nickel, gold, and cobalt.

\end{tablenotes}
}
\end{threeparttable}
\end{table}

Overall, most of the evidence points towards $m=4$ common stochastic trends;
there is much less evidence in support of $m=5$. The $m=4$ common trends can
be estimated as explained in Section \ref{esttrend}. In order to identify
the trends, based on the results above we propose to order the series as
follows: WTI, gold, cobalt, copper, Brent crude, Dubai crude, nickel. Then,
we constrain the upper $m\times m$ block of the estimated loadings matrix $%
\widehat{P}$ to be the identity matrix. In this way the trends are
identified with the first four series, i.e. WTI, gold, cobalt, and copper.
In Figure \ref{fig:FigComm}, we plot the estimated common trends $\widehat{x}%
_{t}$ and in Table \ref{tab:load} we report the associated loadings $%
\widehat{P}$.

\begin{table}[h!]
\caption{Commodity prices - Estimated loadings $\widehat P$}
\label{tab:load}\centering
\par
{\tiny 
\begin{tabular}{l|cccc}
\hline\hline
&  &  &  &  \\ 
& $\widehat x_{1,t}$ & $\widehat x_{2,t}$ & $\widehat x_{3,t}$ & $\widehat
x_{4,t}$ \\ 
&  &  &  &  \\ \hline
&  &  &  &  \\ 
WTI & 1 & 0 & 0 & 0 \\ 
Cobalt & 0 & 1 & 0 & 0 \\ 
Gold & 0 & 0 & 1 & 0 \\ 
Copper & 0 & 0 & 0 & 1 \\ 
Brent Crude & 1.0409 & -0.0252 & 0.1232 & -0.0292 \\ 
Dubai crude & 1.0429 & -0.0164 & 0.1515 & -0.0679 \\ 
Nickel & -0.2133 & -0.3436 & -1.8855 & 2.5744 \\ 
&  &  &  &  \\ \hline\hline
\end{tabular}
}
\end{table}

By construction, the first and second trends ($\widehat x_{1,t}$ and $%
\widehat x_{2,t}$) are associated with oil prices and cobalt repsectively.
The third one ($\widehat x_{3,t}$) is associated with gold (by
construction), and nickel, with a negative loading; finally, the fourth
trend ($\widehat x_{4,t}$) is associate with copper by construction, and
with nickel (with a positive loading). The trends driving metals are also
common to oil prices, albeit with smaller loadings.

\subsection{The term structure of US interest rates\label{rate}}

In this example, we illustrate our methodology through a small scale
application based on the same dataset as in \citet{ling}, consisting of $N=3 
$ monthly, seasonally unadjusted series: the 3-month Treasury Bill rate, the
1-year Treasury Bill rate (both from sales on the secondary market) and the
Effective Federal Fund rate, spanning the period between February 1974 and
February 1999 (thus corresponding to $T=301$ monthly observations).\footnote{%
The data have been downloaded from the Federal Reserve Economic Data
website, https://fred.stlouisfed.org}

As in \citet{ling}, we consider the logs of the original data, say $%
Y_{t}=\left( y_{1,t},y_{2,t},y_{3,t}\right) ^{\prime }$ where $y_{1,t}$ is
the log of the 3-month Treasury Bill rate, $y_{2,t}$ is the log of the
1-year Treasury Bill rate, and $y_{3,t}$ represents the log of the Fed Fund
rate. \citet{ling} carry out their analysis using a VAR(2)\ in error
correction format:%
\begin{equation}
\Delta Y_{t}=\mu +\alpha \beta ^{\prime }Y_{t-1}+\Phi \Delta
Y_{t-1}+\varepsilon _{t}\text{,}  \label{ecm-ling}
\end{equation}%
with $\mu =\alpha \rho$, such that the data do not have a linear trend but
may have a non-zero mean in the stationary directions. In the light of
Corollary \ref{constant}, in this case our set-up can be applied with no
modifications. Moreover, we do not need to make specific assumptions on the
lag structure for $\Delta Y_{t}$.

As demonstrated above, the main feature of our contribution is that we do
not require any prior knowledge of the index $\eta $. This advantage can be
further understood upon observing the estimates of $\eta $ obtained in \citet%
*{ling} -- see in particular, their Figure 6, where $\eta $ ranges between $%
1 $ and $1.5$ (thereby implying that the data have infinite variance and,
possibly, infinite mean). On account of this information, \citet{ling}
report the critical values for the Likelihood Ratio tests in \citet{caner}
for $\eta =1 $ and $\eta =1.5$, using both values to build a decision rule.%

\begin{table}[h]
\centering
\begin{threeparttable}
{\tiny
\caption{Estimated number of common trends} 

\label{tab:TableLing}
\par

\begin{tabular}{llllllllll}
\hline\hline
&  &  &  &  &  &  &  &  &  \\ 
\multicolumn{10}{c}{\textbf{Panel A: results}} \\ [0.25cm]
&  & \multicolumn{1}{c}{$\widehat{m}$} &  & \multicolumn{6}{l}{\textit{Notes}
} \\ 
&  & \multicolumn{1}{c}{} &  & \multicolumn{6}{l}{} \\ 
BCT &  & \multicolumn{1}{c}{$2$} &  & \multicolumn{6}{l}{} \\ 
\citet{ling} &  & \multicolumn{1}{c}{$1$} &  & \multicolumn{6}{l}{($\alpha
=5\%$; $\eta =1$)} \\ 
\citet{ling} &  & \multicolumn{1}{c}{$1$} &  & \multicolumn{6}{l}{($\alpha
=5\%$; $\eta =1.5$)} \\ 
\citet{ling} &  & \multicolumn{1}{c}{$3$} &  & \multicolumn{6}{l}{($\alpha
=1\%$; $\eta =1$)} \\ 
\citet{ling} &  & \multicolumn{1}{c}{$1$} &  & \multicolumn{6}{l}{($\alpha
=1\%$; $\eta =1.5$)} \\ 
Johansen's trace test &  & \multicolumn{1}{c}{$1$} &  & \multicolumn{6}{l}{($%
\alpha =1\%$)} \\ 
Johansen's $\lambda _{\max }$ test &  & \multicolumn{1}{c}{$1$} &  & 
\multicolumn{6}{l}{($\alpha =1\%$)}
\\ 
&  &  &  & \multicolumn{6}{l}{} \\ \hline\hline
&  &  &  &  &  &  &  &  &  \\ 
\multicolumn{10}{c}{\textbf{Panel B: sensitivity analysis}} \\ [0.25cm] 
&  &  &  &  & \multicolumn{1}{|l}{} &  &  &  &  \\ 
$\left( \kappa =10^{-4};n_{S}=2\right) $ &  & \multicolumn{1}{c}{$\frac{0.05%
}{T}$} & \multicolumn{1}{c}{$\frac{0.05}{\ln T}$} & \multicolumn{1}{c}{$%
\frac{0.05}{N}$} & \multicolumn{1}{|l}{$\left( \kappa
=10^{-2};n_{S}=2\right) $} &  & \multicolumn{1}{c}{$\frac{0.05}{T}$} & 
\multicolumn{1}{c}{$\frac{0.05}{\ln T}$} & \multicolumn{1}{c}{$\frac{0.05}{N}
$} \\ 
$M$ &  & \multicolumn{1}{c}{} & \multicolumn{1}{c}{} & \multicolumn{1}{c}{}
& \multicolumn{1}{|l}{$M$} &  & \multicolumn{1}{c}{} & \multicolumn{1}{c}{}
& \multicolumn{1}{c}{} \\ 
$100$ &  & \multicolumn{1}{c}{$2$} & \multicolumn{1}{c}{$2$} & 
\multicolumn{1}{c}{$2$} & \multicolumn{1}{|l}{$100$} &  & \multicolumn{1}{c}{%
$2$} & \multicolumn{1}{c}{$2$} & \multicolumn{1}{c}{$2$} \\ 
$T/2$ &  & \multicolumn{1}{c}{$2$} & \multicolumn{1}{c}{$2$} & 
\multicolumn{1}{c}{$2$} & \multicolumn{1}{|l}{$T/2$} &  & \multicolumn{1}{c}{%
$2$} & \multicolumn{1}{c}{$2$} & \multicolumn{1}{c}{$2$} \\ 
$T$ &  & \multicolumn{1}{c}{$2$} & \multicolumn{1}{c}{$2$} & 
\multicolumn{1}{c}{$2$} & \multicolumn{1}{|l}{$T$} &  & \multicolumn{1}{c}{$2
$} & \multicolumn{1}{c}{$2$} & \multicolumn{1}{c}{$2$} \\ 
$2T$ &  & \multicolumn{1}{c}{$2$} & \multicolumn{1}{c}{$2$} & 
\multicolumn{1}{c}{$2$} & \multicolumn{1}{|l}{$2T$} &  & \multicolumn{1}{c}{$%
2$} & \multicolumn{1}{c}{$2$} & \multicolumn{1}{c}{$2$} \\ 
&  & \multicolumn{1}{c}{} & \multicolumn{1}{c}{} & \multicolumn{1}{c}{} & 
\multicolumn{1}{|l}{} &  & \multicolumn{1}{c}{} & \multicolumn{1}{c}{} & 
\multicolumn{1}{c}{} \\ \hline
&  & \multicolumn{1}{c}{} & \multicolumn{1}{c}{} & \multicolumn{1}{c}{} & 
\multicolumn{1}{|l}{} &  & \multicolumn{1}{c}{} & \multicolumn{1}{c}{} & 
\multicolumn{1}{c}{} \\ 
$\left( \kappa =10^{-4};n_{S}=4\right) $ &  & \multicolumn{1}{c}{$\frac{0.05%
}{T}$} & \multicolumn{1}{c}{$\frac{0.05}{\ln T}$} & \multicolumn{1}{c}{$%
\frac{0.05}{N}$} & \multicolumn{1}{|l}{$\left( \kappa
=10^{-2};n_{S}=4\right) $} &  & \multicolumn{1}{c}{$\frac{0.05}{T}$} & 
\multicolumn{1}{c}{$\frac{0.05}{\ln T}$} & \multicolumn{1}{c}{$\frac{0.05}{N}
$} \\ 
$M$ &  & \multicolumn{1}{c}{} & \multicolumn{1}{c}{} & \multicolumn{1}{c}{}
& \multicolumn{1}{|l}{$M$} &  & \multicolumn{1}{c}{} & \multicolumn{1}{c}{}
& \multicolumn{1}{c}{} \\ 
$100$ &  & \multicolumn{1}{c}{$2$} & \multicolumn{1}{c}{$2$} & 
\multicolumn{1}{c}{$2$} & \multicolumn{1}{|l}{$100$} &  & \multicolumn{1}{c}{%
$2$} & \multicolumn{1}{c}{$2$} & \multicolumn{1}{c}{$2$} \\ 
$T/2$ &  & \multicolumn{1}{c}{$2$} & \multicolumn{1}{c}{$2$} & 
\multicolumn{1}{c}{$2$} & \multicolumn{1}{|l}{$T/2$} &  & \multicolumn{1}{c}{%
$2$} & \multicolumn{1}{c}{$2$} & \multicolumn{1}{c}{$2$} \\ 
$T$ &  & \multicolumn{1}{c}{$2$} & \multicolumn{1}{c}{$2$} & 
\multicolumn{1}{c}{$2$} & \multicolumn{1}{|l}{$T$} &  & \multicolumn{1}{c}{$2
$} & \multicolumn{1}{c}{$2$} & \multicolumn{1}{c}{$2$} \\ 
$2T$ &  & \multicolumn{1}{c}{$2$} & \multicolumn{1}{c}{$2$} & 
\multicolumn{1}{c}{$2$} & \multicolumn{1}{|l}{$2T$} &  & \multicolumn{1}{c}{$%
2$} & \multicolumn{1}{c}{$2$} & \multicolumn{1}{c}{$2$} \\ 
&  & \multicolumn{1}{c}{} & \multicolumn{1}{c}{} & \multicolumn{1}{c}{} & 
\multicolumn{1}{|l}{} &  & \multicolumn{1}{c}{} & \multicolumn{1}{c}{} & 
\multicolumn{1}{c}{} \\ \hline\hline
\end{tabular}

\smallskip
\begin{tablenotes}

     \item The top part of the table summarizes the findings using various procedures. The BCT procedure has been implemented with $M=100$, nominal level for individual tests equal to $\frac{0.05}{T}$, $\kappa=10^{-4}$ and $n_{S}=4$; see also Panel B of the table for results obtained using different specifications. 
     \item In Panel A of the table, the results based on \citet{ling} are taken from Table 4 in their paper. The results corresponding to Johansen's procedure (\citealp{johansen1991}) have been derived setting $p=4$ in the specification of the $VAR$ based on BIC.  
     \item We report the outcome of our test under various specifications in Panel B  of the table, similarly to \ref{tab:TableComm1}.

\end{tablenotes}
}
\end{threeparttable}
\end{table}

In particular, \citet{ling} conclude that there is some evidence in favour
of $m=1$ (corresponding to two cointegration relationships). Upon inspecting
their results, however, it seems that they are driven by the choice of the $%
5\%$ nominal level: evidence in favour of $m=1$ is much weaker when
considering a $1\%$ nominal level (in particular, in the latter case, should
one use the critical values computed for $\eta =1$, there would be no
evidence of cointegration at all, whereas using the critical values computed
for the case $\eta =1.5$, one would conclude that $m=1$). Interestingly,
assuming finite variance, Johansen's sequential procedure rejects $m=2$ and $%
m=3$ in favour of $m=1$, even at the $1\%$ nominal level. 

In contrast to these mixed results, we find $\widehat{m}=2$ common
stochastic trends, corresponding to one cointegration relationship. This
result is robust to different values of $M$ and $\kappa $. 

\subsection{Testing for the weak form of the PPP\label{ppp}}

In this application, we check the validity of the weak form of the PPP,
using the same dataset as in \citet{falk}. In particular, for a set of 12
countries, we individually verify the presence of cointegration in the
vectors $\left( p_{c,t},p_{US,t},FX_{c,t}\right) ^{\prime }$, where $p_{c,t}$
is the (log of the) CPI\ of country $c$ at time $t$, $p_{US,t}$ is the (log
of the) US CPI index, and $FX_{c,t}$ is the log of the exchange rate of
country $c$ currency vis-a-vis the dollar. The series are monthly, spanning
from January 1973 to December 1999, corresponding to $T=324$ for each
country.

A minimum requirement for the weak form of the PPP to hold is that there is
at least one cointegrating relation between the three series -- i.e., at
most $m=2$ common stochastic trends. \citet{falk} find evidence of heavy
tails in this dataset (see their Table III), with series typically having
finite mean but infinite variance. Thence, the authors test for
cointegration using critical values based on allowing for heavy tails
(notice that, even in this case, the authors require prior knowledge of the
tail index $\eta $). We compare results from \citet{falk} and Johansen's
procedure with our proposed tests. As in the previous sections we use $M=T$, 
$\kappa =10^{-2}$; the level for the individual tests is set to $0.05/\ln T$%
. Tests are based on the deviations from the initial values. Results are
reported in Table \ref{tab:TableFalk}.

\bigskip

\begin{table}[h!]
\centering
\begin{threeparttable}
{\tiny
\caption{Inference on the weak form of the PPP} 
\label{tab:TableFalk}
\par
\begin{tabular}{llccccccccc}
\hline\hline
&  & \multicolumn{2}{c}{\textit{FW - Gaussian}} & \multicolumn{2}{c}{\textit{%
FW - Gaussian}} & \multicolumn{2}{c}{\textit{FW - heavy tailed}} & 
\multicolumn{2}{c}{\textit{FW - heavy tailed}} & \textit{BCT} \\ 
&  & \multicolumn{2}{c}{trace test} & \multicolumn{2}{c}{$\lambda _{\max }$
test} & \multicolumn{2}{c}{trace test} & \multicolumn{2}{c}{$\lambda _{\max }
$ test} &  \\ 
\textbf{Country} &  & \textit{5\% level} & \textit{10\% level} & \textit{5\%
level} & \textit{10\% level} & \textit{5\% level} & \textit{10\% level} & 
\textit{5\% level} & \textit{10\% level} &  \\ 
&  &  &  &  &  &  &  &  &  &  \\ \hline
&  &  &  &  &  &  &  &  &  &  \\ 
Belgium &  & Y & Y & Y & Y & Y & Y & Y & Y & N \\ 
Canada &  & N & N & N & N & N & N & N & N & Y \\ 
Denmark &  & Y & Y & Y & Y & N & Y & N & N & N \\ 
France &  & Y & Y & Y & Y & N & Y & N & Y & Y \\ 
Germany &  & N & N & N & N & N & N & N & N & N \\ 
Italy &  & N & Y & N & Y & N & N & N & N & Y \\ 
Japan &  & Y & Y & Y & Y & Y & Y & Y & Y & Y \\ 
Netherlands &  & Y & Y & Y & Y & Y & Y & Y & Y & N \\ 
Norway &  & Y & Y & N & N & N & Y & N & N & N \\ 
Spain &  & Y & Y & Y & Y & N & Y & N & N & N \\ 
Sweden &  & N & Y & N & N & N & N & N & N & N \\ 
UK &  & Y & Y & Y & Y & Y & Y & N & Y & Y \\ 
&  &  &  &  &  &  &  &  &  &  \\ \hline\hline
\end{tabular}

\smallskip
\begin{tablenotes}

     \item We indicate with Y cases where tests support the necessary condition for PPP to hold, i.e. $m \leq 2$. Cases with no evidence supporting PPP, corresponding to $m = 3$, are denoted as N. Results corresponding to the paper by \citet{falk} are taken from their Tables VII (for the Gaussian case) and VIII (for the case with heavy tails). The test statistics employed are the maximum eigenvalue ($\lambda_{max}$ column) and the trace test, as developed by \citet{johansen1991} for the Gaussian case. We refer to Table V in \citet{falk} for the critical values for both tests computed under the assumption of heavy-tailed data,

\end{tablenotes}
}
\end{threeparttable}
\end{table}

First, we notice that Johansen's method at the $10\%$ ($5\%$) nominal level
under the assumptions of finite variance indicates that the necessary
condition for the weak form of the PPP holds for 10 (8) out of 12 countries.
Second, if the critical values are adjusted using the estimated tail index,
as in \citet{falk}, then the support for the weak from of the PPP decreases.
Specifically, they found that the number of countries with $\widehat{m}\leq
2 $ ranges between $3$ (when using a maximum eigenvalue test at the $5\% $
nominal level) and $8$ (when using a trace test at the $10\%$ nominal
level). Our test provides support of PPP in $5$ countries out of 12, which
is closely in line with \citet{falk}. Our empirical evidence reinforces the
view that the assumption of finite variance can lead to spurious support to
the PPP.

\subsection{Statistical arbitrage in cryptocurrency markets\label{crypto}}

In the context of asset pricing, since the seminal contribution by %
\citet{diebold1994cointegration}, common trends and cointegration have
played a prominent role. Indeed, the presence of cointegration entails the
possibility of constructing portfolios which are mean reverting, so that a
profit can be made when such portfolios depart from their mean -- the
so-called \textquotedblleft statistical arbitrage\textquotedblright .
Although the literature usually focuses on portfolios of two assets
(\textquotedblleft pairs trading\textquotedblright ), it is possible to
generalise the notion of statistical arbitrage to the multi-asset case (see %
\citealp{alexander}). In such a case, having $m$ common stochastic trends
allows to construct $N-m$ mean reverting portfolios, each having weights
given by a user-chosen linear combination of the cointegrating vectors.
Cryptocurrency data have been paid increasing attention (\citealp{makarov});
however, despite the evidence that returns on cryptocurrencies exhibit heavy
tails (see \citealp{pele}), this is not usually accounted for in
applications.

We consider a portfolio of $N=6$ cryptocurrencies, chosen among those with
the largest market capitalisation within our sample period: Cardano (ADA),
Bitcoin Cash (BCH), BitCoin (BTC), EOS, Ethereum (ETH), and XRP.\footnote{%
Data have been downloaded from https://www.coingecko.com/} We use daily data
from October 18, 2017 until September 18, 2020, which is equivalent to a
sample of $T=1,066$ observations. As in the previous applications, we
consider deviations from the initial value. Table \ref{tab:TableCrypto}
shows that all cryptocurrencies have heavy tails, with the estimated tail
indices around $1$; despite some apparent heterogeneity in the tail indices,
confidence intervals for the Hill's estimator show that the common tail
index assumption which characterises stable distributions is satisfied.

\begin{table}[h!]
\centering
\begin{threeparttable}
{\tiny\caption{Estimated number of common trends} 

\label{tab:TableCrypto}
\par

\begin{tabular}{lllllllllll}
\hline\hline
&  &  &  &  &  &  &  &  &  &  \\ 
\multicolumn{11}{c}{\textbf{Results and sensitivity analysis}} \\[0.25cm]
&  &  &  &  &  & \multicolumn{1}{|l}{} &  &  &  &  \\ 
\textit{Currency} &  & \multicolumn{2}{c}{\textit{Tail index}} &  & \textit{%
Test }$H_{0}:E\left\vert X\right\vert ^{2}=\infty $ & \multicolumn{1}{|l}{%
\textit{nominal level}} &  & $1\%$ & $5\%$ & $10\%$ \\ 
&  &  &  &  &  & \multicolumn{1}{|l}{} &  &  &  &  \\ 
ADA &  & \multicolumn{2}{c}{$\underset{\left( 0.60,1.22\right) }{0.912}$} & 
& \multicolumn{1}{c}{$\underset{\left( \text{do not reject }H_{0}\right) }{%
0.9499}$} & \multicolumn{1}{|l}{Johansen's trace test} &  & $4$ & $4$ & $3$
\\ 
BCH &  & \multicolumn{2}{c}{$\underset{\left( 0.55,1.11\right) }{0.833}$} & 
& \multicolumn{1}{c}{$\underset{\left( \text{do not reject }H_{0}\right) }{%
0.9502}$} & \multicolumn{1}{|l}{} &  &  &  &  \\ 
BTC &  & \multicolumn{2}{c}{$\underset{\left( 0.85,1.71\right) }{1.284}$} & 
& \multicolumn{1}{c}{$\underset{\left( \text{do not reject }H_{0}\right) }{%
0.9503}$} & \multicolumn{1}{|l}{Johansen's $\lambda _{\max }$ test} &  & $4$
& $3$ & $3$ \\ 
EOS &  & \multicolumn{2}{c}{$\underset{\left( 0.74,1.50\right) }{1.126}$} & 
& \multicolumn{1}{c}{$\underset{\left( \text{do not reject }H_{0}\right) }{%
0.9499}$} & \multicolumn{1}{|l}{} &  &  &  &  \\ 
ETH &  & \multicolumn{2}{c}{$\underset{\left( 0.77,1.55\right) }{1.165}$} & 
& \multicolumn{1}{c}{$\underset{\left( \text{do not reject }H_{0}\right) }{%
0.9559}$} & \multicolumn{1}{|l}{} &  &  &  &  \\ 
XRP &  & \multicolumn{2}{c}{$\underset{\left( 0.97,1.95\right) }{1.464}$} & 
& \multicolumn{1}{c}{$\underset{\left( \text{do not reject }H_{0}\right) }{%
0.9504}$} & \multicolumn{1}{|l}{} &  &  &  &  \\ 
&  &  &  &  &  & \multicolumn{1}{|l}{} &  &  &  &  \\ \hline
&  &  &  &  &  &  &  &  &  &  \\ 
$\left( \kappa =10^{-4};n_{S}=2\right) $ &  & \multicolumn{1}{c}{$\frac{0.05%
}{T}$} & \multicolumn{1}{c}{$\frac{0.05}{\ln T}$} & $\frac{0.05}{N}$ & 
\multicolumn{1}{c}{} & \multicolumn{1}{|l}{$\left( \kappa
=10^{-2};n_{S}=2\right) $} &  & \multicolumn{1}{c}{$\frac{0.05}{T}$} & 
\multicolumn{1}{c}{$\frac{0.05}{\ln T}$} & \multicolumn{1}{c}{$\frac{0.05}{N}
$} \\ 
$M$ &  & \multicolumn{1}{c}{} & \multicolumn{1}{c}{} &  & \multicolumn{1}{c}{
} & \multicolumn{1}{|l}{$M$} &  & \multicolumn{1}{c}{} & \multicolumn{1}{c}{}
& \multicolumn{1}{c}{} \\ 
$T/4$ &  & \multicolumn{1}{c}{$5$} & \multicolumn{1}{c}{$5$} & 
\multicolumn{1}{c}{$5$} & \multicolumn{1}{c}{} & \multicolumn{1}{|l}{$T/4$}
&  & \multicolumn{1}{c}{$5$} & \multicolumn{1}{c}{$5$} & \multicolumn{1}{c}{$%
5$} \\ 
$T/2$ &  & \multicolumn{1}{c}{$5$} & \multicolumn{1}{c}{$5$} & 
\multicolumn{1}{c}{$5$} & \multicolumn{1}{c}{} & \multicolumn{1}{|l}{$T/2$}
&  & \multicolumn{1}{c}{$5$} & \multicolumn{1}{c}{$5$} & \multicolumn{1}{c}{$%
5$} \\ 
$T$ &  & \multicolumn{1}{c}{$5$} & \multicolumn{1}{c}{$5$} & 
\multicolumn{1}{c}{$5$} & \multicolumn{1}{c}{} & \multicolumn{1}{|l}{$T$} & 
& \multicolumn{1}{c}{$5$} & \multicolumn{1}{c}{$5$} & \multicolumn{1}{c}{$5$}
\\ 
$2T$ &  & \multicolumn{1}{c}{$5$} & \multicolumn{1}{c}{$5$} & 
\multicolumn{1}{c}{$5$} & \multicolumn{1}{c}{} & \multicolumn{1}{|l}{$2T$} & 
& \multicolumn{1}{c}{$5$} & \multicolumn{1}{c}{$5$} & \multicolumn{1}{c}{$5$}
\\ 
&  & \multicolumn{1}{c}{} & \multicolumn{1}{c}{} &  & \multicolumn{1}{c}{} & 
\multicolumn{1}{|l}{} &  & \multicolumn{1}{c}{} & \multicolumn{1}{c}{} & 
\multicolumn{1}{c}{} \\ \hline
&  & \multicolumn{1}{c}{} & \multicolumn{1}{c}{} &  & \multicolumn{1}{c}{} & 
&  & \multicolumn{1}{c}{} & \multicolumn{1}{c}{} & \multicolumn{1}{c}{} \\ 
$\left( \kappa =10^{-4};n_{S}=4\right) $ &  & \multicolumn{1}{c}{$\frac{0.05%
}{T}$} & \multicolumn{1}{c}{$\frac{0.05}{\ln T}$} & $\frac{0.05}{N}$ & 
\multicolumn{1}{c}{} & \multicolumn{1}{|l}{$\left( \kappa
=10^{-2};n_{S}=4\right) $} &  & \multicolumn{1}{c}{$\frac{0.05}{T}$} & 
\multicolumn{1}{c}{$\frac{0.05}{\ln T}$} & \multicolumn{1}{c}{$\frac{0.05}{N}
$} \\ 
$M$ &  & \multicolumn{1}{c}{} & \multicolumn{1}{c}{} &  & \multicolumn{1}{c}{
} & \multicolumn{1}{|l}{$M$} &  & \multicolumn{1}{c}{} & \multicolumn{1}{c}{}
& \multicolumn{1}{c}{} \\ 
$T/4$ &  & \multicolumn{1}{c}{$5$} & \multicolumn{1}{c}{$5$} & 
\multicolumn{1}{c}{$5$} & \multicolumn{1}{c}{} & \multicolumn{1}{|l}{$T/4$}
&  & \multicolumn{1}{c}{$5$} & \multicolumn{1}{c}{$5$} & \multicolumn{1}{c}{$%
5$} \\ 
$T/2$ &  & \multicolumn{1}{c}{$5$} & \multicolumn{1}{c}{$5$} & 
\multicolumn{1}{c}{$5$} & \multicolumn{1}{c}{} & \multicolumn{1}{|l}{$T/2$}
&  & \multicolumn{1}{c}{$5$} & \multicolumn{1}{c}{$5$} & \multicolumn{1}{c}{$%
5$} \\ 
$T$ &  & \multicolumn{1}{c}{$5$} & \multicolumn{1}{c}{$5$} & 
\multicolumn{1}{c}{$5$} & \multicolumn{1}{c}{} & \multicolumn{1}{|l}{$T$} & 
& \multicolumn{1}{c}{$5$} & \multicolumn{1}{c}{$5$} & \multicolumn{1}{c}{$5$}
\\ 
$2T$ &  & \multicolumn{1}{c}{$5$} & \multicolumn{1}{c}{$5$} & 
\multicolumn{1}{c}{$5$} & \multicolumn{1}{c}{} & \multicolumn{1}{|l}{$2T$} & 
& \multicolumn{1}{c}{$5$} & \multicolumn{1}{c}{$5$} & \multicolumn{1}{c}{$5$}
\\ 
&  & \multicolumn{1}{c}{} & \multicolumn{1}{c}{} &  & \multicolumn{1}{c}{} & 
\multicolumn{1}{|l}{} &  & \multicolumn{1}{c}{} & \multicolumn{1}{c}{} & 
\multicolumn{1}{c}{} \\ \hline\hline
\end{tabular}

\smallskip
\begin{tablenotes}

     \item In the top part of the table we report the estimated values of the tail index using the Hill's estimator - the package 'ptsuite' in R has been employed, using a number of order statistics equal to $k_{T}=32$, which corresponds to $O(T^{1/2})$. We also report the number of cointegration relationships found by Johansen's procedure; this has been implemented using $p=3$ lags in the $VAR$ specification, as suggested using BIC. 
     \item We report the outcome of our test under various specifications in the bottom half of the table, similarly to Table \ref{tab:TableLing}.

\end{tablenotes}
}
\end{threeparttable}
\end{table}

The results in Table \ref{tab:TableCrypto}, based on our test using
different configurations of the tuning parameters $M$ and $\kappa $,
indicate $m=5$ common trends for all cases considered, i.e. the presence of
one cointegrating relation between the $6$ cryptocurrencies considered. Note
that, in contrast to our findings, Johansen's approach based on a VAR($p$)
with $p=10$ (as suggested by BIC) broadly indicates $m=4$ ($m=3$ at the $%
10\% $ nominal level or for larger values of $p$). However, the results by %
\citet{caner} indicate a tendency to understate $m$ in the presence of heavy
tails; therefore, the additional cointegrating relations found by Johansen's
method is likely to be spuriously driven by the heavy-tailedness of the
cryptocurrency returns. 

\section{Conclusions\label{conclusions}}

In this paper, we propose a methodology for inference on the common trends
in multivariate time series with heavy tailed, heterogeneous innovations. We
develop: \textit{(i)} tests for hypotheses on the number of common trends; 
\textit{(ii)} a sequential procedure to consistently detect the number of
common trends; \textit{(iii)} an estimator of the common trends and of the
associated loadings. A key feature of our approach is that estimation of the
tail index of the innovations is not needed, and no prior knowledge as to
whether the data have finite variance or not is required. Indeed, the
procedure can be applied even in the case of finite second moments. \newline
Our method is based on the eigenvalues of the sample second moment matrix of
the data in levels, the largest $m$ ($m$ being the unknown number of common
trends) of which are shown to diverge at a higher rate, as $T$ increases,
than the remaining ones. Based on such rates, we propose a randomised
statistic for testing hypotheses on $m$; its limiting distribution is
chi-squared under the null, and diverges under the alternative. Combining
these individual tests into a sequence of tests, we prove consistency of the
estimator of $m$ by simply letting the nominal level to shrink to zero at a
proper rate. We also show that, once $m$ is determined, estimation of the
common trends and their loadings can be done using PCA. Our simulations show
that our method has good properties, even in samples of small and moderate
size.

\begin{adjustwidth}{-6.0pt}{-6.0pt}

{\footnotesize 
\bibliographystyle{chicago}
\bibliography{BCT_biblio} }

\end{adjustwidth}

{\small 
\begin{figure}[h!]
\caption{Commodity prices - Estimated common trends $\widehat x_t$}
\label{fig:FigComm}{\small \centering
\hspace{-2.25cm} 
\begin{minipage}{0.45\textwidth}
\centering
    \includegraphics[scale=0.4]{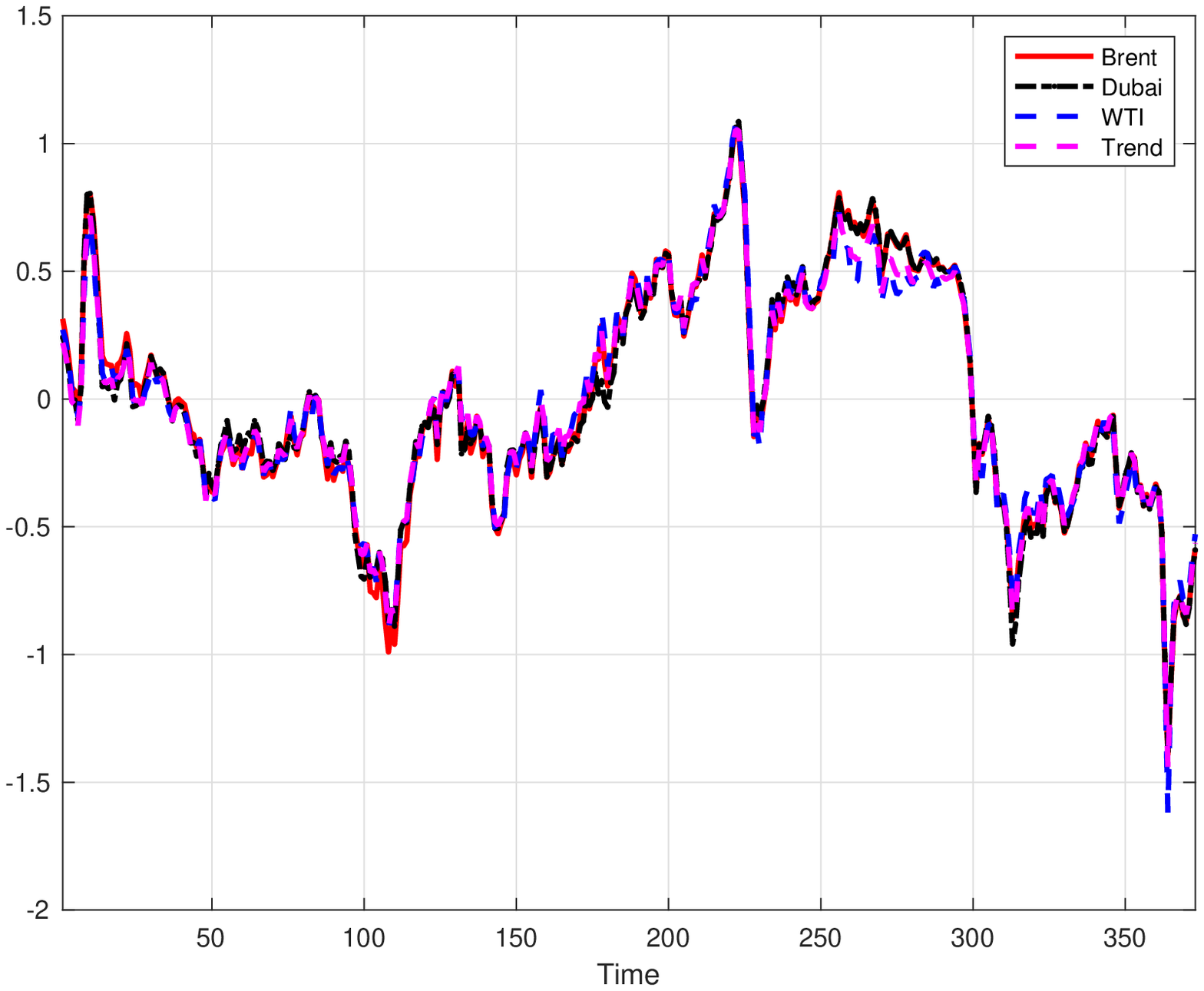}
    \captionof{subfigure}{WTI} 
    \label{fig:t11}
\end{minipage}%
\begin{minipage}{0.4\textwidth}
\centering
   \includegraphics[scale=0.4]{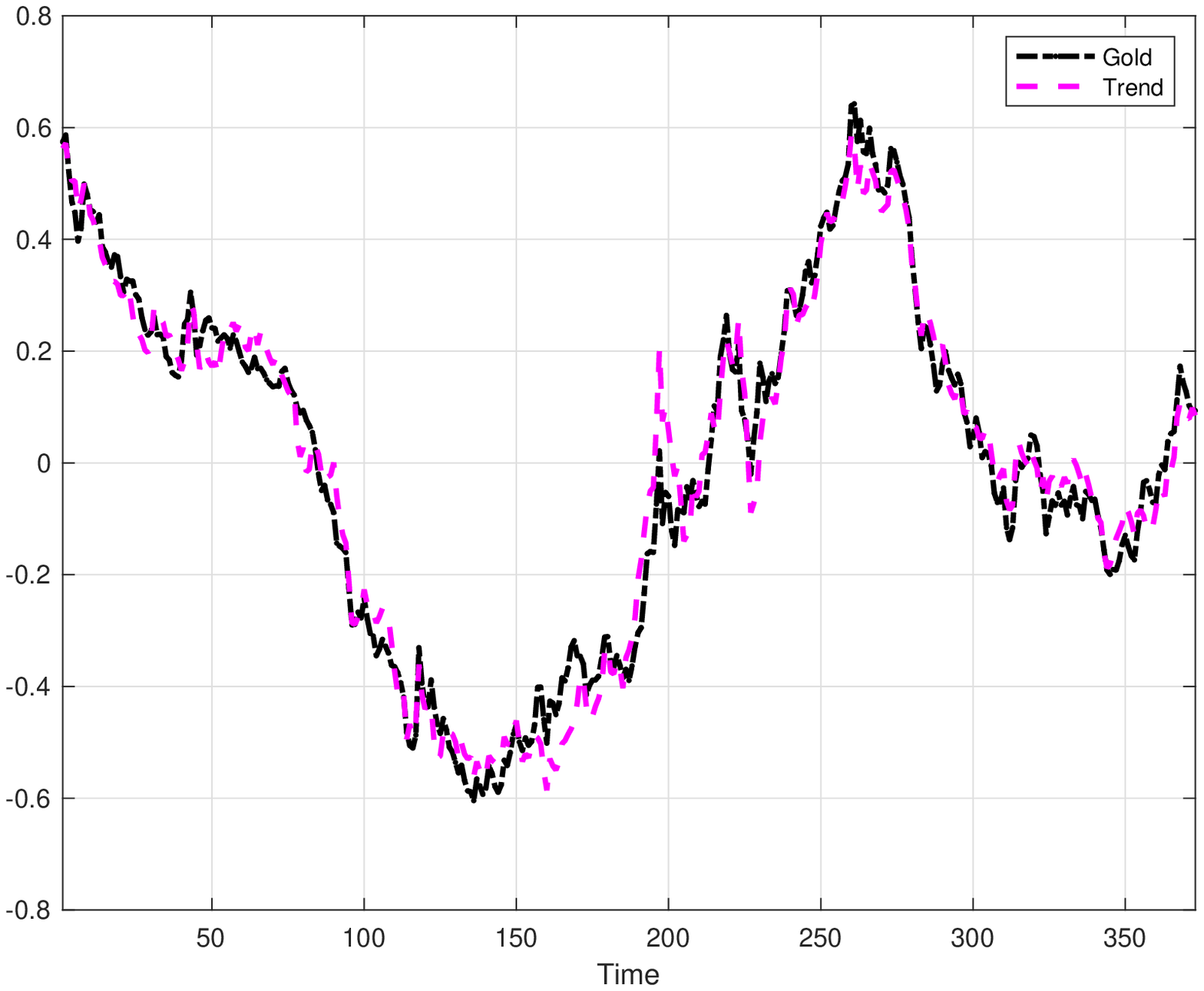}
    \captionof{subfigure}{Gold}
    \label{fig:t12}
\end{minipage} \\[0.1cm]
}
\par
{\small \hspace{-2.25cm} 
\begin{minipage}{0.45\textwidth}
\centering
    \includegraphics[scale=0.4]{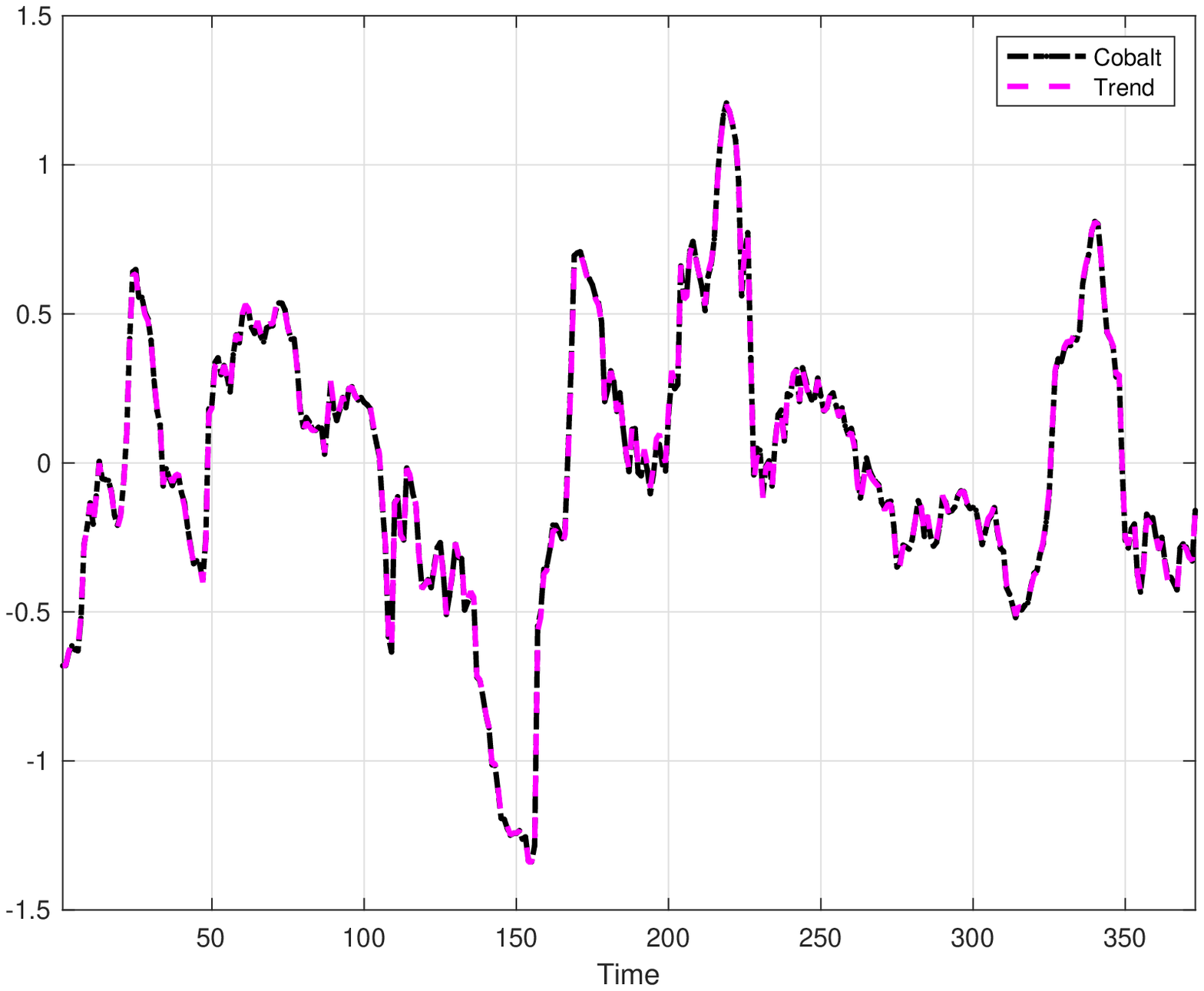}
    \captionof{subfigure}{Cobalt}
    \label{fig:t13}
\end{minipage}%
\begin{minipage}{0.4\textwidth}
\centering
   \includegraphics[scale=0.4]{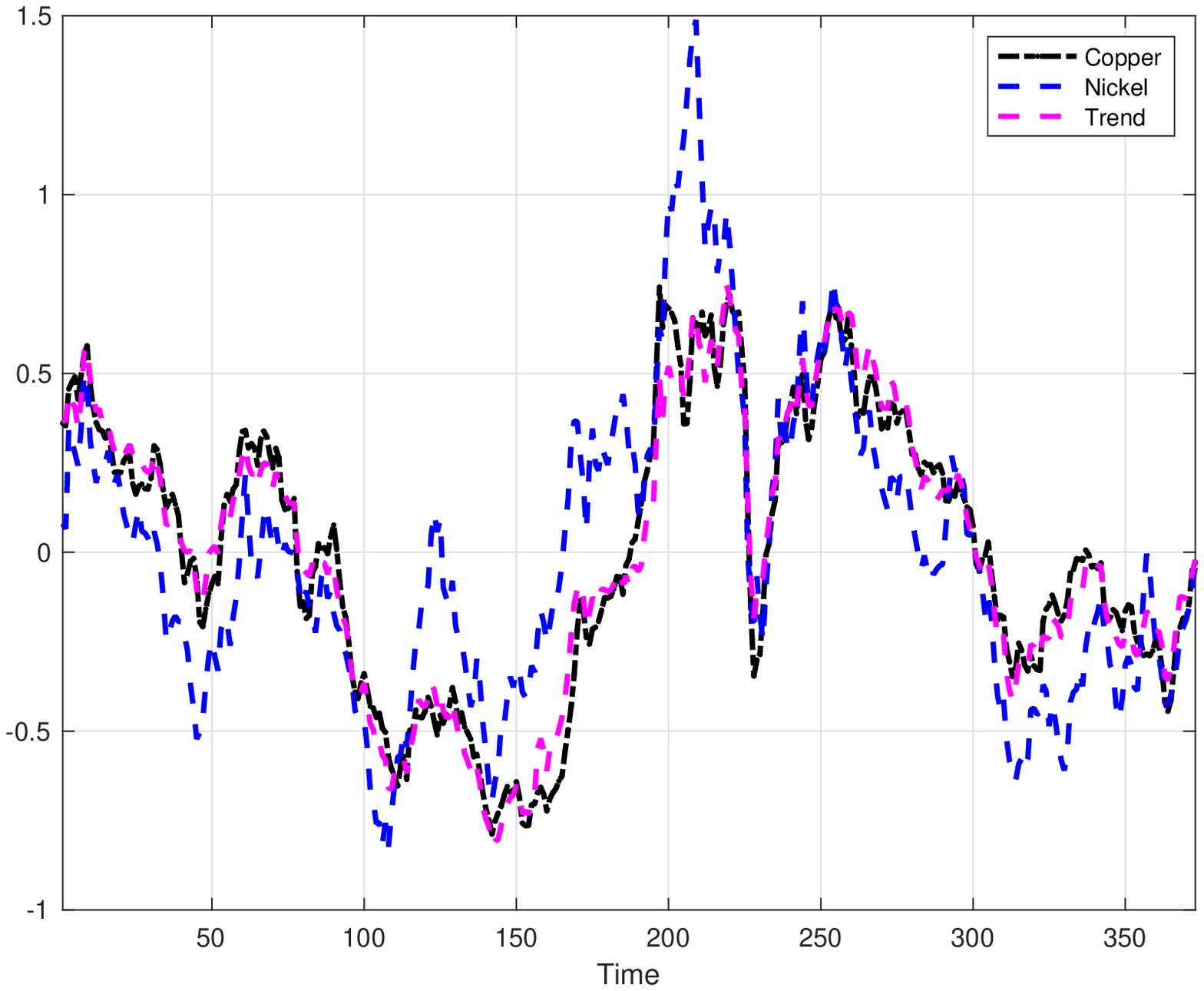}
    \captionof{subfigure}{Copper}
    \label{fig:t14}
\end{minipage} \\[0.1cm]
}
\par
{\small 
}
\end{figure}
}

\newpage

\newpage

\renewcommand{\thesection}{A} 
\renewcommand{\thelemma}{A.\arabic{lemma}} \renewcommand{\thetheorem}{A%
.\arabic{theorem}} \renewcommand{\theequation}{A.\arabic{equation}} %
\renewcommand{\thetable}{S.\arabic{table}}

\setcounter{equation}{0} \setcounter{lemma}{0} \setcounter{theorem}{0} \renewcommand{\thelemma}{A.%
\arabic{lemma}} \renewcommand{\theequation}{A.\arabic{equation}}

\section{Further Monte Carlo evidence\label{fmontecarlo}}

\subsection{Top-down estimation of $m$\label{topdown}}

As mentioned in Section \ref{estimation_R} in the main paper, in order to
determine the number of common trends $m$, it is also possible to consider a
top-down algorithm as an alternative to the \textquotedblleft
bottom-up\textquotedblright\ Algorithm 1, i.e. Algorithm 2. Its asymptotic
validity is demonstrated next.

\begin{theorem}
\label{family2} Let Assumptions \ref{as-B}-\ref{weight} hold and define the
critical value of each individual test as $c_{\alpha }=c_{\alpha }\left(
M\right) $. As $\min \left( T,M\right) \rightarrow \infty $ under (\ref%
{rate-constraint}), if $c_{\alpha }\left( M\right) \rightarrow \infty $ with 
$c_{\alpha }=o\left( M\right) $, then - with $\widetilde{m}$ defined in
Algorithm 2 in the main paper - it holds that $P^{\ast }(\widetilde{m}=m)=1$
for\ almost all realisations of $\left\{ \varepsilon _{t},-\infty <t<\infty
\right\} $.
\end{theorem}

In Tables \ref{tab:TableF21}-\ref{tab:TableF23}, we report the frequencies
of the estimates of $m$, obtained from the same design as in Section \ref%
{montecarlo} in the main paper. As can be seen, results are as good as with $%
\widehat{m}$.

\begin{table}[h!]
\caption{Rank estimation frequencies - $N=3$, Algorithm 2}
\label{tab:TableF21}{\tiny 
\begin{tabular}{cccccccccccc}
\multicolumn{12}{c}{$N=3$} \\ 
&  &  & \multicolumn{4}{c}{$T=100$} &  & \multicolumn{4}{c}{$T=200$} \\ 
\hline
&  & $m$ & $3$ & $2$ & $1$ & $0$ &  & $3$ & $2$ & $1$ & $0$ \\ 
& $m^{\ast }$ &  &  &  &  &  &  &  &  &  &  \\ 
& $3$ &  & $0.950$ & $0.005$ & $0.000$ & $0.000$ &  & $0.989$ & $0.003$ & $%
0.000$ & $0.000$ \\ 
$\eta =0.5$ & $2$ &  & $0.049$ & $0.988$ & $0.013$ & $0.000$ &  & $0.011$ & $%
0.989$ & $0.003$ & $0.000$ \\ 
& $1$ &  & $0.001$ & $0.007$ & $0.987$ & $0.022$ &  & $0.000$ & $0.008$ & $%
0.997$ & $0.008$ \\ 
& $0$ &  & $0.000$ & $0.000$ & $0.000$ & $0.978$ &  & $0.000$ & $0.000$ & $%
0.000$ & $0.992$ \\ 
&  &  &  &  &  &  &  &  &  &  &  \\ 
& $3$ &  & $0.986$ & $0.000$ & $0.000$ & $0.000$ &  & $0.999$ & $0.000$ & $%
0.000$ & $0.000$ \\ 
$\eta =1.0$ & $2$ &  & $0.014$ & $0.995$ & $0.003$ & $0.000$ &  & $0.001$ & $%
1.000$ & $0.002$ & $0.000$ \\ 
& $1$ &  & $0.000$ & $0.005$ & $0.997$ & $0.006$ &  & $0.000$ & $0.000$ & $%
0.998$ & $0.001$ \\ 
& $0$ &  & $0.000$ & $0.000$ & $0.000$ & $0.994$ &  & $0.000$ & $0.000$ & $%
0.000$ & $0.999$ \\ 
&  &  &  &  &  &  &  &  &  &  &  \\ 
& $3$ &  & $0.995$ & $0.000$ & $0.000$ & $0.000$ &  & $1.000$ & $0.000$ & $%
0.000$ & $0.000$ \\ 
$\eta =1.5$ & $2$ &  & $0.005$ & $0.998$ & $0.001$ & $0.000$ &  & $0.000$ & $%
1.000$ & $0.001$ & $0.000$ \\ 
& $1$ &  & $0.000$ & $0.002$ & $0.999$ & $0.002$ &  & $0.000$ & $0.000$ & $%
0.999$ & $0.001$ \\ 
& $0$ &  & $0.000$ & $0.000$ & $0.000$ & $0.998$ &  & $0.000$ & $0.000$ & $%
0.000$ & $0.999$ \\ 
&  &  &  &  &  &  &  &  &  &  &  \\ 
& $3$ &  & $0.995$ & $0.000$ & $0.000$ & $0.000$ &  & $1.000$ & $0.000$ & $%
0.000$ & $0.000$ \\ 
$\eta =2$ & $2$ &  & $0.005$ & $0.998$ & $0.000$ & $0.000$ &  & $0.000$ & $%
1.000$ & $0.001$ & $0.000$ \\ 
& $1$ &  & $0.000$ & $0.002$ & $1.000$ & $0.000$ &  & $0.000$ & $0.000$ & $%
0.999$ & $0.000$ \\ 
& $0$ &  & $0.000$ & $0.000$ & $0.000$ & $1.000$ &  & $0.000$ & $0.000$ & $%
0.000$ & $1.000$ \\ 
&  &  &  &  &  &  &  &  &  &  &  \\ \hline\hline
\end{tabular}
}
\end{table}

\begin{table}[h!]
\caption{Rank estimation frequencies - $N=4$, Algorithm 2}
\label{tab:TableF22}{\tiny 
\begin{tabular}{cccccccccccccc}
\multicolumn{14}{c}{$N=4$} \\ 
&  &  & \multicolumn{5}{c}{$T=100$} &  & \multicolumn{5}{c}{$T=200$} \\ 
\hline
&  &  &  &  &  &  &  &  &  &  &  &  &  \\ 
&  & $m$ & $4$ & $3$ & $2$ & $1$ & $0$ &  & $4$ & $3$ & $2$ & $1$ & $0$ \\ 
& $m^{\ast }$ &  &  &  &  &  &  &  &  &  &  &  &  \\ 
& $4$ &  & $0.884$ & $0.007$ & $0.000$ & $0.000$ & $0.000$ &  & $0.970$ & $%
0.000$ & $0.000$ & $0.000$ & $0.000$ \\ 
& $3$ &  & $0.114$ & $0.962$ & $0.008$ & $0.000$ & $0.000$ &  & $0.029$ & $%
0.990$ & $0.002$ & $0.000$ & $0.000$ \\ 
$\eta =0.5$ & $2$ &  & $0.002$ & $0.032$ & $0.989$ & $0.018$ & $0.000$ &  & $%
0.001$ & $0.010$ & $0.996$ & $0.007$ & $0.000$ \\ 
& $1$ &  & $0.000$ & $0.000$ & $0.003$ & $0.981$ & $0.032$ &  & $0.000$ & $%
0.000$ & $0.002$ & $0.993$ & $0.015$ \\ 
& $0$ &  & $0.000$ & $0.000$ & $0.000$ & $0.001$ & $0.968$ &  & $0.000$ & $%
0.000$ & $0.000$ & $0.000$ & $0.985$ \\ 
&  &  &  &  &  &  &  &  &  &  &  &  &  \\ 
& $4$ &  & $0.951$ & $0.000$ & $0.000$ & $0.000$ & $0.000$ &  & $0.995$ & $%
0.000$ & $0.000$ & $0.000$ & $0.000$ \\ 
& $3$ &  & $0.047$ & $0.987$ & $0.000$ & $0.000$ & $0.000$ &  & $0.005$ & $%
1.000$ & $0.000$ & $0.000$ & $0.000$ \\ 
$\eta =1.0$ & $2$ &  & $0.002$ & $0.013$ & $0.997$ & $0.004$ & $0.000$ &  & $%
0.000$ & $0.000$ & $0.998$ & $0.002$ & $0.000$ \\ 
& $1$ &  & $0.000$ & $0.000$ & $0.003$ & $0.994$ & $0.006$ &  & $0.000$ & $%
0.000$ & $0.002$ & $0.998$ & $0.005$ \\ 
& $0$ &  & $0.000$ & $0.000$ & $0.000$ & $0.002$ & $0.994$ &  & $0.000$ & $%
0.000$ & $0.000$ & $0.000$ & $0.995$ \\ 
&  &  &  &  &  &  &  &  &  &  &  &  &  \\ 
& $4$ &  & $0.969$ & $0.000$ & $0.000$ & $0.000$ & $0.000$ &  & $0.997$ & $%
0.000$ & $0.000$ & $0.000$ & $0.000$ \\ 
& $3$ &  & $0.030$ & $0.991$ & $0.000$ & $0.000$ & $0.000$ &  & $0.003$ & $%
1.000$ & $0.000$ & $0.000$ & $0.000$ \\ 
$\eta =1.5$ & $2$ &  & $0.001$ & $0.009$ & $0.997$ & $0.002$ & $0.000$ &  & $%
0.000$ & $0.000$ & $1.000$ & $0.000$ & $0.000$ \\ 
& $1$ &  & $0.000$ & $0.000$ & $0.003$ & $0.997$ & $0.002$ &  & $0.000$ & $%
0.000$ & $0.000$ & $1.000$ & $0.001$ \\ 
& $0$ &  & $0.000$ & $0.000$ & $0.000$ & $0.001$ & $0.998$ &  & $0.000$ & $%
0.000$ & $0.000$ & $0.000$ & $0.999$ \\ 
&  &  &  &  &  &  &  &  &  &  &  &  &  \\ 
& $4$ &  & $0.982$ & $0.000$ & $0.000$ & $0.000$ & $0.000$ &  & $0.998$ & $%
0.000$ & $0.000$ & $0.000$ & $0.000$ \\ 
& $3$ &  & $0.018$ & $0.995$ & $0.000$ & $0.000$ & $0.000$ &  & $0.002$ & $%
1.000$ & $0.000$ & $0.000$ & $0.000$ \\ 
$\eta =2$ & $2$ &  & $0.000$ & $0.005$ & $0.998$ & $0.000$ & $0.000$ &  & $%
0.000$ & $0.000$ & $1.000$ & $0.000$ & $0.000$ \\ 
& $1$ &  & $0.000$ & $0.000$ & $0.002$ & $1.000$ & $0.001$ &  & $0.000$ & $%
0.000$ & $0.000$ & $1.000$ & $0.000$ \\ 
& $0$ &  & $0.000$ & $0.000$ & $0.000$ & $0.000$ & $0.999$ &  & $0.000$ & $%
0.000$ & $0.000$ & $0.000$ & $1.000$ \\ 
&  &  &  &  &  &  &  &  &  &  &  &  &  \\ \hline\hline
\end{tabular}
}
\end{table}

\begin{table}[h!]
\caption{Rank estimation frequencies - $N=5$, Algorithm 2}
\label{tab:TableF23}{\tiny 
\begin{tabular}{cccccccccccccccc}
\multicolumn{16}{c}{$N=5$} \\ 
&  &  & \multicolumn{6}{c}{$T=100$} &  & \multicolumn{6}{c}{$T=200$} \\ 
\hline
&  & $m$ & $5$ & $4$ & $3$ & $2$ & $1$ & $0$ &  & $5$ & $4$ & $3$ & $2$ & $1$
& $0$ \\ 
& $m^{\ast }$ &  &  &  &  &  &  &  &  &  &  &  &  &  &  \\ 
& $5$ &  & $0.782$ & $0.001$ & $0.000$ & $0.000$ & $0.000$ & $0.000$ &  & $%
0.937$ & $0.000$ & $0.000$ & $0.000$ & $0.000$ & $0.000$ \\ 
& $4$ &  & $0.211$ & $0.910$ & $0.005$ & $0.000$ & $0.000$ & $0.000$ &  & $%
0.061$ & $0.979$ & $0.004$ & $0.000$ & $0.000$ & $0.000$ \\ 
$\eta =0.5$ & $3$ &  & $0.007$ & $0.089$ & $0.975$ & $0.011$ & $0.000$ & $%
0.000$ &  & $0.002$ & $0.021$ & $0.995$ & $0.005$ & $0.000$ & $0.000$ \\ 
& $2$ &  & $0.000$ & $0.000$ & $0.020$ & $0.984$ & $0.014$ & $0.000$ &  & $%
0.000$ & $0.000$ & $0.001$ & $0.995$ & $0.016$ & $0.000$ \\ 
& $1$ &  & $0.000$ & $0.000$ & $0.000$ & $0.005$ & $0.984$ & $0.039$ &  & $%
0.000$ & $0.000$ & $0.000$ & $0.000$ & $0.984$ & $0.018$ \\ 
& $0$ &  & $0.000$ & $0.000$ & $0.000$ & $0.000$ & $0.002$ & $0.961$ &  & $%
0.000$ & $0.000$ & $0.000$ & $0.000$ & $0.000$ & $0.982$ \\ 
&  &  &  &  &  &  &  &  &  &  &  &  &  &  &  \\ 
& $5$ &  & $0.876$ & $0.001$ & $0.000$ & $0.000$ & $0.000$ & $0.000$ &  & $%
0.984$ & $0.000$ & $0.000$ & $0.000$ & $0.000$ & $0.000$ \\ 
& $4$ &  & $0.122$ & $0.954$ & $0.001$ & $0.000$ & $0.000$ & $0.000$ &  & $%
0.016$ & $0.997$ & $0.000$ & $0.000$ & $0.000$ & $0.000$ \\ 
$\eta =1.0$ & $3$ &  & $0.002$ & $0.045$ & $0.992$ & $0.003$ & $0.000$ & $%
0.000$ &  & $0.000$ & $0.003$ & $1.000$ & $0.000$ & $0.000$ & $0.000$ \\ 
& $2$ &  & $0.000$ & $0.000$ & $0.007$ & $0.996$ & $0.003$ & $0.000$ &  & $%
0.000$ & $0.000$ & $0.000$ & $1.000$ & $0.002$ & $0.000$ \\ 
& $1$ &  & $0.000$ & $0.000$ & $0.000$ & $0.001$ & $0.997$ & $0.008$ &  & $%
0.000$ & $0.000$ & $0.000$ & $0.000$ & $0.996$ & $0.003$ \\ 
& $0$ &  & $0.000$ & $0.000$ & $0.000$ & $0.000$ & $0.000$ & $0.992$ &  & $%
0.000$ & $0.000$ & $0.000$ & $0.000$ & $0.002$ & $0.997$ \\ 
&  &  &  &  &  &  &  &  &  &  &  &  &  &  &  \\ 
& $5$ &  & $0.892$ & $0.000$ & $0.000$ & $0.000$ & $0.000$ & $0.000$ &  & $%
0.997$ & $0.000$ & $0.000$ & $0.000$ & $0.000$ & $0.000$ \\ 
& $4$ &  & $0.108$ & $0.976$ & $0.000$ & $0.000$ & $0.000$ & $0.000$ &  & $%
0.003$ & $1.000$ & $0.000$ & $0.000$ & $0.000$ & $0.000$ \\ 
$\eta =1.5$ & $3$ &  & $0.000$ & $0.024$ & $0.998$ & $0.002$ & $0.000$ & $%
0.000$ &  & $0.000$ & $0.000$ & $1.000$ & $0.000$ & $0.000$ & $0.000$ \\ 
& $2$ &  & $0.000$ & $0.000$ & $0.002$ & $0.998$ & $0.001$ & $0.000$ &  & $%
0.000$ & $0.000$ & $0.000$ & $1.000$ & $0.000$ & $0.000$ \\ 
& $1$ &  & $0.000$ & $0.000$ & $0.000$ & $0.000$ & $0.999$ & $0.003$ &  & $%
0.000$ & $0.000$ & $0.000$ & $0.000$ & $1.000$ & $0.000$ \\ 
& $0$ &  & $0.000$ & $0.000$ & $0.000$ & $0.000$ & $0.000$ & $0.997$ &  & $%
0.000$ & $0.000$ & $0.000$ & $0.000$ & $0.000$ & $1.000$ \\ 
&  &  &  &  &  &  &  &  &  &  &  &  &  &  &  \\ 
& $5$ &  & $0.912$ & $0.000$ & $0.000$ & $0.000$ & $0.000$ & $0.000$ &  & $%
0.997$ & $0.000$ & $0.000$ & $0.000$ & $0.000$ & $0.000$ \\ 
& $4$ &  & $0.088$ & $0.984$ & $0.000$ & $0.000$ & $0.000$ & $0.000$ &  & $%
0.003$ & $1.000$ & $0.000$ & $0.000$ & $0.000$ & $0.000$ \\ 
$\eta =2$ & $3$ &  & $0.000$ & $0.016$ & $0.998$ & $0.000$ & $0.000$ & $%
0.000 $ &  & $0.000$ & $0.000$ & $1.000$ & $0.000$ & $0.000$ & $0.000$ \\ 
& $2$ &  & $0.000$ & $0.000$ & $0.002$ & $1.000$ & $0.000$ & $0.000$ &  & $%
0.000$ & $0.000$ & $0.000$ & $1.000$ & $0.000$ & $0.000$ \\ 
& $1$ &  & $0.000$ & $0.000$ & $0.000$ & $0.000$ & $0.999$ & $0.000$ &  & $%
0.000$ & $0.000$ & $0.000$ & $0.000$ & $1.000$ & $0.000$ \\ 
& $0$ &  & $0.000$ & $0.000$ & $0.000$ & $0.000$ & $0.001$ & $1.000$ &  & $%
0.000$ & $0.000$ & $0.000$ & $0.000$ & $0.000$ & $1.000$ \\ 
&  &  &  &  &  &  &  &  &  &  &  &  &  &  &  \\ \hline\hline
\end{tabular}
}
\end{table}

\bigskip

\subsection{Further evidence on $\widehat{m}$\label{furthermhat}}

Here, we use the same DGP as in Section \ref{montecarlo}, viz. 
\begin{equation}
y_{t}=Ay_{t-1}+\varepsilon _{t},  \label{var-1}
\end{equation}%
with $A=I_{N}-PP^{\prime }$ and $P$ an orthonormal $N\times \left(
N-m\right) $ matrix; however, as well as considering \textit{i.i.d.} errors,
we also allow for serial correlation in $\varepsilon _{t}$. Specifically,
the innovations $\varepsilon _{t}$ in (\ref{var-1}) have again been
generated as coordinate-wise independent, and we consider serial dependence
through an $AR\left( 1\right) $ specification%
\begin{equation}
\varepsilon _{i,t}=\theta \varepsilon _{i,t-1}+e_{i,t}  \label{error-serial}
\end{equation}%
for $1\leq i\leq N$; we set $\theta \in \left\{ 0,0.5,-0.5\right\} $. The
innovation $e_{i,t}$ is generated as \textit{i.i.d.} across $t$, with a
power law distribution of tail index $\eta $, using $\eta \in \left\{
0.5,1,1.5,2\right\} $ and again, as a benchmark, we consider the case $%
\varepsilon _{t}\sim i.i.d.N\left( 0,I_{N}\right) $. In order to simulate
the power law distribution, we generate%
\begin{equation}
e_{i,t}=\left( 1-v_{i,t}\right) ^{-1/\eta },  \label{powerlaw}
\end{equation}%
for $1\leq i\leq N$, where $v_{i,t}$ is $i.i.d.U\left[ 0,1\right] $; $%
e_{i,t} $ is subsequently centered. The first $1,000$ datapoints are
discarded in order to avoid dependence on initial conditions.

All other specifications are the same as in the main paper.

\bigskip

We use information criteria as a comparison. In particular, we consider the
three most popular ones (i.e., AIC, BIC and HQ), but in the tables we only
report the results corresponding to BIC, which is consistently the best out
of the three for all experiments\footnote{%
We note that AIC and HQ tend to grossly overstate $R$ in all experiments,
with this tendency worsening for lower values of $\eta $. Results are
available upon request.}. We point out that the use of information criteria
in the context of cointegration analysis is well-studied in the finite
variance case (see e.g. \citealp*{aznar} and \citealp*{%
cavaliere2018determining}); conversely, the case of infinite variance has
not been studied in the case of cointegrated systems (see however the
contribution by \citealp*{knight}). In our experiments, we use an
\textquotedblleft infeasible\textquotedblright\ version of each information
criterion, assuming that the lag structure in (\ref{var-1}) is known, so
that the only quantity to be determined is the number of common trends $m$.

For all rank estimates, we report three measures of performance. Denoting
the estimate of $m$ at iteration $1\leq mc\leq 1,000$ as $\widehat{m}_{mc}$,
we define%
\begin{eqnarray}
ME &=&\frac{1}{1000}\sum_{mc=1}^{1000}\widehat{m}_{mc},  \label{me} \\
STD &=&\frac{1}{1000}\sum_{mc=1}^{1000}\left( \widehat{m}_{mc}-ME\right)
^{2},  \label{std} \\
PCW &=&\frac{1}{1000}\sum_{mc=1}^{1000}I\left( \widehat{m}_{mc}\neq m\right)
,  \label{pcw}
\end{eqnarray}

Results are reported in Tables \ref{tab:Table1a} and \ref{tab:Table1b},
where we analyse the properties of our estimator of $m$ for a small-scale $%
VAR$ with $N\in \left\{ 3,4,5\right\} $, and for the simple case of \textit{%
i.i.d.} errors. Broadly speaking, as far as the $PCW$ and $ME$ indicators
are concerned, the tables show that our procedure is very good on average at
estimating $m$ - and better than the best information criterion, BIC - for
all values of $N$ and $T$, even for the (rather extreme) case $\eta =0.5$.
This is true across all values of $m$, including the case of stationarity ($%
m=0$) and lack of cointegration ($m=N$). Such evidence is further
corroborated by considering the $STD$ indicator, which shows that our
procedure is systematically better than BIC in this respect. Indeed, the BIC
performs (very) marginally better when $m=0$, but this is more than offset
when considering that BIC tends to overstate $m$ in general (as the values
of $ME$ indicate), and especially when $m=N$ and $m=N-1$; in particular, BIC
seems to perform puzzlingly badly when errors are Gaussian. This is even
more remarkable as information criteria are based on the (infeasible)
assumption that the lag structure is known. In line with the theory, results
are better for larger $T$; also, results do improve as $\eta $ increases,
although this becomes less evident as $T$ increases. The impact of $N$ is
less clear, although generally speaking larger values of $N$ worsen the
performance of both BIC and our procedure when $m=N$, whereas it has
virtually no impact for other values of $m$. \newline

In Tables \ref{tab:Table2a}-\ref{tab:Table3b} we consider serially
correlated innovations. When $\theta =0.5$ in (\ref{error-serial}), our
procedure fares better than in the \textit{i.i.d.} case when $m=N$,
worsening instead when $m=0$. In the latter case, as $T$ increases (and as $%
\eta $ becomes closer to $2$), the impact of serial correlation is mitigated
-\ see Tables \ref{tab:Table2a} and \ref{tab:Table2b}. Conversely, the case
of negative autocorrelation (Tables \ref{tab:Table3a} and \ref{tab:Table3b}%
), results worsen in a more pronounced way, especially for small values of $%
T $ and/or as $N$ increases. Even in this case, however, our procedure is
better than BIC for the case $m=N$ (and in general for large $m$), with an
almost indistinguishable performance when $m=0$.

Finally, and as also mentioned in the main paper, in Tables \ref{tab:Table4}-%
\ref{tab:Table6}, we report results using $N\in \left\{ 10,15,20\right\} $.
As overview of the results is also in Section \ref{montecarlo} in the main
paper.

\renewcommand{\thesection}{B} 
\renewcommand{\thelemma}{B.\arabic{lemma}} \renewcommand{\thetheorem}{B%
.\arabic{theorem}} \renewcommand{\theequation}{B.\arabic{equation}}

\newpage

\setcounter{equation}{0} \setcounter{lemma}{0} \renewcommand{\thelemma}{B.%
\arabic{lemma}} \renewcommand{\theequation}{B.\arabic{equation}}

\section{Technical lemmas\label{lemmas}}

Henceforth, we denote the $L_{p}$-norm of a scalar random variable $X$ as $%
\left\vert X\right\vert _{p}=\left( E\left\vert X\right\vert ^{p}\right)
^{1/p}$; also, the subscript \textquotedblleft $_{\perp }$%
\textquotedblright\ denotes the orthogonal complement of a matrix. Further,
we often use the notation $R=N-m$.

\begin{lemma}
\label{stout}Consider a sequence $U_{T}$ for which $E\left\vert
U_{T}\right\vert ^{\delta }\leq a_{T}$, where $a_{T}$ is a positive,
monotonically non-decreasing sequence, and let $d_{\delta }=1$ if $\delta
\leq 2$ and zero otherwise. Then there exists a $C_{0}<\infty $ such that%
\begin{equation*}
\lim \sup_{T\rightarrow \infty }\frac{\left\vert U_{T}\right\vert }{%
a_{T}^{1/\delta }\left( \ln T\right) ^{\left( 1+d_{\delta }+\epsilon \right)
/\delta }}\leq C_{0}\text{ a.s.}
\end{equation*}%
for every $\epsilon >0$.

\begin{proof}
It holds that%
\begin{equation*}
E\max_{1\leq t\leq T}\left\vert U_{t}\right\vert ^{\delta }\leq a_{T}\left(
\ln T\right) ^{d_{\delta }},
\end{equation*}%
having used Theorem B in \citet*{serfling1970} when $\delta >2$, and the
first theorem in the same paper when $\delta \leq 2$. Thus%
\begin{equation*}
\sum_{T=1}^{\infty }\frac{1}{T}P\left( \max_{1\leq t\leq T}\left\vert
U_{t}\right\vert >a_{T}^{1/\delta }\left( \ln T\right) ^{\left( 1+d_{\delta
}+\epsilon \right) /\delta }\right) \leq \sum_{T=1}^{\infty }\frac{%
E\max_{1\leq t\leq T}\left\vert U_{t}\right\vert ^{\delta }}{Ta_{T}\left(
\ln T\right) ^{1+d_{\delta }+\epsilon }}=\sum_{T=1}^{\infty }\frac{1}{%
T\left( \ln T\right) ^{1+\epsilon }}<\infty ,
\end{equation*}%
having used Markov inequality. By the proof of Corollary 2.4 in \citet*{%
cai2006} (see also Lemma A.5 in \citealp*{HT16}), it ultimately follows that 
\begin{equation*}
\sum_{T=1}^{\infty} P \left[ \left\vert U_{T}\right\vert >
a_{T}^{1/\delta}\left( \ln T\right) ^{1+d_{\delta }+\epsilon } \right] <
\infty ,
\end{equation*}
which entails that the sequence $U_{T}$ converges completely (\citealp*{%
hsurobbins}). The desired result now follows immediately from the
Borel-Cantelli lemma. 
\end{proof}
\end{lemma}

\begin{lemma}
\label{donsker}We assume that Assumption \ref{as-moments}\textit{(ii)} is
satisfied. Then there exists a random variable $T_{0}$ and a positive
definite matrix $D$ such that, for all $T\geq T_{0}$%
\begin{equation*}
\sum_{t=1}^{T}x_{t}x_{t}^{\prime }\geq D\frac{T^{1+2/\eta }}{\left( \ln \ln
T\right) ^{2/\eta }}.
\end{equation*}

\begin{proof}
The lemma follows immediately upon noting that, by Assumption \ref%
{as-moments}\textit{(ii)} and Theorem 1.3 in \citet*{jain}, it holds a.s.
that, for all $b\in \mathbb{R}^{m}$, $b\neq 0$ 
\begin{equation*}
\lim \inf_{T\rightarrow \infty }\frac{\left( \ln \ln T\right) ^{2/\eta }}{%
T^{1+2/\eta }}\sum_{t=1}^{T}\left( b^{\prime }x_{t}\right) ^{2}=c_{0}>0\text{
a.s.}
\end{equation*}%
\end{proof}
\end{lemma}

\begin{lemma}
\label{b-xt}Let Assumptions \ref{as-B} and \ref{as-moments}\textit{(ii)} be
satisfied. Then there exists a random variable $T_{0}$ such that, for all $%
T\geq T_{0}$%
\begin{equation*}
\lambda ^{\left( j\right) }\left( P\sum_{t=1}^{T}x_{t}x_{t}^{\prime
}P^{\prime }\right) \geq c_{0}\frac{T^{1+2/\eta }}{\left( \ln \ln T\right)
^{2/\eta }},
\end{equation*}%
for all $1\leq j\leq m$.

\begin{proof}
The lemma is an immediate consequence of Lemma \ref{donsker}. Indeed, by the
multiplicative Weyl's inequality (see Theorem 7 in \citealp*{merikoski2004}%
), it holds that%
\begin{equation*}
\lambda ^{\left( j\right) }\left( P\sum_{t=1}^{T}x_{t}x_{t}^{\prime
}P^{\prime }\right) \geq \lambda ^{\left( \min \right) }\left( P^{\prime
}P\right) \lambda ^{\left( j\right) }\left( \sum_{t=1}^{T}x_{t}x_{t}^{\prime
}\right) ,
\end{equation*}%
and Lemma \ref{donsker} readily implies that, for every $j$ 
\begin{equation*}
\lim \inf_{T\rightarrow \infty }\frac{\left( \ln \ln T\right) ^{2/\eta }}{%
T^{1+2/\eta }}\lambda ^{\left( j\right) }\left(
\sum_{t=1}^{T}x_{t}x_{t}^{\prime }\right) >0\;\;\quad \mbox{a.s.}
\end{equation*}%
Assumption \ref{as-B} entails that $\lambda ^{\left( \min
\right) }\left( P^{\prime }P\right) >0$, whence the desired result. 
\end{proof}
\end{lemma}

\begin{lemma}
\label{ut-2}Let Assumption \ref{as-moments} be satisfied. Then, for all $%
\epsilon >0$ it holds that 
\begin{equation*}
\lambda ^{\left( \max \right) }\left( \sum_{t=1}^{T}u_{t}u_{t}^{\prime
}\right) =o_{a.s.}\left( T^{2/p}\left( \ln T\right) ^{2\left( 2+\epsilon
\right) /p}\right) ,
\end{equation*}%
where $0< p<\eta $ when $\eta \leq 2$ with $E\left\Vert \varepsilon
_{t}\right\Vert ^{\eta }=\infty $, and $p=2$ when $\eta =2$ with $%
E\left\Vert \varepsilon _{t}\right\Vert ^{\eta }<\infty $.

\begin{proof}
The proof does not require any distinction between $p<2$ and $p=2$. We have%
\begin{equation*}
\lambda ^{\left( \max \right) }\left( \sum_{t=1}^{T}u_{t}u_{t}^{\prime
}\right) \leq \sum_{i=1}^{N}\sum_{t=1}^{T}u_{i,t}^{2}.
\end{equation*}%
Also, given that $u_{i,t}=\sum_{k=1}^{N}\sum_{m=0}^{\infty }C_{m,ik}^{\ast
}\varepsilon _{k,t-m}$%
\begin{equation*}
\sum_{t=1}^{T}u_{i,t}^{2}=\sum_{t=1}^{T}\left(
\sum_{m=0}^{t-1}C_{m,ik}^{\ast }\varepsilon _{k,t-m}\right)
^{2}=\sum_{k,k^{\prime
}=1}^{N}\sum_{t=1}^{T}\sum_{m=0}^{t-1}\sum_{n=0}^{t-1}C_{m,ik}^{\ast
}C_{n,ik^{\prime }}^{\ast }\varepsilon _{k,t-m}\varepsilon _{k^{\prime
},t-n}.
\end{equation*}%
It holds that%
\begin{align*}
\left\vert \sum_{t=1}^{T}u_{i,t}^{2}\right\vert _{p/2}^{p/2}& \leq
\sum_{k,k^{\prime
}=1}^{N}\sum_{t=1}^{T}\sum_{m=0}^{t-1}\sum_{n=0}^{t-1}\left\vert
C_{m,ik}^{\ast }\right\vert ^{p/2}\left\vert C_{n,ik^{\prime }}^{\ast
}\right\vert ^{p/2}\left\vert \varepsilon _{k,t-m}\varepsilon _{k^{\prime
},t-n}\right\vert _{p/2}^{p/2} \\
& \leq \sum_{k,k^{\prime
}=1}^{N}\sum_{t=1}^{T}\sum_{m=0}^{t-1}\sum_{n=0}^{t-1}\left\vert
C_{m,ik}^{\ast }\right\vert ^{p/2}\left\vert C_{n,ik^{\prime }}^{\ast
}\right\vert ^{p/2}\left\vert \varepsilon _{k,t-m}\right\vert
_{p}^{p/2}\left\vert \varepsilon _{k^{\prime },t-n}\right\vert _{p}^{p/2} \\
& \leq c_{0}\sum_{k,k^{\prime }=1}^{N}\sum_{t=1}^{T}\left(
\sum_{m=0}^{t-1}\left\vert C_{m,ik}^{\ast }\right\vert ^{p/2}\right) \left(
\sum_{n=0}^{t-1}\left\vert C_{n,ik^{\prime }}^{\ast }\right\vert
^{p/2}\right) \\
& \leq c_{0}T,
\end{align*}%
where we have used the Cauchy-Schwartz inequality and the fact that $%
C_{m,ik}^{\ast }=O\left( \exp \left( -c_{0}m\right) \right) $, which follows from Assumption \ref{as-B}. Using Lemma \ref{stout}, the desired
result follows immediately.

\end{proof}
\end{lemma}

\begin{lemma}
\label{cross-prod}Let Assumptions \ref{as-B}-\ref{as-moments} be satisfied.
Then, for all $\epsilon >0$, it holds that 
\begin{equation*}
\lambda ^{\left( \max \right) }\left( P\sum_{t=1}^{T}x_{t}u_{t}^{\prime
}+\sum_{t=1}^{T}u_{t}x_{t}^{\prime }P^{\prime }\right) =o_{a.s.}\left(
T^{2/p}\left( \ln T\right) ^{2\left( 2+\epsilon \right) /p}\right) ,
\end{equation*}%
where $0< p<\eta $ when $\eta \leq 2$ with $E\left\Vert \varepsilon
_{t}\right\Vert ^{\eta }=\infty $, and $p=2$ when $\eta =2$ with $%
E\left\Vert \varepsilon _{t}\right\Vert ^{2}<\infty $.

\begin{proof}
Recall that $C_{h,k}$ indicates the element of $C$ in position $h,k$. Using
the definition of Frobenius norm and the $L_{p}$-norm inequality, it follows
that%
\begin{equation*}
\lambda ^{\left( \max \right) }\left( B\sum_{t=1}^{T}x_{t}u_{t}^{\prime
}+\sum_{t=1}^{T}u_{t}x_{t}^{\prime }B^{\prime }\right) \leq
2\sum_{k,h=1}^{N}\left\vert \sum_{t=1}^{T}C_{h,k}\sum_{s=1}^{t}\varepsilon
_{k,t}\sum_{m=0}^{\infty }C_{m,jh}^{\ast }\varepsilon _{h,t-m}\right\vert .
\end{equation*}%
Thus, using the $C_{r}$-inequality, after some algebra it holds that%
\begin{equation*}
E\left\vert \sum_{k,h=1}^{N}\left\vert
\sum_{t=1}^{T}C_{h,k}\sum_{s=1}^{t}\varepsilon _{k,t}\sum_{m=0}^{\infty
}C_{m,jh}^{\ast }\varepsilon _{h,t-m}\right\vert \right\vert ^{p}\leq
c_{0}\sum_{k,h=1}^{N}E\left\vert \sum_{t=1}^{T}\sum_{s=1}^{t}\varepsilon
_{k,t}\sum_{m=0}^{\infty }C_{m,jh}^{\ast }\varepsilon _{h,t-m}\right\vert
^{p},
\end{equation*}%
where the constant $c_{0}<\infty $ depends only on $p$, $N$ and $m$. It is
convenient to consider separately the cases where $0\leq p<\eta $ when $\eta
\leq 2$ with $E\left\Vert \varepsilon _{t}\right\Vert ^{\eta }=\infty $, and 
$p=2$ when $\eta =2$ with $E\left\Vert \varepsilon _{t}\right\Vert
^{2}<\infty $. In the former case, let $p<2$. In order to lighten up the
notation, we omit the subscripts $j$, $h$ and $k$, and study%
\begin{equation*}
\sum_{m=0}^{\infty }C_{m}^{\ast }\sum_{t=1}^{T}\sum_{s=1}^{t}\varepsilon
_{s}\varepsilon _{t-m}.
\end{equation*}%
It holds that 
\begin{align*}
\sum_{m=0}^{\infty }C_{m}^{\ast }\sum_{t=1}^{T}\sum_{s=1}^{t}\varepsilon
_{s}\varepsilon _{t-m}& =\sum_{m=0}^{\infty }C_{m}^{\ast
}\sum_{t=1}^{T}\sum_{s=1}^{t}\varepsilon _{s}\varepsilon
_{t}+\sum_{m=0}^{\infty }C_{m}^{\ast
}\sum_{j=1}^{m}\sum_{t=1}^{T-m-j}\varepsilon _{t}\varepsilon _{t+j} \\
& -\sum_{m=0}^{\infty }C_{m}^{\ast }\sum_{j=1}^{m}\varepsilon
_{T+1-j}\sum_{t=1}^{T+1-j}\varepsilon _{t}+\sum_{m=0}^{\infty }C_{m}^{\ast
}\sum_{j=1}^{m}\varepsilon _{T+1-j}\sum_{i=1}^{j}\varepsilon _{T-m-i} \\
& =I+II+III+IV,
\end{align*}%
with the convention that $\sum_{j=1}^{0}=0$. Consider $I$; it holds that%
\begin{equation*}
I\leq c_{0}\left\vert \sum_{t=1}^{T}\sum_{s=1}^{t}\varepsilon
_{s}\varepsilon _{t}\right\vert ,
\end{equation*}%
by the summability of $C_{m}^{\ast }$. Also%
\begin{equation*}
\left\vert \sum_{t=1}^{T}\sum_{s=1}^{t}\varepsilon _{s}\varepsilon
_{t}\right\vert \leq \left\vert \sum_{t=1}^{T}\sum_{s=1}^{t-1}\varepsilon
_{s}\varepsilon _{t}\right\vert +\left\vert \sum_{t=1}^{T}\varepsilon
_{t}^{2}\right\vert =I_{a}+I_{b}.
\end{equation*}%
Similar steps as above yield%
\begin{equation*}
\left\vert \sum_{t=1}^{T}\varepsilon _{t}^{2}\right\vert _{p/2}^{p/2}\leq
\sum_{t=1}^{T}\left\vert \varepsilon _{t}\right\vert _{p}^{p}\leq c_{0}T,
\end{equation*}%
so that Lemma \ref{stout} entails that $I_{b}=o_{a.s.}\left( T^{2/p}\left(
\ln T\right) ^{2\left( 1+\epsilon \right) /p}\right) $ for all $\epsilon >0$%
. Turning to $I_{a}$, when $p\leq 1$%
\begin{equation*}
\left\vert \sum_{t=1}^{T}\sum_{s=1}^{t-1}\varepsilon _{s}\varepsilon
_{t}\right\vert _{p}^{p}\leq \sum_{t=1}^{T}\sum_{s=1}^{t-1}\left\vert
\varepsilon _{s}\varepsilon _{t}\right\vert
_{p}^{p}=\sum_{t=1}^{T}\sum_{s=1}^{t-1}\left\vert \varepsilon
_{s}\right\vert _{p}^{p}\left\vert \varepsilon _{t}\right\vert _{p}^{p}\leq
c_{0}T^{2}.
\end{equation*}%
When $1<p\leq 2$, the von Bahr-Esseen inequality (\citealp*{bahr}) can be
applied, giving%
\begin{equation*}
\left\vert \sum_{t=1}^{T}\sum_{s=1}^{t-1}\varepsilon _{s}\varepsilon
_{t}\right\vert _{p}^{p}\leq c_{0}\sum_{t=1}^{T}\left\vert
\sum_{s=1}^{t-1}\varepsilon _{s}\varepsilon _{t}\right\vert
_{p}^{p}=c_{0}\sum_{t=1}^{T}\left\vert \sum_{s=1}^{t-1}\varepsilon
_{s}\right\vert _{p}^{p}\left\vert \varepsilon _{t}\right\vert _{p}^{p}\leq
c_{1}\sum_{t=1}^{T}\left\vert \sum_{s=1}^{t-1}\varepsilon _{s}\right\vert
_{p}^{p}\leq c_{2}\sum_{t=1}^{T}\sum_{s=1}^{t-1}\left\vert \varepsilon
_{s}\right\vert _{p}^{p}\leq c_{3}T^{2}.
\end{equation*}%
By Lemma \ref{stout}, it follows that, for all $0<p\leq 2$, $%
I_{a}=o_{a.s.}\left( T^{2/p}\left( \ln T\right) ^{\left( 1+\epsilon \right)
/p}\right) $. Turning to $II$, note that%
\begin{equation*}
\left\vert \sum_{m=1}^{\infty }C_{m}^{\ast
}\sum_{j=1}^{m}\sum_{t=1}^{T-j}\varepsilon _{t}\varepsilon _{t+j}\right\vert
_{p/2}^{p/2}\leq \sum_{m=1}^{\infty }\left\vert C_{m}^{\ast }\right\vert
^{p/2}\sum_{j=1}^{m}\sum_{t=1}^{T-j}\left\vert \varepsilon _{t}\varepsilon
_{t+j}\right\vert _{p/2}^{p/2}\leq c_{0}T\sum_{m=1}^{\infty }m\left\vert
C_{m}^{\ast }\right\vert ^{p/2}\leq c_{0}T,
\end{equation*}%
so that $II=o_{a.s.}\left( T^{2/p}\left( \ln T\right) ^{2\left( 1+\epsilon
\right) /p}\right) $. Finally, for $III$ we apply the same logic as above to
get%
\begin{equation*}
\left\vert \sum_{m=1}^{\infty }C_{m}^{\ast }\sum_{j=0}^{m-1}\varepsilon
_{T-j}\sum_{t=1}^{T}\varepsilon _{t}\right\vert _{p/2}^{p/2}\leq
c_{0}\sum_{m=1}^{\infty }\left\vert C_{m}^{\ast }\right\vert
^{p/2}\sum_{j=0}^{m-1}\sum_{t=1}^{T}\left\vert \varepsilon _{T-j}\varepsilon
_{t}\right\vert _{p/2}^{p/2}\leq c_{0}T,
\end{equation*}%
which yields $III=o_{a.s.}\left( T^{2/p}\left( \ln T\right) ^{2\left(
1+\epsilon \right) /p}\right) $ as above. Finally, note that, as above%
\begin{align*}
\left\vert \sum_{m=0}^{\infty }C_{m}^{\ast }\sum_{j=1}^{m}\varepsilon
_{T+1-j}\sum_{i=1}^{j}\varepsilon _{T-m-i}\right\vert _{p/2}^{p/2}& \leq
\sum_{m=0}^{\infty }C_{m}^{\ast }\sum_{j=1}^{m}\left\vert \varepsilon
_{T+1-j}\sum_{i=1}^{j}\varepsilon _{T-m-i}\right\vert _{p/2}^{p/2} \\
& =\sum_{m=0}^{\infty }C_{m}^{\ast }\sum_{j=1}^{m}\left\vert \varepsilon
_{T+1-j}\right\vert _{p/2}^{p/2}\sum_{i=1}^{j}\left\vert \varepsilon
_{T-m-i}\right\vert _{p/2}^{p/2} \\
& \leq c_{0}\sum_{m=0}^{\infty }C_{m}^{\ast }\sum_{j=1}^{m}j\leq c_{1},
\end{align*}%
again by the exponential summability of $C_{m}^{\ast }$; this entails that $%
IV=o_{a.s.}\left( \left( \ln T\right) ^{2\left( 1+\epsilon \right)
/p}\right) $. The lemma follows by putting everything together; the case $%
\eta =2$ with finite variance follows from exactly the same calculations,
with $p=2$. 
\end{proof}

\setcounter{subsection}{-1} \setcounter{equation}{0} %
\renewcommand{\theequation}{D.\arabic{equation}}
\end{lemma}

\begin{lemma}
\label{mixing-alpha}We assume that Assumptions \ref{as-B}-\ref{as-deltayt}
are satisfied. Then it holds that $z_{t}=l^{\prime }\Delta y_{t}$ is a
strong mixing sequence with mixing numbers $\alpha _{m}=O\left( \rho
^{m}\right) $, for all $l$ and some $0<\rho <1$.

\begin{proof}
In order to prove the lemma, recall that, by Assumption \ref{as-B}, we have
the following equivalent representations:%
\begin{equation}
\Delta y_{t}=C\left( L\right) \varepsilon _{t}=C\varepsilon _{t}+C^{\ast
}\left( L\right) \varepsilon _{t}-C^{\ast }\left( L\right) \varepsilon _{t-1}
\label{vma}
\end{equation}%
where $C^{\ast }\left( L\right) $ is an invertible filter, and there exists
an $N\times R$ matrix $b$ of full rank $R$ such that $b^{\prime }C=0$. 

We prove the lemma by considering separately the cases $m=0$, $m=N$ and $%
0<m<N$. When $m=N$, (\ref{vma}) boils down to%
\begin{equation*}
\Delta y_{t}=C^{\ast }\left( L\right) \varepsilon _{t}-C^{\ast }\left(
L\right) \varepsilon _{t-1}.
\end{equation*}%
We start by showing that $C^{\ast }\left( L\right) \varepsilon _{t}$ is
strong mixing with geometrically declining mixing numbers, by verifying the
conditions of Lemma 2.2 and Theorem 2.1 in \citet*{phamtran}. Firstly,
Assumption \ref{as-deltayt} ensures that the error term $\varepsilon _{t}$
satisfies condition (i) in Lemma 2.2 in \citet*{phamtran}. Also, $C^{\ast
}\left( L\right) $ is invertible,\footnote{%
We point out that in the paper by \citet*{phamtran}, condition (ii) in their
Lemma 2.2 contains a typo, as it requires (using our notation) that $%
\sum_{j=0}^{\infty }C_{j}^{\ast }z^{j}\neq 0$ for all $\left\vert
z\right\vert \leq 1$ - upon inspecting the proof, this is required in order
to invert the polynomial $C^{\ast }\left( L\right) $, and thus it should
instead read $\det \left( \sum_{j=0}^{\infty }C_{j}^{\ast }z^{j}\right) \neq
0$ for all $\left\vert z\right\vert \leq 1$, where $\det \left( A\right) $
is the determinant of matrix $A$.} and, using Assumption \ref{as-B}, it is
easy to see that $\sum_{j=0}^{\infty }\left\Vert C_{j}^{\ast }\right\Vert
<\infty $. Hence, Theorem 2.1 in \citet*{phamtran} can be applied, yielding
that $C^{\ast }\left( L\right) \varepsilon _{t}$ is a strongly mixing
sequence with geometrically declining mixing numbers. Using Theorem 14.1 in %
\citet*{davidson}, it follows that both $\Delta y_{t}$ and $l^{\prime
}\Delta y_{t}$\ are also strong mixing sequence with geometrically declining
mixing numbers. When $m=0$, the same arguments can be applied to (%
\ref{vma}), noting that in this case $C\left( L\right) $ is invertible, and
that the condition $\sum_{j=0}^{\infty }\left\Vert C_{j}\right\Vert <\infty $
follows from Assumption \ref{as-B}, which entails that $C_{j}$ declines
geometrically. Finally, consider the case $0<m<N$. Since the $N\times N$ matrix $\left(
b,b_{\perp }\right) $ has full rank, any $l\in 
\mathbb{R}
^{N}$ can be expressed as $l=bv_{1}+b_{\perp }v_{2}$, where $v_{1}$ is $%
R\times 1$ and $v_{2}$ is $m\times 1$. Therefore it holds that%
\begin{eqnarray*}
z_{t} &=&l^{\prime }\Delta y_{t}=v_{1}^{\prime }b^{\prime }\Delta
y_{t}+v_{2}^{\prime }b_{\perp }^{\prime }\Delta y_{t} \\
&=&l^{\prime }\left[ C^{\ast }\left( L\right) \varepsilon _{t}-C^{\ast
}\left( L\right) \varepsilon _{t-1}\right] +v_{2}^{\prime }b_{\perp
}^{\prime }C\varepsilon _{t}.
\end{eqnarray*}%
having used (\ref{vma}). We already know from above that $l^{\prime }\left[
C^{\ast }\left( L\right) \varepsilon _{t}-C^{\ast }\left( L\right)
\varepsilon _{t-1}\right] $ is strong mixing with geometrically declining
mixing numbers; the same (obviously) applies to the \textit{i.i.d.} sequence $v_{2}^{\prime
}b_{\perp }^{\prime }C\varepsilon _{t}$. The \ desired result
follows again from Theorem 14.1 in \citet*{davidson}.
\end{proof}
\end{lemma}

\renewcommand{\thesection}{C} 
\setcounter{equation}{0} \setcounter{lemma}{0} \renewcommand{\thelemma}{C.%
\arabic{lemma}} \renewcommand{\theequation}{C.\arabic{equation}}

\section{Proofs\label{proofs}}

Henceforth, we let $E^{\ast }$ and $V^{\ast }$\ denote the expected value
and the variance with respect to $P^{\ast }$ respectively.

\begin{proof}[Proof of Proposition \protect\ref{s11-lambda}]
The result follows from Weyl's inequality (see e.g. \citealp*{hornjohnson})
and Lemmas \ref{b-xt}, \ref{ut-2} and \ref{cross-prod}. Indeed, it holds that%
\begin{align}
\lambda ^{\left( j\right) }\left( S_{11}\right) & \geq  & & \lambda ^{\left(
j\right) }\left( P\sum_{t=1}^{T}x_{t}x_{t}^{\prime }P^{\prime }\right)
+\lambda ^{\left( \min \right) }\left( \sum_{t=1}^{T}u_{t}u_{t}^{\prime
}+P\sum_{t=1}^{T}x_{t}u_{t}^{\prime }+\sum_{t=1}^{T}u_{t}x_{t}^{\prime
}P^{\prime }\right) ,  \label{weyl-1} \\
\lambda ^{\left( j\right) }\left( S_{11}\right) & \leq  & & \lambda ^{\left(
j\right) }\left( P\sum_{t=1}^{T}x_{t}x_{t}^{\prime }P^{\prime }\right)
+\lambda ^{\left( \max \right) }\left( \sum_{t=1}^{T}u_{t}u_{t}^{\prime
}+P\sum_{t=1}^{T}x_{t}u_{t}^{\prime }+\sum_{t=1}^{T}u_{t}x_{t}^{\prime
}P^{\prime }\right) .  \label{weyl-2}
\end{align}%
Further, it holds that%
\begin{align}
& \lambda ^{\left( \max \right) }\left( \sum_{t=1}^{T}u_{t}u_{t}^{\prime
}+P\sum_{t=1}^{T}x_{t}u_{t}^{\prime }+\sum_{t=1}^{T}u_{t}x_{t}^{\prime
}P^{\prime }\right)   \notag \\
& \leq \lambda ^{\left( \max \right) }\left(
\sum_{t=1}^{T}u_{t}u_{t}^{\prime }\right) +\lambda ^{\left( \max \right)
}\left( P\sum_{t=1}^{T}x_{t}u_{t}^{\prime }+\sum_{t=1}^{T}u_{t}x_{t}^{\prime
}P^{\prime }\right) .  \label{weyl-3}
\end{align}%
In light of (\ref{weyl-3}), Lemmas \ref{ut-2} and \ref{cross-prod} entail
that 
\begin{equation*}
\lambda ^{\left( \max \right) }\left( \sum_{t=1}^{T}u_{t}u_{t}^{\prime
}+P\sum_{t=1}^{T}x_{t}u_{t}^{\prime }+\sum_{t=1}^{T}u_{t}x_{t}^{\prime
}P^{\prime }\right) =o_{a.s.}\left( T^{2/p}\left( \ln T\right) ^{2\left(
2+\epsilon \right) /p}\right) ,
\end{equation*}%
for all $0\leq p<\eta $ when $\eta \leq 2$ with $E\left\Vert \varepsilon
_{t}\right\Vert ^{\eta }=\infty $, and $p=2$ when $\eta =2$ with $%
E\left\Vert \varepsilon _{t}\right\Vert ^{\eta }<\infty $. Consider (\ref%
{weyl-1}); when $0\leq j\leq m$, the term that dominates is $\lambda
^{\left( j\right) }\left( P\sum_{t=1}^{T}x_{t}x_{t}^{\prime }P^{\prime
}\right) $, and (\ref{s11-lambda-1}) follows immediately from Lemma \ref%
{b-xt}. When $j>m$, $\lambda ^{\left( j\right) }\left(
P\sum_{t=1}^{T}x_{t}x_{t}^{\prime }P^{\prime }\right) =0$ by construction,
and therefore (\ref{s11-lambda-2}) follows from (\ref{weyl-2}). 
\end{proof}

\begin{proof}[Proof of Proposition \protect\ref{s00-lambda}]
Lemma \ref{mixing-alpha} entails that, for all $l\neq 0$, $z_{t}=l^{\prime
}\Delta y_{t}$ is strong mixing with geometrically declining mixing numbers.
Thence, (\ref{s00-lambda-max}) follows immediately from Theorem 2.1 in \citet%
*{trapaniLIL}. Consider now (\ref{s00-lambda-min}), and define the sequence%
\begin{equation*}
\varphi _{T}=\frac{T^{1/\eta }}{\left( \ln T\right) ^{\left( 2+\epsilon
\right) /\eta }}.
\end{equation*}%
It holds that%
\begin{eqnarray*}
\sum_{t=1}^{T}z_{t}^{2} &=&\sum_{t=1}^{T}z_{t}^{2}I\left( \left\vert
z_{t}\right\vert >\varphi _{T}\right) +\sum_{t=1}^{T}z_{t}^{2}I\left(
\left\vert z_{t}\right\vert \leq \varphi _{T}\right)  \\
&\geq &\sum_{t=1}^{T}z_{t}^{2}I\left( \left\vert z_{t}\right\vert >\varphi
_{T}\right) \geq \varphi _{T}^{2}\sum_{t=1}^{T}I\left( \left\vert
z_{t}\right\vert >\varphi _{T}\right) .
\end{eqnarray*}%
Define the array $\omega _{T,t}=I\left( \left\vert z_{t}\right\vert >\varphi
_{T}\right) $. We define the sigma-field $^{\omega }\mathcal{F}%
_{T,0}^{t}=\sigma \left( \omega _{T,t},..,\omega _{T,0}\right)
=\bigcup\limits_{i=0}^{t}\sigma \left( \omega _{Ti}\right) $. Note that, by
construction, $\sigma \left( \omega _{T,i}\right) =\left\{ \varnothing
,\Omega ,\left[ -\varphi _{T},\varphi _{T}\right] ,\left\{ \left( -\infty
,-\varphi _{T}\right) \cup \left( \varphi _{T},\infty \right) \right\}
\right\} $, so that $\sup_{T}\sigma \left( \omega _{T,i}\right) \subset
\sigma \left( z_{i}\right) $. This entails that $\sup_{T}\left( ^{\omega }%
\mathcal{F}_{T,0}^{t}\right) \subset \bigcup\limits_{i=0}^{t}\sigma \left(
z_{i}\right) $, so that%
\begin{equation*}
\alpha _{k}^{\omega }=\sup_{T}\sup_{t}\alpha \left( ^{\omega }\mathcal{F}%
_{T,0}^{t},^{\omega }\mathcal{F}_{T,t+k}^{\infty }\right) \leq \alpha _{k}%
\text{,}
\end{equation*}%
where $\alpha \left( ^{\omega }\mathcal{F}_{T,0}^{t},^{\omega }\mathcal{F}%
_{T,t+k}^{\infty }\right) $ is the mixing number of order $k$, and $\alpha
_{k}$ is the mixing number of order $k$ associated to the sequence $z_{t}$.
Since $\alpha _{k}=O\left( \rho ^{k}\right) $, $0<\rho <1$, it also follows
that the sequence $\omega _{T,t}$ is strong mixing with geometrically
declining mixing numbers. Note also that $\omega _{T,t}$ has finite moments
up to any order; also, by Assumption \ref{as-moments}\textit{(i)}, it holds
that $E\omega _{T,t}=\left( c_{l,1}+c_{l,2}\right) \varphi _{T}^{-\eta
}L\left( \varphi _{T}\right) \left( 1+o\left( 1\right) \right) $. Letting $%
\overline{\omega }_{T,t}=\omega _{T,t}-E\omega _{T,t}$, equation (5.1) in %
\citet*{rio1995} ensures that%
\begin{equation}
P\left( \max_{1\leq s\leq T}\left\vert \sum_{t=1}^{s}\varphi _{T}^{\eta }%
\overline{\omega }_{T,t}\right\vert >2x_{T}\right) \leq \frac{c_{0}}{%
x_{T}^{2}}T\int_{0}^{1}\left( \alpha _{k}^{\omega }\right) ^{-1}\left( \frac{%
u}{2}\right) Q_{w}^{2}\left( u\right) du,  \label{rio5.1}
\end{equation}%
where%
\begin{equation*}
x_{T}=\sum_{t=1}^{T}E\left( \varphi _{T}^{\eta }\omega _{T,t}\right) ,
\end{equation*}%
and $Q_{w}\left( u\right) $ is the quantile function of $\varphi _{T}^{\eta }%
\overline{\omega }_{T,t}$. Consider now the function $f\left( x\right) =x\ln
\left( 1+x\right) $, and let 
\begin{equation*}
f^{\ast }\left( y\right) =\sup_{x>0}\left( xy-f\left( x\right) \right) 
\end{equation*}%
be its Young dual function (see Appendix A in \citealp*{rio1995}); as $%
x\rightarrow \infty $, $f^{\ast }\left( y\right) =\exp \left( y\right) $.
Then it holds that%
\begin{equation*}
\int_{0}^{1}\left( \alpha _{k}^{\omega }\right) ^{-1}\left( \frac{u}{2}%
\right) Q_{w}^{2}\left( u\right) du\leq \left\Vert \left( \alpha
_{k}^{\omega }\right) ^{-1}\left( U\right) \right\Vert _{f^{\ast
}}\left\Vert \left( \varphi _{T}^{\eta }\overline{\omega }_{T,t}\right)
^{2}\right\Vert _{f},
\end{equation*}%
where $U$ is a random variable with a uniform distribution on $\left[ 0,1%
\right] $ and $\left\Vert X\right\Vert _{f}$ is the Luxemburg norm (\citealp*%
{luxemburg}) of a random variable $X$ with respect to the function $f$, i.e.%
\begin{equation*}
\left\Vert X\right\Vert _{f}=\inf \left\{ c>0:Ef\left( c^{-1}\left\vert
X\right\vert \right) \leq 1\right\} .
\end{equation*}%
Equation (1.29) in \citet*{rio1995} ensures that $\left\Vert \alpha
^{-1}\left( U\right) \right\Vert _{f^{\ast }}<\infty $. Also%
\begin{equation*}
\left\Vert \left( \varphi _{T}^{\eta }\overline{\omega }_{T,t}\right)
^{2}\right\Vert _{f}=\inf \left\{ c>0:Ef\left( c^{-1}\left\vert \varphi
_{T}^{\eta }\overline{\omega }_{T,t}\right\vert ^{2}\right) \leq 1\right\} .
\end{equation*}%
Note that we have 
\begin{equation*}
Ef\left( c^{-1}\left\vert \varphi _{T}^{\eta }\overline{\omega }%
_{T,t}\right\vert ^{2}\right) =c^{-1}\left( \varphi _{T}^{\eta }-1\right)
^{2}\varphi _{T}^{-\eta }\ln \left( 1+c^{-1}\left( \varphi _{T}^{\eta
}-1\right) ^{2}\right) +c^{-1}\left( 1-\varphi _{T}^{-\eta }\right) \ln
\left( 1+\frac{1}{c}\right) ;
\end{equation*}%
the choice $c=\varphi _{T}^{\eta }\ln \left( \varphi _{T}^{2\eta }\right) $
can be shown to correspond to $Ef\left( c^{-1}\left\vert \varphi _{T}^{\eta }%
\overline{\omega }_{T,t}\right\vert ^{2}\right) <1$ after some algebra.
Therefore%
\begin{equation*}
\int_{0}^{1}\left( \alpha _{k}^{\omega }\right) ^{-1}\left( \frac{u}{2}%
\right) Q_{w}^{2}\left( u\right) du=O\left( \varphi _{T}^{\eta }\ln \left(
\varphi _{T}\right) \right) .
\end{equation*}%
Noting that, by definition, $x_{T}=T$, equation (\ref{rio5.1}) yields%
\begin{equation*}
P\left( \max_{1\leq s\leq T}\left\vert \sum_{t=1}^{s}\varphi _{T}^{\eta }%
\overline{\omega }_{T,t}\right\vert >2x_{T}\right) \leq c_{0}\frac{1}{\left(
\ln T\right) ^{1+\epsilon }},
\end{equation*}%
which entails that%
\begin{equation*}
\sum_{T=1}^{\infty }\frac{1}{T}P\left( \max_{1\leq s\leq T}\left\vert
\sum_{t=1}^{s}\varphi _{T}^{\eta }\overline{\omega }_{T,t}\right\vert
>2x_{T}\right) <\infty .
\end{equation*}%
By Lemma \ref{stout}, this entails that%
\begin{equation*}
\lim \sup_{T\rightarrow \infty }\frac{\sum_{t=1}^{T}\varphi _{T}^{\eta }%
\overline{\omega }_{T,t}}{\sum_{t=1}^{T}E\left( \varphi _{T}^{\eta }\omega
_{T,t}\right) }=0\text{ a.s.,}
\end{equation*}%
so that we can write%
\begin{equation*}
\frac{\sum_{t=1}^{T}\varphi _{T}^{\eta }\omega _{T,t}}{\sum_{t=1}^{T}E\left(
\varphi _{T}^{\eta }\omega _{T,t}\right) }=1+o_{a.s.}\left( 1\right) .
\end{equation*}%
Putting all together, (\ref{s00-lambda-min}) obtains.
\end{proof}

\begin{proof}[Proof of Theorem \protect\ref{s11s00}]
Consider (\ref{lambdaj-1}). By Theorem 7 in \citet*{merikoski2004}%
\begin{equation*}
\lambda ^{\left( j\right) }\left( S_{00}^{-1}S_{11}\right) \geq \lambda
^{\left( j\right) }\left( S_{11}\right) \lambda ^{\left( \min \right)
}\left( S_{00}^{-1}\right) =\frac{\lambda ^{\left( j\right) }\left(
S_{11}\right) }{\lambda ^{\left( \max \right) }\left( S_{00}\right) };
\end{equation*}%
combining (\ref{s11-lambda-1}) and (\ref{s00-lambda-max}), we obtain (\ref%
{lambdaj-1}). Turning to (\ref{lambdaj-2}), applying again Theorem 7 in %
\citet*{merikoski2004}, it holds that%
\begin{equation*}
\lambda ^{\left( j\right) }\left( S_{00}^{-1}S_{11}\right) \leq \lambda
^{\left( j\right) }\left( S_{11}\right) \lambda ^{\left( \max \right)
}\left( S_{00}^{-1}\right) =\frac{\lambda ^{\left( j\right) }\left(
S_{11}\right) }{\lambda ^{\left( \min \right) }\left( S_{00}\right) };
\end{equation*}%
thence, (\ref{lambdaj-2}) follows from (\ref{s11-lambda-2}) and (\ref%
{s00-lambda-min}). 
\end{proof}

\begin{proof}[Proof of Theorem \protect\ref{theta}]
The proof follows exactly the same lines as in \citet*{HT16} and we
therefore omit most passages. Under $H_{0}$, it follows from Theorem \ref%
{s11s00} that%
\begin{equation}
P\left( \omega :\;\lim_{T\rightarrow \infty }\exp \left\{ -T^{1-\kappa
-\epsilon }\right\} \phi _{T}^{\left( j\right) }=\infty \right) =1,
\label{as-lim-phi-null-2}
\end{equation}%
for every $j$ and any $\epsilon >0$. Let $\Phi \left( \cdot \right) $\
denote the standard normal distribution, and note that 
\begin{equation*}
E^{\ast }\zeta _{i}^{\left( j\right) }=\Phi \left( u/\phi _{T}^{\left(
j\right) }\right) \;\;\;\mbox{and}\;\;\;E^{\ast }\left( \zeta _{i}^{\left(
j\right) }-E^{\ast }\zeta _{i}^{\left( j\right) }\right) ^{2}=\Phi \left(
u/\phi _{T}^{\left( j\right) }\right) \left[ 1-\Phi \left( u/\phi
_{T}^{\left( j\right) }\right) \right] .
\end{equation*}%
We can write%
\begin{align*}
M^{-1/2}\sum_{i=1}^{M}\left( \zeta _{i}^{\left( j\right) }-\frac{1}{2}%
\right) & =M^{-1/2}\sum_{i=1}^{M}\left[ I\left( \xi _{i}^{\left( j\right)
}\leq 0\right) -\frac{1}{2}\right] +M^{-1/2}\sum_{i=1}^{M}\left( \Phi \left(
u/\phi _{T}^{\left( j\right) }\right) -\frac{1}{2}\right) \\
& +M^{-1/2}\sum_{i=1}^{M}\left[ I\{\left( \xi _{i}^{\left( j\right) }\leq
u/\phi _{T}^{\left( j\right) }\right) -I\left( \xi _{i}^{\left( j\right)
}\leq 0\right) -\left( \Phi \left( u/\phi _{T}^{\left( j\right) }\right) -%
\frac{1}{2}\right) \right] .
\end{align*}%
It holds that%
\begin{align*}
& E^{\ast }\left( M^{-1/2}\sum_{i=1}^{M}\left[ I\left( \xi _{i}^{\left(
j\right) }\leq u/\phi _{T}^{\left( j\right) }\right) -I\left( \xi
_{i}^{\left( j\right) }\leq 0\right) -\left( \Phi \left( u/\phi _{T}^{\left(
j\right) }\right) -\frac{1}{2}\right) \right] \right) ^{2} \\
& =E^{\ast }\left[ I\left( \xi _{1}^{\left( j\right) }\leq u/\phi
_{T}^{\left( j\right) }\right) -I\left( \xi _{1}^{\left( j\right) }\leq
0\right) -\left( \Phi \left( u/\phi _{T}^{\left( j\right) }\right) -\frac{1}{%
2}\right) \right] ^{2},
\end{align*}%
on account of the independence of the $\xi _{i}^{\left( j\right) }$s across $%
i$. Given that 
\begin{align}
E^{\ast }\left( I\left( \xi _{1}^{\left( j\right) }\leq u/\phi _{T}^{\left(
j\right) }\right) -I\left( \xi _{1}^{\left( j\right) }\leq 0\right) \right)
=& \Phi \left( u/\phi _{T}^{\left( j\right) }\right) -\frac{1}{2},
\label{exp} \\
V^{\ast }\left( I\left( \xi _{1}^{\left( j\right) }\leq u/\phi _{T}^{\left(
j\right) }\right) -I\left( \xi _{1}^{\left( j\right) }\leq 0\right) \right)
\leq & \Phi \left( u/\phi _{T}^{\left( j\right) }\right) -\frac{1}{2},
\label{var}
\end{align}%
we have%
\begin{equation*}
E^{\ast }\left[ I\left( \xi _{1}^{\left( j\right) }\leq u/\phi _{T}^{\left(
j\right) }\right) -I\left( \xi _{1}^{\left( j\right) }\leq 0\right) -\left(
\Phi \left( u/\phi _{T}^{\left( j\right) }\right) -\frac{1}{2}\right) \right]
^{2}\leq \Phi \left( u/\phi _{T}^{\left( j\right) }\right) -\frac{1}{2}\leq 
\frac{|u|}{\sqrt{2\pi }\phi _{T}^{\left( j\right) }}<\infty .
\end{equation*}%
Thus%
\begin{equation*}
\int_{U}E^{\ast }\left[ I\left( \xi _{1}^{\left( j\right) }\leq u/\phi
_{T}^{\left( j\right) }\right) -I\left( \xi _{1}^{\left( j\right) }\leq
0\right) -\left( \Phi \left( u/\phi _{T}^{\left( j\right) }\right) -\frac{1}{%
2}\right) \right] ^{2}dF\left( u\right) \leq \int_{U}\frac{|u|}{\sqrt{2\pi }%
\phi _{T}^{\left( j\right) }}dF\left( u\right) .
\end{equation*}%
Also, it follows immediately that%
\begin{equation*}
M^{-1/2}\sum_{i=1}^{M}\left( \Phi \left( u/\phi _{T}^{\left( j\right)
}\right) -\frac{1}{2}\right) \leq \left( \frac{M^{1/2}|u|}{\sqrt{2\pi }\phi
_{T}^{\left( j\right) }}\right) ^{2},
\end{equation*}%
so that%
\begin{equation*}
\int_{U}\left[ M^{-1/2}\sum_{i=1}^{M}\left( \Phi \left( u/\phi _{T}^{\left(
j\right) }\right) -\frac{1}{2}\right) \right] ^{2}dF\left( u\right) \leq
\int_{U}\left( \frac{M^{1/2}|u|}{\sqrt{2\pi }\phi _{T}^{\left( j\right) }}%
\right) ^{2}dF\left( u\right) .
\end{equation*}%
Using Assumption \ref{weight}\textit{(ii)}, (\ref{as-lim-phi-null-2}), (\ref%
{rate-constraint}) and Markov's inequality, it is easy to see that 
\begin{equation*}
\Theta _{T,M}^{\left( j\right) }=\left\{ \frac{2}{\sqrt{M}}\sum_{i=1}^{M}%
\left[ I\left( \xi _{i}^{\left( j\right) }\leq 0\right) -\frac{1}{2}\right]
\right\} ^{2}+o_{P^{\ast }}(1),
\end{equation*}%
and therefore the desired result follows from the Central Limit Theorem.
Under $H_{A}$, Theorem \ref{s11s00} entails that%
\begin{equation}
P\left( \omega :\;\lim_{T\rightarrow \infty }\phi _{T}^{\left( j\right)
}=0\right) =1.  \label{lim-asy-alt}
\end{equation}%
Noting that 
\begin{equation*}
\zeta _{i}^{(j)}-\frac{1}{2}=I\left( \xi _{1}^{\left( j\right) }\leq u/\phi
_{T}^{\left( j\right) }\right) \pm \Phi \left( u/\phi _{T}^{\left( j\right)
}\right) -\frac{1}{2},
\end{equation*}%
and we have 
\begin{equation}
E^{\ast }\left[ M^{-1/2}\sum_{i=1}^{M}\left( I\left( \xi _{1}^{\left(
j\right) }\leq u/\phi _{T}^{\left( j\right) }\right) -\frac{1}{2}\right) %
\right] ^{2}=E^{\ast }\left[ \left( I\left( \xi _{1}^{\left( j\right) }\leq
u/\phi _{T}^{\left( j\right) }\right) -\Phi \left( u/\phi _{T}^{\left(
j\right) }\right) \right) \right] ^{2}+M\left[ \Phi \left( u/\phi
_{T}^{\left( j\right) }\right) -\frac{1}{2}\right] ^{2}.  \notag
\end{equation}%
Since $E^{\ast }\left[ \left( I\left( \xi _{1}^{\left( j\right) }\leq u/\phi
_{T}^{\left( j\right) }\right) -\Phi \left( u/\phi _{T}^{\left( j\right)
}\right) \right) \right] ^{2}<\infty $, by the Markov inequality it follows
that 
\begin{equation*}
M^{-1/2}\sum_{i=1}^{M}\left( I\left( \xi _{1}^{\left( j\right) }\leq u/\phi
_{T}^{\left( j\right) }\right) -\Phi \left( u/\phi _{T}^{\left( j\right)
}\right) \right) =O_{P^{\ast }}(1).
\end{equation*}%
The desired result now follows from (\ref{lim-asy-alt}).

\end{proof}

\begin{proof}[Proof of Theorem \protect\ref{family}]
The proof is the same as the proof of Theorem 3 in \citet*{trapani17}, and
therefore we only report the main passages. Consider the events $\left\{ 
\widehat{m}=N-j\right\} $, $0\leq j\leq N$, and recall that the test
statistics $\Theta _{T,M}^{\left( j\right) }$ are independent across $j$
conditional on the sample. Thus, for all $j\leq N-m-1$, we have $P^{\ast
}\left( \widehat{m}=j\right) =\alpha \left( 1-\alpha \right) ^{j}$, where we
omit dependence on $N$ and $T$ from the nominal level $\alpha $ for the sake
of the notation. Letting $\pi $ be the power of the test, it easy to see
that $P^{\ast }\left( \widehat{m}=j\right) =\left( 1-\alpha \right)
^{N-m}\pi \left( 1-\pi \right) ^{j-N+m}$, whenever $j\geq N-m+1$. Thus%
\begin{equation*}
P^{\ast }\left( \widehat{m}=m\right) =1-\sum_{j\neq N-m}P^{\ast }\left( 
\widehat{m}=N-j\right) =\left( 1-\alpha \right) ^{N-m}\left[ \pi +\left(
1-\pi \right) ^{m+1}\right] .
\end{equation*}%
Note now that%
\begin{eqnarray*}
P^{\ast }\left( \Theta _{T,M}^{\left( j\right) }\leq c_{\alpha }\right) 
&=&P^{\ast }\left( \left\vert Z+\frac{M^{1/2}}{2}\right\vert ^{2}\leq
c_{\alpha }\right) +o_{P^{\ast }}\left( 1\right)  \\
&\leq &P^{\ast }\left( \left\vert Z\right\vert \leq c_{\alpha }^{1/2}-\frac{%
M^{1/2}}{2}\right) +o_{P^{\ast }}\left( 1\right) ,
\end{eqnarray*}%
which drifts to zero on account of the fact that $c_{\alpha }=o\left(
M\right) $. Thus, as $\min \left( T,M\right) \rightarrow \infty $ we have $%
\pi \overset{P^{\ast }}{\rightarrow }1$; hence, it follows that 
\begin{equation}
\lim_{\min \left( T,M\right) \rightarrow \infty }P^{\ast }\left( \widehat{m}%
=m\right) =\left( 1-\alpha \right) ^{N-m},  \label{procedure-size}
\end{equation}%
for almost all realisations of $\left\{ \varepsilon _{t},0<t<\infty \right\} 
$. The desired result follows by specialising (\ref{procedure-size}) to the
case where $\lim_{\min \left( T,M\right) \rightarrow \infty }\alpha =0$. 
\end{proof}

\begin{proof}[Proof of Corollary \protect\ref{constant}]
The result follows from the theory developed in \citet*{prhansen}. Indeed,
using the results in Section 3 in his paper, it follows that, under (\ref%
{var-const}), equation (\ref{bn}) still holds with%
\begin{equation*}
y_{t}=C\sum_{s=1}^{t}\varepsilon _{s}+C^{\ast }\left( L\right) \varepsilon
_{t},
\end{equation*}%
and therefore the conclusions of Theorem \ref{s11-lambda} remain unaltered.
Similarly, (\ref{ma}) also holds, and therefore Theorem \ref{s00-lambda}
also holds. The final results now follows immediately. 
\end{proof}

\begin{proof}[Proof of Corollary \protect\ref{heterosk}]
The proof follows from exactly the same passages (with one minor
modification, which we describe at the end of this proof) as in the proof of
Theorem \ref{theta}. In turn, this follows immediately if (\ref{s11-lambda-1}%
)-(\ref{s00-lambda-min}) hold. Equation (\ref{s00-lambda-max}) has already
been shown in \citet*{trapaniLIL}. Similarly, (\ref{s11-lambda-2}) and (\ref%
{s00-lambda-min}) can be shown by exactly the same passages as in the proofs
of Lemma \ref{ut-2} and Proposition \ref{s00-lambda}, with minor if tedious
complications arising from the presence of $h\left( \frac{t}{T}\right) $.

Proving (\ref{s11-lambda-1})\ requires more work; an a.s. result is
predicated on extending the paper by \citet*{jain}\ to the case of a
heterogeneous \textit{i.i.d.} sequence. This is entirely feasible (the key
requirement is independence, not homogeneity), but in order for the proof to
be transparent, the whole original paper would have to be replicated
allowing for the presence of $h\left( \frac{t}{T}\right) $. In order to
present a more transparent and self-contained argument, we instead begin by
showing a \textquotedblleft weak\textquotedblright , in probability, result,
which replaces Lemma \ref{donsker}. We show that, for all $b\in 
\mathbb{R}
^{N}$, it holds that%
\begin{equation}
\lim_{\delta \rightarrow 0^{+}}\lim \inf_{T\rightarrow \infty }P\left( \frac{%
1}{T^{1+2/\eta }}\sum_{t=1}^{T}\left( \sum_{i=1}^{t}b^{\prime }\varepsilon
_{i}\right) ^{2}>\delta \right) =1.  \label{infprob}
\end{equation}%
Note first that Assumption \ref{as-moments} entails that%
\begin{equation}
\frac{1}{T^{1/\eta }}\sum_{t=1}^{\left\lfloor Tr\right\rfloor }b^{\prime
}u_{t}\overset{w}{\underset{J_{1}}{\rightarrow }}Z_{\eta }\left( r\right) ,
\label{z-alpha}
\end{equation}%
for $r\in \left[ 0,1\right] $, where as usual \textquotedblleft $\overset{w}{%
\underset{J_{1}}{\rightarrow }}$\textquotedblright\ denotes weak convergence
in the space $D\left[ 0,1\right] $ equipped with the $J_{1}$-topology, and $%
Z_{\eta }\left( r\right) $ is a symmetric stable process of index $\eta $.
Consider 
\begin{equation}
\frac{1}{T^{1/\eta }}\sum_{t=1}^{\left\lfloor Tr\right\rfloor }b^{\prime
}\varepsilon _{t}  \label{maejima}
\end{equation}%
and note that the norming sequence in (\ref{z-alpha}) and (\ref{maejima}) is
exactly the same. This, and Assumption \ref{function}, entail that
Assumptions B and C on p. 165 in \citet*{maejima} are satisfied. Thus, it
holds that%
\begin{equation}
\frac{1}{T^{1/\eta }}\sum_{t=1}^{\left\lfloor Tr\right\rfloor }b^{\prime
}\varepsilon _{t}\overset{w}{\underset{J_{1}}{\rightarrow }}%
\int_{0}^{r}h\left( u\right) dZ_{\eta }\left( u\right) \equiv Y_{\eta
,b}\left( r\right) ,  \label{maejima2}
\end{equation}%
where we note that Assumption A in \citet*{maejima} holds immediately
because the norming sequence in (\ref{z-alpha}) and (\ref{maejima2}) is
exactly the same. By construction, $Y_{\eta ,b}\left( r\right) $ is a
(well-defined) symmetric stable process of index $\eta $ - the details are
in Chapter 3 in \citet*{taqqu}, see in particular Property 3.2.1. By the
Continuous Mapping Theorem%
\begin{equation}
\frac{1}{T^{1+2/\eta }}\sum_{t=1}^{T}\left( \sum_{i=1}^{t}b^{\prime
}\varepsilon _{i}\right) ^{2}\overset{w}{\underset{J_{1}}{\rightarrow }}%
\int_{0}^{1}Y_{\eta ,b}^{2}\left( r\right) dr.  \label{fclt2}
\end{equation}%
As mentioned above, the process $Y_{\alpha ,b}\left( r\right) $ is symmetric 
$\alpha -$stable; thus, using Lemma 2.2 in \citet*{shilower}, it is
immediate to see that 
\begin{equation}
\lim_{x\rightarrow 0^{+}}P\left( \int_{0}^{1}Y_{\eta ,b}^{2}\left( r\right)
dr>x\right) =1.  \label{shi}
\end{equation}%
Using (\ref{fclt2}) and (\ref{shi})%
\begin{equation*}
\lim_{\delta \rightarrow 0^{+}}\lim \inf_{T\rightarrow \infty }P\left( \frac{%
1}{T^{1+2/\eta }}\sum_{t=1}^{T}\left( \sum_{i=1}^{t}b^{\prime }\varepsilon
_{i}\right) ^{2}>\delta \right) \geq \lim_{\delta \rightarrow 0^{+}}P\left(
\int_{0}^{1}Y_{\eta ,b}^{2}\left( r\right) dr>\delta \right) =1,
\end{equation*}%
proving (\ref{infprob}). Now, by the same passages as in Lemma \ref{b-xt},
it holds that 
\begin{equation*}
\lim_{\delta \rightarrow 0^{+}}\lim \inf_{T\rightarrow \infty }P\left( \frac{%
1}{T^{1+2/\eta }}\lambda ^{\left( j\right) }\left(
P\sum_{t=1}^{T}x_{t}x_{t}^{\prime }P^{\prime }\right) >\delta \right) =1,
\end{equation*}%
for $j\leq m$. Combining this with (\ref{s00-lambda-max}), and following the
same logic as in the proof of (\ref{lambdaj-1}), it ultimately follows that%
\begin{equation}
\lim_{\delta \rightarrow 0^{+}}\lim \inf_{T\rightarrow \infty }P\left( \frac{%
\left( \ln T\right) ^{2\left( 2+\epsilon \right) /p}}{T^{1-\epsilon ^{\prime
}}}\lambda ^{\left( j\right) }\left( S_{00}^{-1}S_{11}\right) >\delta
\right) =1,  \label{eig-lb-p}
\end{equation}%
for $j\leq m$, where $\epsilon $, $\epsilon ^{\prime }$ and $p$ are defined
in (\ref{lambdaj-1}).

Now we are ready to prove the theorem. Equation (\ref{alt-convergence}) for
this case follows immediately from (\ref{lambdaj-2}), which is implied by (%
\ref{s11-lambda-2})-(\ref{s00-lambda-min}), using exactly the same passages
as in the original proof. As far as (\ref{size-het}) is concerned, note
that, by continuity, (\ref{eig-lb-p}) entails%
\begin{equation*}
\lim \inf_{T\rightarrow \infty }P\left( \exp \left\{ -T^{1-\kappa -\epsilon
}\right\} \phi _{T}^{\left( j\right) }=\infty \right) =1,
\end{equation*}%
for every $\epsilon >0$. The desired result now follows by applying the
proof Theorem \ref{theta} - see also the proof of Theorem C1 in \citet*{HT16}%
. 
\end{proof}

\begin{proof}[Proof of Theorem \protect\ref{bai04}]
We present the proof for the homoskedastic case. In the heteroskedastic
case, the proof follows from exactly the same arguments, and therefore we
omit them to save space. \newline
In light of the consistency of $\widehat{m}$, we prove the result assuming
that $m$ is known. Consider the $N\times T$ matrix $Y=\left[ y_{1}|...|y_{T}%
\right] $; we can write (\ref{trends}) in matrix form as%
\begin{equation*}
Y=PX+U,
\end{equation*}%
where $X=\left[ x_{1}|...|x_{T}\right] $ and $U=\left[ u_{1}|...|u_{T}\right]
$. Then $S_{11}=YY^{\prime }$ and we have the well-known identity%
\begin{equation*}
YY^{\prime }\widehat{P}=\widehat{P}V,
\end{equation*}%
where $V$ is an $m\times m$ diagonal matrix whose entries are $m$ largest
the eigenvalues of $YY^{\prime }$. Then it holds that%
\begin{eqnarray*}
\widehat{P} &=&\left( PX+U\right) \left( PX+U\right) ^{\prime }\widehat{P}%
V^{-1} \\
&=&PXX^{\prime }P^{\prime }\widehat{P}V^{-1}+PXU^{\prime }\widehat{P}%
V^{-1}+UX^{\prime }P^{\prime }\widehat{P}V^{-1}+UU^{\prime }\widehat{P}%
V^{-1}.
\end{eqnarray*}%
Let $XX^{\prime }P^{\prime }\widehat{P}V^{-1}=H$; it is easy to see that $%
\left\Vert H\right\Vert =O_{P}\left( 1\right) $. We have 
\begin{equation*}
\widehat{P}-PH=PXU^{\prime }\widehat{P}V^{-1}+UX^{\prime }P^{\prime }%
\widehat{P}V^{-1}+UU^{\prime }\widehat{P}V^{-1}=I+II+III.
\end{equation*}%
Consider $I$%
\begin{equation*}
\left\Vert PXU^{\prime }\widehat{P}V^{-1}\right\Vert \leq \left\Vert
P\right\Vert \left\Vert XU^{\prime }\right\Vert \left\Vert \widehat{P}%
\right\Vert \left\Vert V^{-1}\right\Vert ;
\end{equation*}%
clearly, $\left\Vert P\right\Vert $ and $\left\Vert \widehat{P}\right\Vert $
are both bounded. Also, by the FCLT (\citealp*{paulauskas}), $\left\Vert
V^{-1}\right\Vert =O_{P}\left( T^{-1-2/\eta }\right) $. Finally, using the
calculations in the proof of Lemma \ref{cross-prod}, it is easy to see that $%
\left\Vert XU^{\prime }\right\Vert =O_{P}\left( T^{2/\eta }\right) $. Thus $%
I=O_{P}\left( T^{-1}\right) $, and the same holds for $II$. Considering $III$%
, essentially the same arguments and the passages in the proof of Lemma \ref%
{ut-2} yield $III=O_{P}\left( T^{-1}\right) $. Putting all together, (\ref%
{loadings}) follows. Also%
\begin{equation*}
\left\Vert \widehat{x}_{t}-H^{-1}x_{t}\right\Vert \leq \left\Vert \widehat{P}%
^{\prime }\left( P-\widehat{P}H^{-1}\right) x_{t}\right\Vert +\left\Vert 
\widehat{P}^{\prime }u_{t}\right\Vert =I+II.
\end{equation*}%
Note that $E\left\Vert x_{t}\right\Vert ^{\eta }\leq
\sum_{j=1}^{t}E\left\Vert \varepsilon _{j}\right\Vert ^{\eta }=c_{0}t$ for
all $0<\eta \leq 2$, using the von Bahr-Esseen inequality (\citealp*{bahr})
when $1<\eta \leq 2$. Hence, standard arguments entail $\left\Vert
x_{t}\right\Vert =O_{P}\left( t^{1/\eta }\right) $. Thus, recalling (\ref%
{loadings}), it follows that $I=O_{P}\left( T^{-1+1/\eta }\right) $. Also,
recalling that $\left\Vert \widehat{P}^{\prime }\widehat{P}\right\Vert =1$,
it is easy to see that $II\leq \left\Vert \widehat{P}^{\prime }\right\Vert
\left\Vert u_{t}\right\Vert =O_{P}\left( 1\right) $. Putting all together,(%
\ref{factors}) obtains. Finally, note that $\max_{1\leq t\leq T}\left\Vert
u_{t}\right\Vert =O_{P}\left( T^{1/\eta }\right) $ by standard arguments.
Note also that, when $\eta \leq 1$%
\begin{equation*}
\max_{1\leq t\leq T}\left\Vert x_{t}\right\Vert \leq
\sum_{t=1}^{T}\left\Vert x_{t}\right\Vert \leq \left(
\sum_{t=1}^{T}\left\Vert x_{t}\right\Vert ^{\eta }\right) ^{1/\eta },
\end{equation*}%
by the $L_{p}$-norm inequality. Thus%
\begin{equation*}
E\left( \sum_{t=1}^{T}\left\Vert x_{t}\right\Vert ^{\eta }\right) ^{1/\eta
}\leq c_{0}T^{-1+1/\eta }E\left( \sum_{t=1}^{T}\left\Vert x_{t}\right\Vert
^{\eta }\right) \leq c_{1}T^{1+1/\eta },
\end{equation*}%
whence $\max_{1\leq t\leq T}\left\Vert x_{t}\right\Vert =O_{P}\left(
T^{1+1/\eta }\right) $. When $\eta >1$, we still use%
\begin{equation*}
\max_{1\leq t\leq T}\left\Vert x_{t}\right\Vert \leq
\sum_{t=1}^{T}\left\Vert x_{t}\right\Vert ,
\end{equation*}%
and note that%
\begin{equation*}
E\left( \sum_{t=1}^{T}\left\Vert x_{t}\right\Vert \right) ^{\eta }\leq
c_{0}T^{\eta -1}E\sum_{t=1}^{T}\left\Vert x_{t}\right\Vert ^{\eta }\leq
c_{1}T^{1+1/\eta }.
\end{equation*}%
Combining the above with (\ref{loadings}), (\ref{unif-x}) follows.%
\end{proof}

\begin{proof}[Proof of Theorem \protect\ref{family2}]

The proof is fairly similar to that of Theorem \ref{family}, and we therefore report
only the main passages. Consider the case $0<m<N$; it is easy to see that%
\begin{equation*}
P^{\ast }\left( \widetilde{m}=p\right) =\pi ^{N-p}\left( 1-\pi \right) ,
\end{equation*}%
for all $m+1\leq p\leq N$. Also, it can be shown, again by the conditional
independence of the test statistics across $p$, that%
\begin{equation*}
P^{\ast }\left( \widetilde{m}=p\right) =\pi ^{N-m}\alpha ^{m-p}\left(
1-\alpha \right) ,
\end{equation*}%
for $1\leq p\leq m-1$, and 
\begin{equation*}
P^{\ast }\left( \widetilde{m}=0\right) =\pi ^{N-m}\alpha ^{m}.
\end{equation*}%
Hence we have%
\begin{equation*}
P^{\ast }\left( \widetilde{m}=m\right) =1-\sum_{p\neq m}P^{\ast }\left( 
\widetilde{m}=p\right) =\pi ^{N-m}\left( 1-\alpha \right) ,
\end{equation*}%
which entails that, as $\alpha \rightarrow 0$, $P^{\ast }\left( \widetilde{m}%
=m\right) $ tends to one for almost all realisations of $\left\{ \varepsilon
_{t},1\leq t\leq T\right\} $. When $m=N$, by definition%
\begin{equation*}
P^{\ast }\left( \widetilde{m}=N\right) =1-\alpha ,
\end{equation*}%
whence again the desired result. Finally, when $m=0$, we have%
\begin{equation*}
P^{\ast }\left( \widetilde{m}=0\right) =\pi ^{N},
\end{equation*}%
and the desired result again obtains immediately. 
\end{proof}

\newpage

\renewcommand{\thesection}{S} 
\renewcommand{\thelemma}{S.\arabic{lemma}} \renewcommand{\thetheorem}{A%
\arabic{theorem}} \renewcommand{\theequation}{S.\arabic{equation}}

\textwidth 20.5cm \oddsidemargin 2cm \evensidemargin 2cm \textheight 20.3cm %
\topmargin 1cm

\begin{landscape}
\begin{table*}[h!]
\caption{Comparison between BCT and BIC - \textit{i.i.d.} case, Part 1}
\label{tab:Table1a}
\resizebox{\textwidth}{!}{%

\begin{tabular}{ll|cccccccccccccccccc}
$\theta =0$ & $n$ & \multicolumn{6}{|c}{$3$} & \multicolumn{6}{c}{$4$} & 
\multicolumn{6}{c}{$5$} \\ 
& $T$ & \multicolumn{2}{|c}{$100$} & \multicolumn{2}{c}{$200$} & 
\multicolumn{2}{c}{$400$} & \multicolumn{2}{|c}{$100$} & \multicolumn{2}{c}{$%
200$} & \multicolumn{2}{c}{$400$} & \multicolumn{2}{|c}{$100$} & 
\multicolumn{2}{c}{$200$} & \multicolumn{2}{c}{$400$} \\ \hline
& \multicolumn{1}{c|}{} & BCT & BIC & \multicolumn{1}{|c}{BCT} & BIC & 
\multicolumn{1}{|c}{BCT} & BIC & \multicolumn{1}{|c}{BCT} & BIC & BCT & BIC
& BCT & BIC & \multicolumn{1}{|c}{BCT} & BIC & BCT & BIC & BCT & BIC \\ 
$M$ & \multicolumn{1}{c|}{$\eta $} &  &  &  &  &  &  &  &  &  &  &  &  &  & 
&  &  &  &  \\ 
& \multicolumn{1}{c|}{} &  &  &  &  &  &  &  &  &  &  &  &  &  &  &  &  &  & 
\\
\hline
\end{tabular}

}

\begin{tablenotes}
  
\tiny
      \item Each entry in the table contains the average estimated rank across replications; the standard deviation (in round brackets); and the percentage of times the rank is estimated incorrectly (in square brackets).
      \item We report our estimator of $M$, denoted as \textit{BCT} and the number of common trends (implicitly) estimated by BIC, which is always the best information criterion in all cases. 
   
\end{tablenotes} 

\end{table*}

\newpage

\begin{table*}[h] 
\caption{Comparison between BCT and BIC - \textit{i.i.d.} case, Part 2} %
\label{tab:Table1b} {\resizebox{\textwidth}{!}{

\begin{tabular}{ll|cccccccccccccccccc}
$\theta =0$ & $n$ & \multicolumn{6}{|c}{$3$} & \multicolumn{6}{c}{$4$} & 
\multicolumn{6}{c}{$5$} \\ 
& $T$ & \multicolumn{2}{|c}{$100$} & \multicolumn{2}{c}{$200$} & 
\multicolumn{2}{c}{$400$} & \multicolumn{2}{|c}{$100$} & \multicolumn{2}{c}{$%
200$} & \multicolumn{2}{c}{$400$} & \multicolumn{2}{|c}{$100$} & 
\multicolumn{2}{c}{$200$} & \multicolumn{2}{c}{$400$} \\ \hline
& \multicolumn{1}{c|}{} & BCT & BIC & \multicolumn{1}{|c}{BCT} & BIC & 
\multicolumn{1}{|c}{BCT} & BIC & \multicolumn{1}{|c}{BCT} & BIC & BCT & BIC
& BCT & BIC & \multicolumn{1}{|c}{BCT} & BIC & BCT & BIC & BCT & BIC \\ 
$M$ & \multicolumn{1}{c|}{$\eta $} &  &  &  &  &  &  &  &  &  &  &  &  &  & 
&  &  &  &  \\ 
& \multicolumn{1}{c|}{} &  &  &  &  &  &  &  &  &  &  &  &  &  &  &  &  &  & 
\\
\hline
\end{tabular}

}
\begin{tablenotes}
  
\tiny
      \item The table complements Table \ref{tab:Table1a}; the numbers in the table have the same meaning as in the previous table. 
   
\end{tablenotes} 
}
\end{table*} 

\newpage

\begin{table*}[h] 
\caption{Comparison between BCT and BIC - positive autocorrelation case, Part 1} %
\label{tab:Table2a} {\resizebox{\textwidth}{!}{%

\begin{tabular}{lc|llllllllllllllllll}
$\theta =0.5$ & $n$ & \multicolumn{6}{|c}{$3$} & \multicolumn{6}{c}{$4$} & 
\multicolumn{6}{c}{$5$} \\ 
& $T$ & \multicolumn{2}{|c}{$100$} & \multicolumn{2}{c}{$200$} & 
\multicolumn{2}{c}{$400$} & \multicolumn{2}{|c}{$100$} & \multicolumn{2}{c}{$%
200$} & \multicolumn{2}{c}{$400$} & \multicolumn{2}{|c}{$100$} & 
\multicolumn{2}{c}{$200$} & \multicolumn{2}{c}{$400$} \\ \hline
&  & BCT & BIC & \multicolumn{1}{|l}{BCT} & BIC & \multicolumn{1}{|l}{BCT} & 
BIC & \multicolumn{1}{|l}{BCT} & BIC & BCT & BIC & BCT & BIC & 
\multicolumn{1}{|l}{BCT} & BIC & BCT & BIC & BCT & BIC \\ 
$M$ & $\eta $ &  &  &  &  &  &  &  &  &  &  &  &  &  &  &  &  &  &  \\ 
&  &  &  &  &  &  &  &  &  &  &  &  &  &  &  &  &  &  &  \\
\hline
\end{tabular}

}
\begin{tablenotes}
  
\tiny
      \item Each entry in the table contains the average estimated rank across replications; the standard deviation (in round brackets); and the percentage of times the rank is estimated incorrectly (in square brackets). 
   
\end{tablenotes} 
}
\end{table*} 

\newpage

\begin{table*}[h] 
\caption{Comparison between BCT and BIC - positive autocorrelation, Part 2} %
\label{tab:Table2b} {\resizebox{\textwidth}{!}{%

\begin{tabular}{lc|cccccccccccccccccc}
$\theta =0.5$ & $n$ & \multicolumn{6}{|c}{$3$} & \multicolumn{6}{c}{$4$} & 
\multicolumn{6}{c}{$5$} \\ 
& $T$ & \multicolumn{2}{|c}{$100$} & \multicolumn{2}{c}{$200$} & 
\multicolumn{2}{c}{$400$} & \multicolumn{2}{|c}{$100$} & \multicolumn{2}{c}{$%
200$} & \multicolumn{2}{c}{$400$} & \multicolumn{2}{|c}{$100$} & 
\multicolumn{2}{c}{$200$} & \multicolumn{2}{c}{$400$} \\ \hline
&  & BCT & BIC & \multicolumn{1}{|c}{BCT} & BIC & \multicolumn{1}{|c}{BCT} & 
BIC & \multicolumn{1}{|c}{BCT} & BIC & BCT & BIC & BCT & BIC & 
\multicolumn{1}{|c}{BCT} & BIC & BCT & BIC & BCT & BIC \\ 
$M$ & $\eta $ &  &  &  &  &  &  &  &  &  &  &  &  &  &  &  &  &  &  \\ 
&  &  &  &  &  &  &  &  &  &  &  &  &  &  &  &  &  &  &  \\
\hline
\end{tabular}

}

\begin{tablenotes}
  
\tiny
      \item Each entry in the table contains the average estimated rank across replications; the standard deviation (in round brackets); and the percentage of times the rank is estimated incorrectly (in square brackets). 
   
\end{tablenotes} 
}
\end{table*} 

\newpage 

\begin{table*}[h] 
\caption{Comparison between BCT and BIC - negative autocorrelation, Part 1} %
\label{tab:Table3a}{\resizebox{\textwidth}{!}{%
\begin{tabular}{ll|cccccccccccccccccc}
$\theta =-0.5$ & $n$ & \multicolumn{6}{|c}{$3$} & \multicolumn{6}{c}{$4$} & 
\multicolumn{6}{c}{$5$} \\ 
& $T$ & \multicolumn{2}{|c}{$100$} & \multicolumn{2}{c}{$200$} & 
\multicolumn{2}{c}{$400$} & \multicolumn{2}{|c}{$100$} & \multicolumn{2}{c}{$%
200$} & \multicolumn{2}{c}{$400$} & \multicolumn{2}{|c}{$100$} & 
\multicolumn{2}{c}{$200$} & \multicolumn{2}{c}{$400$} \\ \hline
& \multicolumn{1}{c|}{} & BCT & BIC & \multicolumn{1}{|c}{BCT} & BIC & 
\multicolumn{1}{|c}{BCT} & BIC & \multicolumn{1}{|c}{BCT} & BIC & BCT & BIC
& BCT & BIC & \multicolumn{1}{|c}{BCT} & BIC & BCT & BIC & BCT & BIC \\ 
$M$ & \multicolumn{1}{c|}{$\eta $} &  &  &  &  &  &  &  &  &  &  &  &  &  & 
&  &  &  &  \\ 
& \multicolumn{1}{c|}{} &  &  &  &  &  &  &  &  &  &  &  &  &  &  &  &  &  & 
\\
\hline
\end{tabular}%

}

\begin{tablenotes}
  
\tiny
      \item Each entry in the table contains the average estimated rank across replications; the standard deviation (in round brackets); and the percentage of times the rank is estimated incorrectly (in square brackets). 
   
\end{tablenotes} 
}
\end{table*} 

\newpage 

\begin{table*}[h] 
\caption{Comparison between BCT and BIC - negative autocorrelation, Part 2} %
\label{tab:Table3b} {\resizebox{\textwidth}{!}{%
\begin{tabular}{ll|cccccccccccccccccc}
$\theta =-0.5$ & $n$ & \multicolumn{6}{|c}{$3$} & \multicolumn{6}{c}{$4$} & 
\multicolumn{6}{c}{$5$} \\ 
& $T$ & \multicolumn{2}{|c}{$100$} & \multicolumn{2}{c}{$200$} & 
\multicolumn{2}{c}{$400$} & \multicolumn{2}{|c}{$100$} & \multicolumn{2}{c}{$%
200$} & \multicolumn{2}{c}{$400$} & \multicolumn{2}{|c}{$100$} & 
\multicolumn{2}{c}{$200$} & \multicolumn{2}{c}{$400$} \\ \hline
& \multicolumn{1}{c|}{} & BCT & BIC & \multicolumn{1}{|c}{BCT} & BIC & 
\multicolumn{1}{|c}{BCT} & BIC & \multicolumn{1}{|c}{BCT} & BIC & BCT & BIC
& BCT & BIC & \multicolumn{1}{|c}{BCT} & BIC & BCT & BIC & BCT & BIC \\ 
$M$ & \multicolumn{1}{c|}{$\eta $} &  &  &  &  &  &  &  &  &  &  &  &  &  & 
&  &  &  &  \\ 
& \multicolumn{1}{c|}{} &  &  &  &  &  &  &  &  &  &  &  &  &  &  &  &  &  & 
\\
\hline
\end{tabular}%

}

\begin{tablenotes}
  
\tiny
      \item Each entry in the table contains the average estimated rank across replications; the standard deviation (in round brackets); and the percentage of times the rank is estimated incorrectly (in square brackets). 
   
\end{tablenotes} 
}
\end{table*} 

\newpage

\newpage

\begin{table*}[h] 
\caption{Comparison between BCT and BIC - large $N$, \textit{i.i.d.} case} %
\label{tab:Table4} {\resizebox{\textwidth}{!}{%

\begin{tabular}{ll|cccccccccccccccccc}
$\theta =0$ & $n$ & \multicolumn{6}{|c}{$10$} & \multicolumn{6}{c}{$15$} & 
\multicolumn{6}{c}{$20$} \\ 
& $T$ & \multicolumn{2}{|c}{$1000$} & \multicolumn{2}{c}{$1500$} & 
\multicolumn{2}{c}{$2000$} & \multicolumn{2}{|c}{$1000$} & 
\multicolumn{2}{c}{$1500$} & \multicolumn{2}{c}{$2000$} & 
\multicolumn{2}{|c}{$1000$} & \multicolumn{2}{c}{$1500$} & 
\multicolumn{2}{c}{$2000$} \\ \hline
& \multicolumn{1}{c|}{} & BCT & BIC & \multicolumn{1}{|c}{BCT} & BIC & 
\multicolumn{1}{|c}{BCT} & BIC & \multicolumn{1}{|c}{BCT} & BIC & BCT & BIC
& BCT & BIC & \multicolumn{1}{|c}{BCT} & BIC & BCT & BIC & BCT & BIC \\ 
$m$ & \multicolumn{1}{c|}{$\eta $} &  &  &  &  &  &  &  &  &  &  &  &  &  & 
&  &  &  &  \\ 
& \multicolumn{1}{c|}{} &  &  &  &  &  &  &  &  &  &  &  &  &  &  &  &  &  & 
\\
\hline
\end{tabular}%

}

\begin{tablenotes}
  
\tiny
      \item Each entry in the table contains the average estimated rank across replications; the standard deviation (in round brackets); and the percentage of times the rank is estimated incorrectly (in square brackets). 
   
\end{tablenotes} 
}
\end{table*} 

\newpage 
\begin{table*}[h] 
\caption{Comparison between BCT and BIC - large $N$, positive autocorrelation} %
\label{tab:Table5} {\resizebox{\textwidth}{!}{%

\begin{tabular}{ll|cccccccccccccccccc}
$\theta =0.5$ & $n$ & \multicolumn{6}{|c}{$10$} & \multicolumn{6}{c}{$15$} & 
\multicolumn{6}{c}{$20$} \\ 
& $T$ & \multicolumn{2}{|c}{$1000$} & \multicolumn{2}{c}{$1500$} & 
\multicolumn{2}{c}{$2000$} & \multicolumn{2}{|c}{$1000$} & 
\multicolumn{2}{c}{$1500$} & \multicolumn{2}{c}{$2000$} & 
\multicolumn{2}{|c}{$1000$} & \multicolumn{2}{c}{$1500$} & 
\multicolumn{2}{c}{$2000$} \\ \hline
& \multicolumn{1}{c|}{} & BCT & BIC & \multicolumn{1}{|c}{BCT} & BIC & 
\multicolumn{1}{|c}{BCT} & BIC & \multicolumn{1}{|c}{BCT} & BIC & BCT & BIC
& BCT & BIC & \multicolumn{1}{|c}{BCT} & BIC & BCT & BIC & BCT & BIC \\ 
$m$ & \multicolumn{1}{c|}{$\eta $} &  &  &  &  &  &  &  &  &  &  &  &  &  & 
&  &  &  &  \\ 
& \multicolumn{1}{c|}{} &  &  &  &  &  &  &  &  &  &  &  &  &  &  &  &  &  & 
\\
\hline
\end{tabular}%

}

\begin{tablenotes}
  
\tiny
      \item Each entry in the table contains the average estimated rank across replications; the standard deviation (in round brackets); and the percentage of times the rank is estimated incorrectly (in square brackets). 
   
\end{tablenotes} 
}
\end{table*} 

\newpage 

\begin{table*}[h] 
\caption{Comparison between BCT and BIC - large $N$, negative autocorrelation} %
\label{tab:Table6} {\resizebox{\textwidth}{!}{%
\begin{tabular}{ll|cccccccccccccccccc}
$\theta =-0.5$ & $n$ & \multicolumn{6}{|c}{$10$} & \multicolumn{6}{c}{$15$}
& \multicolumn{6}{c}{$20$} \\ 
& $T$ & \multicolumn{2}{|c}{$1000$} & \multicolumn{2}{c}{$1500$} & 
\multicolumn{2}{c}{$2000$} & \multicolumn{2}{|c}{$1000$} & 
\multicolumn{2}{c}{$1500$} & \multicolumn{2}{c}{$2000$} & 
\multicolumn{2}{|c}{$1000$} & \multicolumn{2}{c}{$1500$} & 
\multicolumn{2}{c}{$2000$} \\ \hline
& \multicolumn{1}{c|}{} & BCT & BIC & \multicolumn{1}{|c}{BCT} & BIC & 
\multicolumn{1}{|c}{BCT} & BIC & \multicolumn{1}{|c}{BCT} & BIC & BCT & BIC
& BCT & BIC & \multicolumn{1}{|c}{BCT} & BIC & BCT & BIC & BCT & BIC \\ 
$m$ & \multicolumn{1}{c|}{$\eta $} &  &  &  &  &  &  &  &  &  &  &  &  &  & 
&  &  &  &  \\ 
& \multicolumn{1}{c|}{} &  &  &  &  &  &  &  &  &  &  &  &  &  &  &  &  &  & 
\\
\hline
\end{tabular}%

}

\begin{tablenotes}
  
\tiny
      \item Each entry in the table contains the average estimated rank across replications; the standard deviation (in round brackets); and the percentage of times the rank is estimated incorrectly (in square brackets). 
   
\end{tablenotes} 
}
\end{table*} 
\end{landscape}

\end{document}